\documentclass[fleqn,11pt,twoside]{article}
\usepackage[headings]{espcrc1}
\usepackage{amsmath}
\usepackage[strict]{changepage}
\usepackage{amssymb} 
\usepackage{bbm}
\usepackage{proof}
\usepackage{cmll}
\usepackage{mathbbol}
\usepackage{mathrsfs}
\usepackage{upgreek}
\usepackage{enumerate}
\usepackage[all]{xy} 

\newtheorem{conjecture}{Conjecture}

\newtheorem{theorem}{Theorem}[section]
\newtheorem{notation}[theorem]{Notation}
\newtheorem{corollary}[theorem]{Corollary}
\newtheorem{proposition}[theorem]{Proposition}
\newtheorem{lemma}[theorem]{Lemma}
\newtheorem{convention}[theorem]{Convention}
\newtheorem{definition}[theorem]{Definition}
\newtheorem{remark}[theorem]{Remark}
\newtheorem{claim}[theorem]{Claim}

\newtheorem{example}[theorem]{Example}
\newenvironment{proof}[1][Proof]{\noindent{\textbf{#1.}} }{\ \rule{0.5em}{0.5em}\vspace*{+2mm}}
\newenvironment{subproof}[1][Sub-proof]{\noindent{\textbf{#1.}} }{}

\title{What is a Categorical Model of the Differential and the Resource $\lambda$-Calculi?}

\author{Giulio Manzonetto\thanks{
	This work is funded by the NWO Project 612.000.936 CALMOC (CAtegorical and ALgebraic Models of Computation),
	and partly funded by Digiteo/\^Ile-de-France Project 2009-28HD COLLODI 
	(Complexity and concurrency through ludics and differential linear logic)
	and MIUR Project CONCERTO (CONtrol and CERTification of Resources Usage).
	}\address[NL]{Department of Computer Science,\\
        Radboud University,\\ 
        Nijmegen, The Netherlands\\
        {\em Email:} G.Manzonetto@cs.ru.nl}\\
}
       
\runtitle{What is a categorical model of the differential and the resource $\lambda$-calculi?}
\runauthor{Giulio Manzonetto}

\begin{document}
\newcommand{\fixme}[2][]{\todo[color=orange!60,#1]{#2}{}} 
\newcommand*{\autosubref}[2]{\hyperref[#1#2]{\autoref*{#1}(\ref*{#1#2})}}
\providecommand*{\autopageref}[1]{\hyperref[#1]{page~\pageref*{#1}}}
\newcommand*{\Autopageref}[1]{\hyperref[#1]{Page~\pageref*{#1}}}
\newcommand*{\ssubref}[2]{\hyperref[#1#2]{\ref*{#1}(\ref*{#1#2})}}
\renewcommand{\eqref}[1]{\hyperref[#1]{(\ref*{#1})}}
\newcommand{\finsubset}{\subset_{\mathrm{f}}}
\newcommand{\Perm}{\mathfrak{S}}
\newcommand{\dcup}{\dot{\cup}}
\newcommand{\supp}[1]{\mathrm{supp}(#1)}
\newcommand{\Fin}{{\sf F}}
\newcommand{\MFin}{\bold{MFin}}
\newcommand{\R}[1]{R\langle #1 \rangle}
\renewcommand{\bold}[1]{{\bf #1}}
\newcommand\Wn{\mathord?}  
\newcommand\IPar{\mathrel\wp}
\newcommand{\const}[1]{\textnormal{\scshape#1}}
\newcommand{\MELL}{\mbox{\sf MELL}}
\newcommand{\Der}[2]{{\sf D\hspace{1pt}}#1\cdot #2}
\newcommand{\Dern}[3]{{\sf D}^{#1}\hspace{1pt} #2\cdot(#3)}
\newcommand{\sw}{\mathrm{sw}}
\newcommand{\subst}[3]{#1\{#3/#2\}}
\newcommand{\dsubst}[3]{\textstyle{\frac{\partial #1}{\partial #2}}\cdot #3}
\newcommand{\Dsubst}[3]{\textstyle{\frac{\partial}{\partial #2}}#1\cdot #3}
\newcommand{\dsubstn}[4]{\textstyle{\frac{\partial^{#1} #2}{\partial #3}}\cdot(#4)}
\newcommand{\rsubst}[3]{#1\langle\!\langle #2 : = #3 \rangle\!\rangle}

\newcommand{\FV}{\mathrm{FV}}
\newcommand{\bto}{\to^b}
\newcommand{\gto}{\to^g}
\newcommand{\sA}{\mathbb{A}}
\newcommand{\sB}{\mathbb{B}}
\newcommand{\sC}{\mathbb{C}}
\newcommand{\sM}{\mathbb{M}}
\newcommand{\sN}{\mathbb{N}}
\newcommand{\sP}{\mathbb{P}}
\newcommand{\sQ}{\mathbb{Q}}
\newcommand{\sums}[1]{\nat\langle \Lambda^{#1}\rangle}
\newcommand{\lsubst}[3]{#1\langle #3/#2 \rangle}

\newcommand{\catC}{\bold{C}}
\newcommand{\cat}[1]{\bold{#1}}
\newcommand{\Rel}{\bold{Rel}}
\newcommand{\MRel}{\bold{M\hspace{-1.5pt}Rel}}
\newcommand{\Termobj}{\mathbbm{1}}
\newcommand{\eval}{\mathrm{ev}}
\newcommand{\curry}{\Uplambda}
\newcommand{\pairfun}[2]{\langle#1,#2\rangle}
\newcommand{\Id}[1]{\mathrm{Id}_{#1}}
\newcommand{\bang}{\oc}
\newcommand{\retract}{\vartriangleleft}

\newcommand{\D}{\mathcal{D}}
\newcommand{\leD}{\le_\mathcal{D}}
\newcommand{\Abs}{\lambda}
\newcommand{\App}{\mathcal{A}}
\newcommand{\Cint}[1]{\mathbb{\Lbrack} #1\mathbb{\Rbrack}}
\newcommand{\Par}[2]{{#1}\IPar{#2}}
\newcommand{\ITens}{\otimes}
\newcommand{\Tens}[2]{{#1}\ITens{#2}}
\newcommand{\With}[2]{{#1}\hspace{-2pt}\mathrel{\&}\hspace{-3pt}{#2}}
\newcommand{\IPlus}{\oplus}
\newcommand{\Plus}[2]{{#1}\IPlus{#2}}
\newcommand{\ILfun}{\multimap}
\newcommand{\Lfun}[2]{{#1}\ILfun{#2}}
\newcommand{\Orth}[1]{#1^\Bot}
\newcommand{\Proj}[1]{\pi_{#1}}
\newcommand{\Pair}[2]{\langle{#1},{#2}\rangle}
\newcommand{\Mpair}[2]{({#1},{#2})}
\newcommand{\Funint}[2]{{#1}\hspace{-2pt}\Rightarrow\hspace{-2pt}{#2}}
\newcommand{\Absint}[1]{curry({#1})}
\newcommand{\comp}{\hspace{-1pt}\circ\hspace{-1pt}}
\newcommand{\Th}{\mathrm{Th}}

\newcommand{\Comega}{\underline{\Omega}}
\newcommand\Omegaunit{*}
\newcommand{\mcup}{\uplus}
\newcommand{\hv}{\mathrm{hv}}
\newcommand{\mcups}{\bar{\uplus}}
\newcommand\emptymset{[]}
\newcommand{\Mfin}[1]{\mathcal{M}_{f}( #1 )}
\newcommand{\Mset}[1]{[{#1}]}
\newcommand\Omegatuple[1]{\Mfin{#1}^{(\omega)}}
\newcommand{\Multi}[1]{\mathcal{M}( #1 )}
\newcommand{\cons}{::}

\newcommand{\ssk}{\bold{k}}
\newcommand{\ssi}{\bold{i}}
\newcommand{\sso}{\boldsymbol{\varepsilon}}
\newcommand{\sss}{\bold{s}}
\newcommand{\Succ}{\bold{s}}

\newcommand{\TE}{\mathcal{E}}
\newcommand{\NF}{\mathrm{NF}}
\newcommand{\diffinf}{\Lambda^{d}_\infty}
\newcommand{\st}{\ \vert\ }
\newcommand{\nat}{\mathcal{N}}
\newcommand{\Pow}[1]{\mathcal{P}(#1)}
\newcommand{\Var}{\mathrm{Var}}
\newcommand{\Int}[1]{\vert#1\vert}
\renewcommand{\Int}[1]{\mathbb{\Lbrack} #1\mathbb{\Rbrack}}
\newcommand{\fin}{\mathrm{f}}
\newcommand{\Ad}{\mathrm{Ad}}

\newcommand{\seq}[1]{\vec #1}
\newcommand{\len}[1]{\ell(#1)}
\newcommand{\at}[2]{#1 :: #2}
\newcommand\Lamappn[3]{({#2})_{#1}{#3}}
\newcommand\Lamnamed[2]{[#1]{#2}}
\renewcommand{\phi}{\varphi}
\newcommand{\msto}{\twoheadrightarrow}
\newcommand\Subst[3]{{#1}[{#3}/{#2}]}
\newcommand{\Church}[1]{\underline{#1}}

\renewcommand{\parallel}{\| }
\newcommand{\Lpp}{\Lambda_{+\parallel}}
\newcommand{\lpp}{\lambda_{+\parallel}}
\newcommand{\pp}{+\hspace{-5pt}\parallel}

\newcommand{\imp}{\Rightarrow}
\newcommand{\To}{\Rightarrow}
\newcommand{\cA}{\mathcal{A}}
\newcommand{\cB}{\mathcal{B}}
\newcommand{\cC}{\mathcal{C}}
\newcommand{\cD}{\mathscr{D}}
\newcommand{\cE}{\mathcal{E}}
\newcommand{\cF}{\mathcal{F}}
\newcommand{\cG}{\mathcal{G}}
\newcommand{\cH}{\mathcal{H}}
\newcommand{\cI}{\mathcal{I}}
\newcommand{\cJ}{\mathcal{J}}
\newcommand{\cK}{\mathcal{K}}
\newcommand{\cL}{\mathcal{L}}
\newcommand{\cM}{\mathscr{M}}
\newcommand{\cN}{\mathcal{N}}
\newcommand{\cO}{\mathcal{O}}
\newcommand{\cP}{\mathcal{P}}
\newcommand{\cQ}{\mathcal{Q}}
\newcommand{\cR}{\mathcal{R}}
\newcommand{\cS}{\mathcal{S}}
\newcommand{\cT}{\mathcal{T}}
\newcommand{\cU}{\mathscr{U}}
\newcommand{\cV}{\mathcal{V}}
\newcommand{\cW}{\mathcal{W}}
\newcommand{\cX}{\mathcal{X}}
\newcommand{\cY}{\mathcal{Y}}
\newcommand{\cZ}{\mathcal{Z}}

\maketitle

\begin{abstract}{\bf Abstract.} 
The \emph{differential $\lambda$-calculus} is a paradigmatic functional programming language endowed 
with a syntactical differentiation operator that allows to apply a program to an argument in a linear way.
One of the main features of this language is that it is \emph{resource conscious} and gives the programmer 
suitable primitives to handle explicitly the resources used by a program during its execution.
The differential operator also allows to write the full Taylor expansion of a program.
Through this expansion every program can be decomposed into an infinite sum (representing non-deterministic choice) of `simpler' programs
that are strictly linear.

The aim of this paper is to develop an abstract `model theory' for the untyped differential $\lambda$-calculus.
In particular, we investigate what should be a general categorical definition of denotational model for this calculus.
Starting from the work of Blute, Cockett and Seely on differential categories we provide the notion of \emph{Cartesian closed
differential category} and we prove that \emph{linear reflexive objects} living in such categories constitute sound models 
of the untyped differential $\lambda$-calculus. 
We also give sufficient conditions for Cartesian closed differential categories to model the Taylor expansion.
This entails that every model living in such categories equates all programs having the same full Taylor expansion.

We then provide a concrete example of a Cartesian closed differential category modeling the Taylor expansion, 
namely the category $\MRel$ of sets and relations from finite multisets to sets. 
We prove that the relational model $\cD$ of $\lambda$-calculus we have recently built in $\MRel$ is linear, 
and therefore it is also a model of the untyped differential $\lambda$-calculus.

Finally, we study the relationship between the differential $\lambda$-calculus and the \emph{resource calculus}, a functional
programming language combining the ideas behind the differential $\lambda$-calculus with those behind the $\lambda$-calculus 
with multiplicities.
We define two translation maps between these two calculi and we study the properties of these translations.
In particular, from this analysis it follows that the two calculi share the same notion of model.
Therefore the resource calculus can be interpreted by translation into every linear reflexive object living in a 
Cartesian closed differential category.
\end{abstract}

\medskip

\newpage
\tableofcontents
\newpage
\section*{Introduction}

Among the variety of computational formalisms that have been studied in the literature, the $\lambda$-calculus \cite{Bare} plays 
an important role as a bridge between logic and computer science. 
The $\lambda$-calculus was originally introduced by Church \cite{Church32,Church41} as a foundation for mathematics, 
where functions -- instead of sets -- were primitive.
This system turned out to be consistent and successful as a tool for formalizing all computable functions.
However, the $\lambda$-calculus is not resource sensitive since a $\lambda$-term can erase its arguments or duplicate them 
an arbitrary large number of times. 
This becomes problematic when one wants to deal with programs that are executed in environments with bounded resources 
(like PDA's) or in presence of depletable arguments (like quantum data that cannot be duplicated for physical reasons).
In these contexts we want to be able to express the fact that a program \emph{actually consumes} its argument.
Such an idea of `resource consumption' is central in Girard's quantitative semantics \cite{Girard88}. 
This semantics establishes an analogy between linearity in the sense of computer science 
(programs using arguments exactly once) and algebraic linearity (commutation of sums and products with scalars), 
giving a new mathematically very appealing interpretation of resource consumption. 
Drawing on these insights, 
Ehrhard and Regnier \cite{EhrhardR03} designed a resource sensitive paradigmatic programming language
called \emph{the differential $\lambda$-calculus}. 

\medskip
{\bf The differential lambda calculus} 
is a conservative (see \cite[Prop.~19]{EhrhardR03}) extension of the untyped $\lambda$-calculus
with differential and linear constructions. 
In this language, there are two different operators that can be used to apply a program to its argument: 
the usual application and a \emph{linear application}. 
This last one defines a syntactic derivative operator $\Der{s}{t}$ which is an excellent candidate to increase control over programs 
executed in environments with bounded resources. 
Indeed, the evaluation of $\Der st$ (the derivative of the program $s$ on the argument $t$) has a precise operational meaning:
it captures the fact that the argument $t$ is available for $s$ ``exactly once''.
The corresponding meta-operation of substitution, that replaces exactly one (linear) occurrence of $x$ in $s$ by $t$,
is called ``differential substitution'' and is denoted by $\dsubst{s}{x}{t}$.
It is worth noting that when $s$ contains several occurrences of $x$, one has to choose which occurrence should be replaced and 
there are several possible choices. 
When $s$ does not contain any occurrence of $x$ then the differential substitution cannot be performed and the result is 0 
(corresponding to an empty program).
Thus, the differential substitution forces the presence of non-determinism in the system, which is represented by a formal sum having 0 as 
neutral element. 
Therefore, the differential $\lambda$-calculus constitutes a useful framework for studying the notions of linearity and non-determinism, 
and the relation between them.

\medskip
{\bf Taylor expansion.} 
As expected, iterated differentiation yields a natural notion of linear approximation of the ordinary application of a program to its argument. 
Indeed, the syntactic derivative operator allows to write all the derivatives of a $\lambda$-term $M$, 
thus it also allows (in presence of countable sums) to define its \emph{full Taylor expansion} $M^*$. 
In general, $M^*$ will be an infinite formal linear combination of simple terms (with coefficients in a field), and should satisfy, when
$M$ is a usual application $NQ$:
$$
	(NQ)^* = \sum_{n=0}^{\infty}\frac{1}{n!}(\Dern{n}{N}{\underbrace{Q,\ldots,Q}_{n\textrm{ times}}})0
$$
where $\frac{1}{n!}$ is a numerical coefficient and $\Dern{n}{N}{Q,\ldots,Q}$ stands for iterated linear application of $N$ to $n$ copies of $Q$.
The precise operational meaning of the Taylor expansion has been extensively studied in \cite{EhrhardR03,EhrhardR06a,EhrhardR08}.
The crucial fact of such an expansion is that it gives a \emph{quantitative} account to the $\beta$-reduction of $\lambda$-calculus 
(in the sense of B\"ohm tree computation).
Formal connections between Taylor expansions and B\"ohm trees of usual $\lambda$-terms have been presented in \cite{EhrhardR06a},
using a decorated version of Krivine's machine.

\medskip
{\bf The resource calculus}, which is a revisitation of Boudol's $\lambda$-calculus with multiplicities \cite{Boudol93,BoudolCL99}, 
shows an alternative approach to the problem of modeling resource consumption within a functional programming language. 
In this calculus there is only one operator of application, while the arguments can be either linear or reusable
and come in finite multisets called `bags'.
Linear arguments must to be used exactly once, while reusable ones can be used \emph{ad libitum}. 
Also in this setting the evaluation of a function applied to a bag of arguments may give rise to different possible choices, 
corresponding to the different possibilities of distributing the arguments between the occurrences of the formal parameter. 

The main differences between Boudol's calculus and the resource calculus are that the former is affine, is equipped with explicit 
substitution and has a lazy operational semantics, while the latter is linear and is a true extension of the classical $\lambda$-calculus.
The current formalization of resource calculus has been proposed by Tranquilli in \cite{TranquilliTh} with the aim of 
defining a Curry-Howard correspondence with differential nets \cite{EhrhardR06b}.

The resource calculus has been recently studied from a syntactical point of view by Pagani and Tranquilli \cite{PaganiT09}
for confluence results and by Pagani and Ronchi della Rocca \cite{PaganiR10} for results about may and must solvability.
Algebraic notions of models for the {\em strictly linear fragment} of resource calculus have been proposed by Carraro, 
Ehrhard and Salibra in \cite{CarraroES10b}.
In the present paper we mainly focus on the study of the differential $\lambda$-calculus,
but we will also draw conclusions for the resource calculus.

\medskip
{\bf Denotational semantics.} Although the differential $\lambda$-calculus is born from semantical considerations (i.e., 
the deep analysis of coherent spaces performed by Ehrhard and Regnier) the investigations on its 
denotational semantics are at the very beginning. 
It is known that finiteness spaces \cite{Ehrhard05} and the relational semantics of linear logic \cite{Girard88} 
are examples of models of the \emph{simply typed} differential $\lambda$-calculus, thus having a very limited 
expressive power. 
Concerning the \emph{untyped} differential $\lambda$-calculus, it is just known \emph{in the folklore} that the relational 
model $\cD$ introduced in \cite{BucciarelliEM07} in the category $\MRel$ constitutes a concrete example of model\footnote{
This follows from \cite{EhrhardR06b} where it is shown that the differential $\lambda$-calculus can be translated into
differential proofnets, plus \cite{VauxTh} where it is proved that $\cD$ is a model of such proofnets.
}. 
This picture is reminiscent of the beginning of denotational semantics of $\lambda$-calculus, when Scott's $\cD_\infty$ 
was the unique concrete example of model of $\lambda$-calculus and no general definition of model was known. 
Only when an abstract model theory for this calculus has been developed the researchers have been able to provide rich semantics 
(like the continuous \cite{Scott72}, stable \cite{Berry78} and strongly stable semantics \cite{BucciarelliE91}) and general methods for 
building huge classes of models in these semantics. 

\medskip
{\bf Categorical notion of model.} 
The aim of the present paper is to provide a general categorical notion of model of the untyped differential $\lambda$-calculus.
Our starting point will be the work of Blute, Cockett and Seely on (Cartesian) differential categories \cite{BluteCS06,BluteCS09}. 
In these categories a derivative operator $D(-)$ on morphisms is equationally axiomatized; 
the derivative of a morphism $f:A\to B$ will be a morphism $D(f): A\times A\to B$, linear in its first component.
The authors have then proved that these categories are sound and complete to model suitable term calculi.
However, it turns out that the properties of differential categories are too weak for modeling the full differential $\lambda$-calculus.
For this reason, we will introduce the more powerful notion of \emph{Cartesian closed differential category}. 
In such categories it is possible to define an operator 
$$
\infer[(\star)]{f\star g:C\times A\to B}{f:C\times A\to B & g:C\to A}
$$
that can be seen as a categorical counterpart of the differential substitution.
Intuitively, the morphism $f\star g$ is obtained by force-feeding the second argument $A$ of $f$ with \emph{one copy} of the result of $g$.
The type is not modified because $f\star g$ may still depend on $A$.

The operator $\star$ allows us to interpret the differential $\lambda$-calculus in every {\em linear reflexive object} $\cU$ 
living in a Cartesian closed differential category $\cat C$. 
We will prove that this categorical notion of model is \emph{sound};
this means that the induced equational theory $\Th(\cU)$ is actually a \emph{differential $\lambda$-theory}. 
The problem of equational completeness for this notion of model is left for future works,
and will be discussed in Section~\ref{sec:FurtherWorks}.

We will also investigate what conditions the category $\cat C$ must satisfy in order to \emph{model the Taylor expansion}.
This entails that all differential programs having the same Taylor expansion are equated in every model living in $\cat C$.

\medskip
{\bf Relational semantics.}
In \cite{BucciarelliEM07} we have built, in collaboration with Bucciarelli and Ehrhard, an extensional model $\cD$ of $\lambda$-calculus
living in the category $\MRel$ of sets and ``relations from finite multisets to sets''.
By virtue of its relational nature, $\cD$ can be used to model several systems, beyond the untyped $\lambda$-calculus.
For instance, in \cite{BucciarelliEM09} the authors have proved that it constitutes an adequate model of a $\lambda$-calculus 
extended with non-deterministic choice and parallel composition, while in \cite{VauxTh} Vaux has shown that it is a model of 
differential proof-nets.

In the present paper we study $\cD$ as a model of the untyped differential $\lambda$-calculus.
Indeed (as expected) the category $\MRel$ turns out to be an instance of the definition of Cartesian closed differential category, 
and the relational model $\cD$ is easily checked to be linear.
We will then study the equational theory induced by $\cD$ and prove that it equates all terms having the same Taylor expansion.
This property follows from the fact that $\MRel$ models the Taylor expansion.

\medskip
{\bf Translations.} Finally, we study the inter-relationships existing between the differential $\lambda$-calculus and the resource calculus.
Actually there is a common belief in the scientific community stating that the two calculi are morally the same, 
and the choice of studying one language or the other one is more a matter of taste than a substantial difference.
We will give a formal meaning to this belief by defining a translation map $(\cdot)^r$ from the differential $\lambda$-calculus to
the resource calculus, and another map $(\cdot)^d$ in the other direction.
We will prove that these translations are `faithful' in the sense that equivalent programs of differential $\lambda$-calculus 
are mapped into equivalent resource programs, and \emph{vice versa}.
This shows that the two calculi share the same notion of denotational model;
in particular the resource calculus can be interpreted by translation in every linear reflexive object living in a Cartesian
closed differential category.

\paragraph{Outline.}
Section~\ref{sec:Definition of MRel} contains the preliminary notions and notations
 needed in the rest of the paper.
In Section~\ref{sec:STDlc} we present the syntax and the axioms of the differential $\lambda$-calculus,
and we define the associated equational theories.
In Section~\ref{sec:aDiffMT} we introduce the notion of Cartesian closed differential category.
Section~\ref{sec:cat-models} is devoted to show that linear reflexive objects in such categories
are sound models of the differential $\lambda$-calculus.
In Section~\ref{sec:examples} we build a relational model $\cD$ and provide a partial characterization of its equational theory.
In Section~\ref{sec:ResCal} we define the resource calculus and we study its relationship with the 
differential $\lambda$-calculus.
Finally, in Section~\ref{sec:FurtherWorks} we present our conclusions and we propose some further lines of research.

\section{Preliminaries}\label{sec:Definition of MRel}

To keep this article self-contained we summarize some definitions and results that will be used in the sequel.
Our main reference for category theory is \cite{AspertiL91}.

\subsection{Sets and Multisets}\label{subs:msets}

We denote by $\nat$ the set of natural numbers.
Given $n\in\nat$ we write $\Perm_n$ for the set of all permutations (bijective maps) of the set $\{1,\ldots,n\}$.

Let $A$ be a set. 
We denote by $\Pow{A}$ the powerset of $A$. 
A \emph{multiset} $m$ over $A$ can be defined as an unordered list $m = [a_1,a_2,\ldots]$ with repetitions 
such that $a_i\in S$ for all indices $i$. 
A multiset $m$ is called \emph{finite} if it is a finite list; we denote by $[]$ the empty multiset.
Given two multisets $m_1 = \Mset{a_1,a_2,\ldots}$ and $m_2 = \Mset{b_1,b_2,\ldots}$ the \emph{multi-union} of $m_1,m_2$ 
is defined by $m_1\mcup m_2 = \Mset{a_1,b_1,a_2,b_2,\ldots}$.

Finally, we write $\Mfin{A}$ for the set of all finite multisets over $A$.

\subsection{Cartesian (Closed) Categories}

Let $\catC$ be a \emph{Cartesian category} and $A,B,C$ be arbitrary objects of $\catC$. 
We write $\catC(A,B)$ for the homset of morphisms from $A$ to $B$; when there is no chance of confusion we 
write $f:A\to B$ instead of $f\in\catC(A,B)$.
We usually denote by $A\times B$ the \emph{categorical product} of $A$ and $B$, by $\Proj1: A\times B\to A$, $\Proj2: A\times B\to B$ the associated 
\emph{projections} and, given a pair of arrows $f:C\to A$ and $g:C\to B$, by $\pairfun{f}{g}: C\to A\times B$ 
the unique arrow such that $\Proj1 \comp \pairfun{f}{g} =f$ and $\Proj2 \comp \pairfun{f}{g} = g$.
We write $f\times g$ for the \emph{product map of $f$ and $g$} which is defined by 
$f\times g = \Pair{f\comp\Proj 1}{g\comp\Proj 2}$.

If the category $\catC$ is \emph{Cartesian closed} we write $\Funint{A}{B}$ for the \emph{exponential object} 
and $\eval_{AB}:[\Funint AB]\times A\to B$ for the \emph{evaluation morphism}.
Moreover, for any object $C$ and arrow $f:C\times A\to B$, $\curry(f):C\to[\Funint AB]$ stands for the 
(unique) morphism such that $\eval_{AB}\comp(\curry(f)\times \Id{A}) = f$.
Finally, $\Termobj$ denotes the terminal object and $!_A$ the only morphism in $\catC(A,\Termobj)$.

We recall that in every Cartesian closed category the following equalities hold:
$$
\begin{array}{llrr}
\textrm{(pair)}\quad&\Pair{f}{g}\comp h = \Pair{f\comp h}{g\comp h}& \quad\curry(f)\comp g = \curry(f\comp(g\times \Id{})) &\quad\textrm{(Curry)}\\
\textrm{(beta-cat)}\quad&\eval\comp\Pair{\curry (f)}{g} = f\comp \Pair{\Id{}}{g}&\curry(\eval)=\Id{}&\quad\textrm{(Id-Curry)}\\
\end{array}
$$
Moreover, we can define the \emph{uncurry} operator $\curry^-(-) = \eval\comp(-\times\Id{})$. 
From (beta-cat), (Curry) and (Id-Curry) it follows that $\curry(\curry^-(f)) = f$ and 
$\curry^-(\curry(g)) = g$.

\section{The Differential Lambda Calculus}\label{sec:STDlc}

In this section we recall the definition of the {\em differential $\lambda$-calculus} \cite{EhrhardR03}, 
together with some standard properties of the language.
We also define the associated equational theories, namely, the {\em differential $\lambda$-theories}. 
The syntax we use in the present paper is freely inspired by \cite{Vaux07}.

\subsection{Differential Lambda Terms}

The set $\Lambda^d$ of \emph{differential $\lambda$-terms} and the set $\Lambda^s$ 
of \emph{simple terms} are defined by mutual induction as follows:
$$
\Lambda^d:\quad	S,T,U,V\ ::=\quad 0\ \vert\ s \ \vert\ s + T\qquad\qquad
\Lambda^s:\quad	s,t,u,v \ ::= \quad x \ \vert\ \lambda x.s\ \vert\ sT\ \vert\ \Der{s}{t}
$$
The differential $\lambda$-term $\Der{s}{t}$ represents the \emph{linear application} of $s$ to $t$. 
Intuitively, this means that $s$ is provided with exactly one copy of $t$. 
Notice that sums may appear also in simple terms as right components of ordinary applications.
Although the rule $s + t = s$ will not be valid in our axiomatization, the sum should still be thought 
as a version of non-deterministic choice where all actual choice operations are postponed.

\begin{convention}
We consider differential $\lambda$-terms up to $\alpha$-conversion, and up to associativity 
and commutativity of the sum.
The term 0 is the neutral element of the sum, thus we also add the equation $S + 0 = S$. 
\end{convention}

As a matter of notation we write $\lambda x_1\ldots x_n.s$ for $\lambda x_1.(\cdots (\lambda x_n.s)\cdots)$
and $sT_1\cdots T_k$ for $(\cdots (sT_1)\cdots )T_k$.
Moreover, we set $\Dern{1}{s}{t_1} = \Der{s}{t_1}$ and $\Dern{n+1}{s}{t,t_1,\ldots,t_{n}} = \Dern{n}{(\Der{s}{t})}{t_1,\ldots,t_{n}}$.
When writing $\Dern{n}{s}{t_1,\ldots,t_n}$ we suppose $n > 0$, unless differently stated.
\begin{definition} The {\em permutative equality} on differential $\lambda$-terms imposes that 
$\Dern{n}{s}{t_1,\ldots,t_n} = \Dern{n}{s}{t_{\sigma(1)},\ldots,t_{\sigma(n)}}$ for all 
permutations $\sigma\in\Perm_n$.
\end{definition}
Hereafter, we will consider differential $\lambda$-terms also up to the permutative equality. 
This is needed, for instance, for proving the Schwarz lemma
(see Subsection~\ref{ssec:Substs}) and hence to speak of a differential operator.
Concerning specific $\lambda$-terms we set:\label{specific_terms}
$$
	\begin{array}{c}
	\bold{I}\equiv \lambda x.x\qquad \Delta \equiv \lambda x.xx\qquad \Omega\equiv\Delta\Delta\qquad
	\bold{Y}\equiv \lambda f.(\lambda x.f(xx))(\lambda x.f(xx))\\
	\Succ\equiv\lambda nxy.nx(xy)\qquad\quad \Church{n}\equiv  \lambda sx. s^n(x),\textrm{ for every natural number }n\in\nat\\
	\end{array}
$$
where $\equiv$ stands for syntactical equality up to the above mentioned equivalences on differential $\lambda$-terms.
Note that $\bold{I}$ is the identity, $\bold{Y}$ is Curry's fixpoint combinator, $\Church{n}$ the $n$-th Church numeral and 
$\Succ$ implements the successor function.
$\Omega$ denotes the usual paradigmatic unsolvable $\lambda$-term.
\begin{definition}
Let $S$ be a differential $\lambda$-term.
The set $\FV(S)$ of \emph{free variables} of $S$ is defined inductively as follows: 
\begin{itemize}
\item $\FV(x) = \{x\}$, 
\item $\FV(\lambda x.s) = \FV(s)-\{x\}$, 
\item $\FV(sT) = \FV(s)\cup\FV(T)$, 
\item $\FV(\Der{s}{t}) = \FV(s)\cup\FV(t)$, 
\item $\FV(s + S) = \FV(s)\cup \FV(S)$, 
\item $\FV(0) = \emptyset$.
\end{itemize}
Given differential $\lambda$-terms $S_1,\ldots,S_k$ 
we set $\FV(S_1,\ldots,S_k) = \FV(S_1)\cup\cdots\cup\FV(S_k)$.
\end{definition}
We now introduce some notations on differential $\lambda$-terms that will be 
particularly useful to define the substitution operators in the next subsection.
\begin{notation}\label{not:sumsonDLT} 
We will often use the following abbreviations (notice that these are just syntactic 
sugar, not real terms): 
\begin{itemize}
\item $\lambda x.(\sum_{i = 1}^k s_i) = \sum_{i = 1}^k \lambda x.s_i$, 
\item $(\sum_{i = 1}^k s_i)T = \sum_{i = 1}^k s_i T$,
\item $\Der{(\sum_{i = 1}^k s_i)}{(\sum_{j = 1}^n t_j)} = \sum_{i,j} \Der{s_i}{t_j}$.
\end{itemize}
\end{notation}
Intuitively, these equalities make sense since the lambda abstraction is linear, 
the usual application is linear in its left component, and the linear application is 
a bilinear operator. 
Notice however that $S(\Sigma_{i=1}^k t_i) \neq \Sigma_{i=1}^k St_i$.

\subsection{Substitutions}\label{ssec:Substs}

We introduce two kinds of meta-operations of substitution on differential $\lambda$-terms:
the usual capture-free substitution and the differential substitution. 
Both definitions strongly use the abbreviations introduced in Notation~\ref{not:sumsonDLT}.

\begin{definition} 
Let $S,T$ be differential $\lambda$-terms and $x$ be a variable.
The \emph{capture-free substitution} of $T$ for $x$ in $S$, denoted by $\subst{S}{x}{T}$, 
is defined by induction on $S$ as follows: 
\begin{itemize}
\item $\subst{y}{x}{T} = \left\{ 
\begin{array}{ll} 
T\quad & \text{if } x = y,\\ 
y & \text{otherwise,}\\ 
\end{array}\right.$
\item $\subst{(\lambda y.s)}{x}{T} = \lambda y.\subst{s}{x}{T}$, where we suppose by 
$\alpha$-conversion that $x\neq y$ and $y\notin\FV(T)$,
\item $\subst{(sU)}{x}{T} = (\subst{s}{x}{T})(\subst{U}{x}{T})$,
\item $\subst{(\Dern{n}{s}{u_1,\ldots,u_n})}{x}{T} = \Dern{n}{(\subst{s}{x}{T})}{\subst{u_1}{x}{T},\ldots,\subst{u_n}{x}{T}}$,
\item $\subst{0}{x}{T} = 0$,
\item $\subst{(s + S)}{x}{T} = \subst{s}{x}{T} + \subst{S}{x}{T}$.
\end{itemize}
\end{definition}

Thus, $\subst{S}{x}{T}$ is the result of substituting $T$ for all free occurrences of $x$ in $S$,
subject to the usual proviso about renaming bound variables in $S$ to avoid capture of free variables in $T$.
On the other hand, the differential substitution $\dsubst{S}{x}{T}$ defined below denotes the result 
of substituting $T$ (still avoiding capture of variables) for \emph{exactly one} -- non-deterministically chosen -- 
occurrence of $x$ in $S$. 
If such an occurrence is not present in $S$ then the result will be $0$.

\begin{definition}
Let $S,T$ be differential $\lambda$-terms and $x$ be a variable.
The {\em differential substitution} of $T$ for $x$ in $S$, denoted by $\dsubst{S}{x}{T}$, is defined by induction on $S$ as follows:
\begin{itemize}
\item $\dsubst{y}{x}{T} = \left\{ 
\begin{array}{ll} 
T\quad & \text{if } x = y,\\ 
0 & \text{otherwise,}\\ 
\end{array}\right.$
\item $\Dsubst{(s U)}{x}{T} = (\dsubst{s}{x}{T})U+ (\Der{s}{(\dsubst{U}{x}{T})})U$,
\item $\Dsubst{(\lambda y.s)}{x}{T} = \lambda y.\dsubst{s}{x}{T}$, where we suppose by $\alpha$-conversion that $x\neq y$ and $y\notin\FV(T)$,
\item $\Dsubst{(\Dern{n}{s}{u_1,\ldots,u_n})}{x}{T} = \Dern{n}{(\dsubst{s}{x}{T})}{u_1,\ldots,u_n} + \sum_{i=1}^n \Dern{n}{s}{u_1,\ldots,\dsubst{u_i}{x}{T},\ldots,u_n}$,
\item $\dsubst{0}{x}{T} = 0$,
\item $\Dsubst{(s + U)}{x}{T} = \dsubst{s}{x}{T} + \dsubst{U}{x}{T}$.
\end{itemize}
\end{definition}

The definition states that the differential substitution distributes over linear constructions.
We now spend some words on the case of the usual application $sU$ because it is the 
most complex one.
The result of $\dsubst{(sU)}{x}{T}$ is the sum of two terms since the differential substitution 
can non-deterministically be applied either to $s$ or to $U$.
In the first case, we can safely apply it to $s$ since the usual application is linear in its left 
argument, so we obtain $(\dsubst{s}{x}{T})U$. 
In the other case we cannot apply it directly to $U$ because the standard application is \emph{not} 
linear in its right argument.
We thus follow two steps: (i) we replace $sU$ by $(\Der{s}{U})U$; 
(ii) we apply the differential substitution to the linear copy of $U$.

Intuitively, this works because $U$ is morally available infinitely many times in $sU$,
so when the differential substitution goes on $U$ we `extract' a linear copy of $U$,
that receives the substitution, and we keep the other infinitely many unchanged.
This will be much more evident in the definition of the analogous operation for the 
resource calculus (cf.\ Definition~\ref{def:linearsubstres}).

\begin{example} Recall that the simple terms $\Delta$ and $\bold{I}$ 
have been defined at page~\pageref{specific_terms}.
\begin{enumerate}[1.]
\item $\dsubst{\Delta}{x}{\bold{I}} = 0$, since $x$ does not occur free in $\Delta$,
\item $\dsubst{x}{x}{\bold{I}} = \bold{I}$,
\item $\dsubst{(xx)}{x}{\bold{I}} = \bold{I}x + (\Der{x}{\bold{I}})x$,
\item $\Dsubst{(\dsubst{(xx)}{x}{\bold{I}})}{x}{\Delta} =  
(\Der{\bold{I}}{\Delta})x + (\Der{\Delta}{\bold{I}})x + (\Der{(\Der{x}{\bold{I}})}{\Delta})x$,
\item $\subst{((\Der{x}{x})x)}{x}{\bold{I}} = (\Der{\bold{I}}{\bold{I}})\bold{I}$.
\end{enumerate}
\end{example}

The differential substitution $\dsubst{S}{x}{T}$ can be thought as the differential of $S$ 
with respect to the variable $x$, linearly applied to $T$. 
This may be inferred from the rule for linear application, which relates to the rule for composition
of the differential. 
Moreover, it is easy to check that if $x\notin \FV(S)$ (i.e., $S$ is constant with respect to
$x$) then $\dsubst{S}{x}{T} = 0$. 
This intuition is also reinforced by the validity of the Schwartz lemma.

\begin{lemma}\label{lemma:Schwartz} (Schwartz lemma) Let $S,T,U$ be differential 
$\lambda$-terms.
Let $x$ and $y$ be variables such that $x$ does not occur free in $U$. 
Then we have:
$$
	\Dsubst{\Big(\dsubst{S}{x}{T}\Big)}{y}{U} = 
	\Dsubst{\Big(\dsubst{S}{y}{U}\Big)}{x}{T} + \dsubst{S}{x}{\Big(\dsubst{T}{y}{U}\Big)}.
$$
In particular, when $y\notin\FV(T)$, then the second addend is 0 and the two 
differential substitutions just commute.
\end{lemma}

\begin{proof} 
The proof is by structural induction on $S$.
Here we just check the case $S \equiv vV$.
$$
\begin{array}{lcl}
\Dsubst{(\dsubst{vV}{x}{T})}{y}{U} & = & 
\Dsubst{((\dsubst{v}{x}{T}) V +
(\Der{v}{(\dsubst{V}{x}{T})})V)}{y}{U}\\
& = & (\Dsubst{(\dsubst{v}{x}{T})}{y}{U})V +
(\Der{(\dsubst{v}{x}{T})}{(\dsubst{V}{y}{U})})V\\
&&\quad + \
(\Der{(\dsubst{v}{y}{U})}{(\dsubst{V}{x}{T})})V + 
(\Der{v}{(\Dsubst{(\dsubst{V}{x}{T})}{y}{U})})V\\ 
&&\quad + \
(\Der{(\Der{v}{(\dsubst{V}{x}{T})}}{(\dsubst{V}{y}{U}))})V
\\
\end{array}
$$
By applying the induction hypothesis (and the permutative equality) we get:
$$
\begin{array}{lcl}
\Dsubst{(\dsubst{vV}{x}{T})}{y}{U} & = & 
(\Dsubst{(\dsubst{v}{y}{U})}{x}{T}
)V 
+
(\Der{(\dsubst{v}{y}{U})}{(\dsubst{V}{x}{T})})V
+ 
(\Der{(\dsubst{v}{x}{T})}{(\dsubst{V}{y}{U})})V\\
&& 
\quad +\ (\Der{v}{(\Dsubst{(\dsubst{V}{y}{U})}{x}{T})})V
+
(\Der{(\Der{v}{(\dsubst{V}{y}{U})})}{(\dsubst{V}{x}{T})})V
\\
&&\quad + \ (\dsubst{v}{x}{(\dsubst{T}{y}{U})})V
+ (\Der{v}{(\dsubst{V}{x}{(\dsubst{T}{y}{U})})})V\\
&=& \Dsubst{(
(\dsubst{v}{y}{U})V + 
(\Der{v}{(\dsubst{V}{y}{U})})V
)}{x}{T}\\
&&\quad + \ (\dsubst{v}{x}{(\dsubst{T}{y}{U})})V
+ (\Der{v}{(\dsubst{V}{x}{(\dsubst{T}{y}{U})})})V\\
&=& \Dsubst{(\dsubst{vV}{y}{U})}{x}{T} +
\dsubst{vV}{x}{(\dsubst{T}{y}{U})}.\\
\end{array}
$$
\end{proof}

For the sake of readability, it will be sometimes useful to adopt the following notation 
for multiple differential substitutions.

\begin{notation} We set
$$
	\dsubstn{n}{S}{x_1,\ldots,x_n}{t_1,\ldots,t_n} = \Dsubst{\Big(\cdots\dsubst{S}{x_1}{t_1}\cdots\Big)}{x_n}{t_n}
$$
where $x_i\notin\FV(t_1,\ldots,t_n)$ for all $1 \le i \le n$.
\end{notation}

\begin{remark} From Lemma~\ref{lemma:Schwartz} we have:
$$
\dsubstn{n}{S}{x_1,\ldots,x_n}{t_1,\ldots,t_n} = 
\dsubstn{n}{S}{x_{\sigma(1)},\ldots,x_{\sigma(n)}}{t_{\sigma(1)},\ldots,t_{\sigma(n)}},\textrm{  for all }\sigma\in\Perm_n.
$$
\end{remark}

\subsection{Differential Lambda Theories}

In this subsection we introduce the axioms associated with the differential $\lambda$-calculus and we define the equational theories of this calculus, namely, the \emph{differential $\lambda$-theories}.

The axioms of the \emph{differential $\lambda$-calculus} are the following (for all $s,t\in\Lambda^s$ and $T\in\Lambda^d$):
$$
	(\beta)\quad (\lambda x.s)T = \subst{s}{x}{T}
$$ 
$$
	(\beta_D)\quad \Der{(\lambda x.s)}{t} = \lambda x.\dsubst{s}{x}{t}.
$$
Once oriented from left to right, the $(\beta)$-conversion expresses 
the way of calculating a function $\lambda x.s$ classically applied to an argument 
$T$, while the $(\beta_D)$-conversion the way of evaluating a function $\lambda x.s$ 
\emph{linearly} applied to a simple argument $t$.

Notice that in the result of a linear application the $\lambda x$ does not disappear.
This is needed since the simple term $s$ may still contain free occurrences of $x$.
The only way to get rid of the outer lambda abstraction in the term $\lambda x.s$ is to apply it 
classically to a term $T$, and then use the $(\beta)$-rule;
when $x\notin\FV(s)$ a standard choice for $T$ is $0$.

The differential $\lambda$-calculus is an intensional language --- there are 
syntactically different programs having the same extensional behaviour. 
We will be sometimes interested in the extensional version of this calculus
which is obtained by adding the following axiom (for every $s\in\Lambda^s$):
$$
	(\eta)\quad \lambda x.sx = s,\textrm{ where }x\notin\FV(s)
$$
A \emph{$\lambda^d$-relation} $\cT$ is any set of equations between differential $\lambda$-terms 
(which can be thought as a relation on $\Lambda^d\times\Lambda^d$). 

A $\lambda^d$-relation $\cT$ is called:
\begin{itemize}
\item an \emph{equivalence} if it is closed under the following rules 
(for all  $S,T,U\in\Lambda^d$):
$$
\infer[\textrm{reflexivity}]{S = S}{}\qquad\quad
\infer[\textrm{symmetry}]{S = T}{T = S}\qquad\quad
\infer[\textrm{transitivity}]{S = U}{S = T& T = U}
$$
\item \emph{compatible} if it is closed under the following rules (for all $S,T,U, S_i\in\Lambda^d$ and $s,t,u,s_i\in\Lambda^s$):
$$
\begin{array}{c}
\infer[\textrm{lambda}]{\lambda x.s = \lambda x.S}{s = S}\qquad\quad
\infer[\textrm{app}]{sT = SU}{s = S& T = U}\qquad\quad
\infer[\textrm{Lapp}]{\Der{s}{u} = \Der{S}{U}}{s = S & u = U}\\
~\\
\infer[\textrm{sum}]{\sum_{i = 1}^{n} s_i = \sum_{i = 1}^{n} S_i}{s_i = S_i& \textrm{ for all }1\le i\le n}\\
\end{array}
$$
\end{itemize}
As a matter of notation, we will write $\cT\vdash S= T$ or $S =_{\cT} T$ for $S= T\in\cT$.

\begin{definition} A \emph{differential $\lambda$-theory} is any compatible $\lambda^d$-relation $\cT$ 
which is an equivalence relation and includes $(\beta)$ and $(\beta_D)$. 
$\cT$ is called \emph{extensional} if it also contains $(\eta)$.
\end{definition}

The differential $\lambda$-theories are naturally ordered by set-theoretical inclusion.
We denote by $\lambda\beta^d$ (resp.\ $\lambda\beta\eta^d$) the minimum differential $\lambda$-theory 
(resp.\ the minimum extensional differential $\lambda$-theory).

We present here some easy examples of equalities between differential $\lambda$-terms 
in $\lambda\beta^d$ (and $\lambda\beta\eta^d$)
in order to help the reader to get familiar with the operations in the calculus.

\begin{example} Recall that $\Delta \equiv \lambda x.xx$. Then we have:
\begin{enumerate}[1.] 
\item $\lambda\beta^d\vdash(\Der{\Delta}{y})z = yz + (\Der{z}{y})z$,
\item $\lambda\beta^d\vdash(\Dern{2}{\Delta}{x,y})0 =  (\Der{x}{y})0 + (\Der{y}{x})0$,
\item $\lambda\beta^d\vdash \Dern{3}{\Delta}{x,y,z} = \lambda r.(\Dern{2}{x}{y,z} + \Dern{2}{y}{x,z} + \Dern{2}{z}{x, y} + \Dern{3}{r}{x, y, z})r$,
\item $\lambda\beta\eta^d\vdash \Der{\Delta}{z} = \lambda x.zx + \lambda x.(\Der{x}{z})x = z + \lambda x.(\Der{x}{z})x $.
\end{enumerate}
\end{example}

Note that in this calculus (as in the usual $\lambda$-calculus extended 
with non-deterministic choice \cite{DezanidP96}) 
a single simple term can generate an 
infinite sum of terms, like in the example below.

\begin{example} Recall (from page~\pageref{specific_terms}) that $\bold{Y}$ is Curry's fixpoint combinator, 
$\Church{n}$ is the $n$-th Church numeral and $\Succ$ denotes the successor.
\begin{enumerate}[1.]
\item $\lambda\beta^d\vdash \bold{Y}(x + y) = x(\bold{Y}(x + y)) + y(\bold{Y}(x + y))$ for all variables $x, y$,
\item $\lambda\beta^d\vdash \bold{Y}((\lambda z.\Church{0}) + \Succ) = \Church{0}\ +\ \Succ(\bold{Y}((\lambda z.\Church{0}) + \Succ)) =
	\Church{0}\ +\ \Church{1}\ +\ \Succ(\Succ(\bold{Y}((\lambda z.\Church{0}) + \Succ))) = \cdots$
\end{enumerate}
\end{example}

\subsection{A Theory of Taylor Expansion}\label{subsec:TE}

One of the most interesting consequences of adding a syntactical differential operator to the $\lambda$-calculus
is that, in presence of infinite sums, this allows to define the Taylor expansion of a program. 
Such an expansion is classically defined in the literature only for ordinary $\lambda$-terms \cite{EhrhardR03,EhrhardR06a,EhrhardR08}.
In this subsection we generalize this notion to general differential $\lambda$-terms.
To avoid the annoying problem of handling coefficients we consider an idempotent sum.

\begin{definition} Given a differential $\lambda$-term $S$ we define its \emph{(full) Taylor expansion} 
$S^*$ by induction on $S$ as follows:
\begin{itemize}
\item $x^* = x$,
\item $(\lambda x.s)^* = \lambda x.s^*$,
\item $(\Dern{k}{s}{t_1,\ldots,t_k})^* = \Dern{k}{s^*}{t_1^*,\ldots,t_k^*}$,
\item $(sT)^* = \Sigma_{k\in\nat} (\Dern{k}{s^*}{T^*,\ldots,T^*})0$, where $\nat$ denotes the set of natural numbers,
\item $(s+T)^* = s^*+T^*$.
\end{itemize}
\end{definition}

Thus, the ``target language'' of the Taylor expansion is much simpler than the full differential $\lambda$-calculus.
For instance, the general application of the $\lambda$-calculus is not needed anymore, we will only need 
iterated linear applications and ordinary applications to $0$. We will however need countable sums, that are not present 
in general in the differential $\lambda$-calculus.
Hereafter, the target calculus of the Taylor expansion will be denoted by $\diffinf$.

We will write $\seq S$ to denote sequences of differential $\lambda$-terms $S_1,\ldots,S_k$ (with $k\ge 0$).

\begin{remark}\label{rem:structure} 
Every term $S\in\diffinf$ can be written as a (possibly infinite) sum of terms of shape:
$$
\lambda \seq y.(\Dern{n_1}{(\cdots(\Dern{n_k}{s}{\seq t_k})\seq 0)\cdots}{\seq t_1})\seq 0
$$
where $\seq t_i$ is a sequence of simple terms of length $n_i\in\nat$ (for $1\le i\le k$)
and the simple term $s$ is either a variable or a lambda abstraction.
\end{remark}

We now try to clarify what does it mean that two differential $\lambda$-terms $S$ and $T$ ``have 
the same Taylor expansion''. 
Indeed we may have that $S^{*} = \Sigma_{i\in I}s_i$ and $T^* = \Sigma_{j\in J}t_j$ where $I,J$ are countable sets.
In this case one could be tempted to define $S^* = T^*$ by asking for the existence of a bijective correspondence 
between $I$ and $J$ such that each $s_i$ is $\lambda\beta^d$-equivalent to some $t_j$.
However, in the general case, this definition does not capture the equivalence between infinite sums that we have in mind.
For instance, $S^* = T^*$ might hold because there are partitions $\{I_k\}_{k\in K}$
and $\{J_k\}_{k\in K}$ of $I$ and $J$, respectively, such that for every $k\in K$ the sets $I_k,J_k$ are finite and
$\Sigma_{i\in I_k} s_i =_{\lambda\beta^d} \Sigma_{i\in J_k} s_j$.
The na\"if definition works well when all addenda of the two sums we are equating are `in normal form'.
Since the $\diffinf$ calculus (morally) enjoys strongly normalization, we can define the normal form of 
every $S\in\diffinf$ as follows.

\begin{definition} Given $S\in\diffinf$, we define the \emph{normal form of $S$} as follows.
\begin{itemize}
\item
	If $S \equiv \sum_{i \in I} s_i$ we set $\NF(S) = \sum_{i \in I} \NF(s_i)$. 
\item
	If $S \equiv \lambda \seq y.(\Dern{n_1}{(\cdots(\Dern{n_k}{x}{\seq t_k})\seq 0)\cdots}{\seq t_1})\seq 0$  then:
	$$
		\NF(S) = \lambda \seq y.(\Dern{n_1}{(\cdots(\Dern{n_k}{x}{\NF(\seq t_k}))\seq 0)\cdots}{\NF(\seq t_1}))\seq 0.
	$$
\item
	If $S \equiv \lambda \seq y.(\Dern{n_1}{(\cdots(\Dern{n_k}{(\lambda x.s)}{\seq t_k})\seq 0)\cdots}{\seq t_1})\seq 0$ with $n_k> 0$ then:
	$$
	\NF(S) = \NF\big(\lambda \seq y.(\Dern{n_1}{(\cdots(\Dern{n_{k-1}}{(
	(\lambda x.\dsubstn{n_k}{s}{x,\ldots,x}{\seq t_{k}})\seq 0
	)}{\seq t_{k-1}})\seq 0)\cdots}{\seq t_1})\seq 0\big).
	$$
\item
	If $S \equiv \lambda \seq y.(\Dern{n_1}{(\cdots(\Dern{n_k}{((\lambda x.s)0\seq 0)}{\seq t_k})\seq 0)\cdots}{\seq t_1})\seq 0$  then:
	$$
	\NF(S) = \NF\big(\lambda \seq y.(\Dern{n_1}{(\cdots(\Dern{n_{k}}{((\subst{s}{x}{0})\seq 0)}{\seq t_{k}})\seq 0)\cdots}{\seq t_1})\seq 0\big).
	$$
\end{itemize}
\end{definition}

By Remark~\ref{rem:structure} the definition above covers all possible cases.

We are now able to define the differential $\lambda$-theory generated by equating all differential $\lambda$-terms 
having the same Taylor expansion.
\begin{definition}
Given $S,T\in\Lambda^d$ we say that $\NF(S^*) = \NF(T^*)$ whenever $\NF(S^*) = \sum_{i\in I} s_i$, $\NF(T^*) = \sum_{j\in J} t_j$
and there is an isomorphism $\iota : I\to J$ such that $\lambda\beta^d\vdash s_i = t_{\iota(i)}$.
We set
$$
	\TE = \{ (S, T)\in\Lambda^d\times\Lambda^d \st \NF(S^*) = \NF(T^*) \}.
$$
It is not difficult to check that $\TE$ is actually a differential $\lambda$-theory.
\end{definition}

Two usual $\lambda$-terms $s,t$ have the same B\"ohm tree \cite[Ch.~10]{Bare} if, and only if, $\TE\vdash s = t$ holds.
The `if' part of this equivalence is fairly straightforward, whereas the `only if' part is proved in \cite{EhrhardR06a}.
Thus, the theory $\TE$ can be seen as an extension of the theory of B\"ohm trees in the context of differential
$\lambda$-calculus.

\section{A Differential Model Theory}\label{sec:aDiffMT}

In this section we will provide the categorical framework in which the models of the differential $\lambda$-calculus live,
namely, the \emph{Cartesian closed differential categories}\footnote{
These categories have been first introduced in \cite{BucciarelliEM10} 
(where they were called \emph{differential $\lambda$-categories})
and proposed as models of the simply typed differential $\lambda$-calculus 
and simply typed resource calculus.
}. 
The material presented in Subsection~\ref{subsec:CDC} is mainly borrowed from \cite{BluteCS09}. 

\subsection{Cartesian Differential Categories}\label{subsec:CDC}

Differential $\lambda$-terms will be interpreted as morphisms in a suitable category $\cat{C}$.
Since in the syntax we have sums of terms, we need a sum on the morphisms of $\cat C$ satisfying the
equations introduced in Notation~\ref{not:sumsonDLT}. For this reason, we will focus our attention on 
left-additive categories.

A category $\catC$ is \emph{left-additive} whenever each homset has a structure of commutative monoid 
$(\catC(A,B),+_{AB},0_{AB})$ and $(g+ h)\comp f = (g\comp f) + (h\comp f)$ and $0\comp f = 0$.

\begin{definition}
A morphism $f$ in $\catC$ is said to be \emph{additive} if, in addition, it satisfies 
$f\comp(g+ h) = (f\comp g) + (f\comp h)$ and $f\comp 0 = 0$.
\end{definition}

A category is \emph{Cartesian left-additive} if it is a left-additive category with products such that all 
projections and pairings of additive maps are additive.

\begin{definition}\label{def:cccLA} A \emph{Cartesian differential category} is a Cartesian left-additive 
category having an operator $D(-)$ that maps a morphism $f : A\to B$ into a morphism 
$D(f) : A\times A\to B$ and satisfies the following axioms:
\begin{enumerate}[D1.]
\item
	$D(f + g) = D(f) + D(g)$ and $D(0) = 0$,
\item
	$D(f)\comp\Pair{h+k}{v} = D(f)\comp\Pair{h}{v} + D(f)\comp\Pair{k}{v} $ and $D(f)\comp\Pair{0}{v} = 0$,
\item
	$D(\Id{}) = \Proj 1$, $D(\Proj 1) = \Proj 1\comp\Proj 1$ and $D(\Proj 2) = \Proj 2\comp \Proj 1$,
\item
	$D(\Pair{f}{g}) = \Pair{D(f)}{D(g)}$, 
\item
	$D(f\comp g) = D(f)\comp\Pair{D(g)}{g\comp\Proj{2}}$,
\item
	$D(D(f))\comp\Pair{\Pair{g}{0}}{\Pair{h}{k}} = D(f)\comp\Pair{g}{k}$,
\item
	$D(D(f))\comp\Pair{\Pair{0}{h}}{\Pair{g}{k}} = D(D(f))\comp\Pair{\Pair{0}{g}}{\Pair{h}{k}}$.
\end{enumerate}
\end{definition}

\noindent We try to provide some intuitions on these axioms. 
(D1) says that the operator $D(-)$ is linear; 
(D2) says that $D(-)$ is additive in its first coordinate; 
(D3) and (D4) ask that $D(-)$ behaves coherently with the product structure; 
(D5) is the usual chain rule;
(D6) requires that $D(f)$ is linear in its first component.
(D7) states the independence of the order of ``partial differentiation''.

\begin{remark}\label{rem:partder} In a Cartesian differential category we obtain partial derivatives from the 
full ones by ``zeroing out'' the components on which the differentiation is not required. 
For example, suppose that we want to define the partial derivative $D_1(f)$ of $f:C\times A \to B$ on its first component;
then, it is sufficient to set $D_1(f) = D(f)\comp(\Pair{\Id{C}}{0_A}\times\Id{C\times A}):C\times (C\times A)\to B$. 

Similarly, we define $D_2(f) = D(f)\comp(\Pair{0_C}{\Id{A}}\times\Id{C\times A}): A\times(C\times A)\to B$, 
the partial derivative of $f$ on its second component.
\end{remark}

This remark follows since every differential $D(f)$ can be reconstructed from its partial derivatives as follows:
$$
	\begin{array}{rl}
	D(f) =& D(f)\comp\Pair{\Pair{\Proj 1\comp\Proj 1}{\Proj 2\comp\Proj 1}}{\Proj 2}\\
	=& D(f)\comp\Pair{\Pair{\Proj 1\comp\Proj 1}{0}}{\Proj 2} + D(f)\comp\Pair{\Pair{0}{\Proj 2\comp\Proj 1}}{\Proj 2}\\
	=& D(f)\comp(\Pair{\Id{}}{0}\times\Id{})\comp(\Proj 1\times\Id{}) + D(f)\comp(\Pair{0}{\Id{}}\times\Id{})\comp(\Proj 2\times\Id{})\\
	=& D_1(f)\comp(\Proj 1\times\Id{}) + D_2(f)\comp(\Proj 2\times\Id{}).\\
	\end{array}
$$

\subsection{Linear Morphisms}

In Cartesian differential categories we are able to express the fact that a morphism is `linear' by asking that its differential is constant.

\begin{definition} In a Cartesian differential category, a morphism $f : A\to B$ is called \emph{linear} if $D(f) = f\comp\Proj 1$.
\end{definition}

\begin{lemma} Every linear morphism $f:A\to B$ is additive.
\end{lemma}

\begin{proof} By definition of linear morphism we have $D(f) = f\comp\Proj 1$. For all $g,h : C\to A$ we have
$$
\begin{array}{l}
f\comp (g + h) = f\comp \Proj 1 \comp\Pair{g+h}{g} = D(f)\comp\Pair{g+h}{g} =\\ 
D(f)\comp\Pair{g}{g} + D(f)\comp\Pair{h}{g}= f\comp\Proj 1\comp\Pair{g}{g} + f\comp\Proj 1\comp\Pair{h}{g} = f \comp g + f\comp h
\end{array}
$$
Moreover $f\comp 0 = f\comp \Proj 1 \comp\Pair{0}{0} = D(f)\comp\Pair{0}{0} = 0$.
We conclude that $f$ is additive.
\end{proof}

\begin{lemma} The composition of two linear morphisms is linear. 
\end{lemma}

\begin{proof} Let $f,g$ be two linear maps. We have to prove that $D(f\comp g) = f\comp g\comp\Proj 1$.
By {\em (D5)} we have $D(f\comp g) = D(f)\comp\Pair{D(g)}{g\comp\Proj{2}}$. 
Since $f,g$ are linear we have $D(f)\comp\Pair{D(g)}{g\comp\Proj{2}} = f\comp\Proj 1\comp \Pair{g\comp\Proj 1}{g\comp\Proj{2}} = f\comp g\comp\Proj 1$.
\end{proof}

Thus, in fact, every Cartesian differential category has a subcategory of linear maps.

\subsection{Cartesian Closed Differential Categories}\label{subsec:DLC}

Cartesian differential categories are not enough to interpret the differential $\lambda$-calculus, 
since the differential operator does not behave automatically well with respect to the 
Cartesian closed structure.
For this reason we now introduce the notion of \emph{Cartesian closed differential category}.

\begin{definition} A category is \emph{Cartesian closed left-additive} if it is a Cartesian left-additive category 
which is Cartesian closed and satisfies:
$$
\textrm{\em (+-curry)}\quad\curry(f + g) = \curry(f) + \curry(g)\qquad\qquad\qquad\qquad\qquad\qquad\qquad\curry(0) = 0\quad\textrm{\em (0-curry)}
$$
\end{definition}
From these properties of $\curry(-)$ we can easily prove that the evaluation morphism is 
additive in its left component.
\begin{lemma}\label{lemma:evalplus} In every Cartesian closed left-additive category the following axioms hold (for all $f,g:C\to[\Funint AB]$ and $h:C\to A$): 
$$
\textrm{\em (+-eval)}\quad \eval\comp\Pair{f+g}{h} = \eval\comp\Pair{f}{h} + \eval\comp\Pair{g}{h}\qquad\qquad\qquad\quad\eval\comp\Pair{0}{h} = 0\quad\textrm{\em (0-eval)}
$$
\end{lemma}
\begin{proof} Let $f' = \curry^-(f)$ and $g' = \curry^-(g)$. Then we have:
$$
	\begin{array}{rll}
	\eval\comp\Pair{f+g}{h}& = \eval\comp((\curry(f')+\curry(g'))\times\Id{})\comp \Pair{\Id{}}{h} &\textrm{by def.\ of }f',g' \\
	&=\curry^-((\curry(f')+\curry(g'))\comp \Pair{\Id{}}{h} &\textrm{by def.\ of }\curry^-\\
	&=\curry^-(\curry(f'+g'))\comp \Pair{\Id{}}{h} &\textrm{by (+-curry)}\\
 	&=(f'+g')\comp \Pair{\Id{}}{h} &\textrm{by def.\ of }\curry^-\\
 	&=f'\comp \Pair{\Id{}}{h} + g'\comp \Pair{\Id{}}{h} &\textrm{by left-additivity}\\
	&=\curry^-(f)\comp \Pair{\Id{}}{h} + \curry^-(g)\comp \Pair{\Id{}}{h} &\textrm{by def.\ of }f',g'\\
	&=\eval \comp(f\times\Id{})\comp \Pair{\Id{}}{h} + \eval\comp(g\times\Id{})\comp \Pair{\Id{}}{h} &\textrm{by def.\ of }\curry^-\\
	 &=\eval\comp\Pair{f}{h} + \eval\comp\Pair{g}{h}
	 \end{array}
$$
Moreover $\eval\comp\Pair{0}{g} = \eval\comp\Pair{\curry(0)}{g} = 0\comp\Pair{\Id{}}{g} = 0$.
\end{proof}

\begin{definition} A \emph{Cartesian closed differential category} is a Cartesian differential category which is 
Cartesian closed left-additive and such that, for all $f:C\times A\to B$: 
$$
    \textrm{{\em (D-curry)}}\quad D(\curry(f)) = \curry(D(f)\comp\Pair{\Proj 1\times 0_A}{\Proj 2\times\Id{A}}).
$$
\end{definition}

Indeed, in a Cartesian closed differential category we have two ways to derivate $f:C\times A\to B$ in its first component:
we can use the trick of Remark~\ref{rem:partder}, or we can `hide' the component $A$ by currying $f$ and then derive $\curry(f)$.
Intuitively, (D-curry) requires that these two methods are equivalent.

\begin{lemma}\label{D-eval-lemma} In every Cartesian closed differential category 
the following axiom holds (for all $h:C\to[\Funint{A}{B}]$ and $g:C\to A$):
$$
	\textrm{{\em (D-eval)}}\ D(\eval\comp\Pair{h}{g}) = \eval\comp\Pair{D(h)}{g\comp\Proj 2} + D(\curry^-(h)) \comp\Pair{\Pair{0_{C}}{D(g)}}{\Pair{\Proj 2}{g\comp\Proj 2}}
$$
\end{lemma}

\begin{proof} Let $h' = \curry^-(h) : C\times A\to B$. Then we have:
$$
\begin{array}{ll}
D(\eval\comp\Pair{h}{g}) = &\textrm{by def.\ of }h'\\
D(\eval\comp\Pair{\curry(h')}{g}) = &\textrm{by (beta-cat)}\\
D(h'\comp\Pair{\Id{C}}{g}) = &\textrm{by (D5)}\\
D(h')\comp\Pair{D(\Pair{\Id{C}}{g})}{\Pair{\Id{C}}{g}\comp\Proj 2} = &\textrm{by (D4) and (D3)}\\
D(h')\comp\Pair{\Pair{\Proj 1}{D(g)}}{\Pair{\Proj 2}{g\comp\Proj 2}} = &\textrm{since pairing is additive}\\
D(h')\comp\Pair{\Pair{\Proj 1}{0_A}+\Pair{0_C}{D(g)}}{\Pair{\Proj 2}{g\comp\Proj 2}} = &\textrm{by (D2)}\\
D(h')\comp\Pair{\Pair{\Proj 1}{0_A}}{\Pair{\Proj 2}{g\comp\Proj 2}} + D(h') \comp\Pair{\Pair{0_{C}}{D(g)}}{\Pair{\Proj 2}{g\comp\Proj 2}} = &\textrm{}\\
D(h')\comp\Pair{\Proj 1\times 0_A}{\Proj 2\times\Id{A}}\comp\Pair{\Id{C\times C}}{g\comp\Proj 2}\\
\quad +\ D(h') \comp\Pair{\Pair{0_{C}}{D(g)}}{\Pair{\Proj 2}{g\comp\Proj 2}} = &\textrm{by (beta-cat)}\\
\eval\comp\Pair{\curry(D(h')\comp\Pair{\Proj 1\times 0_A}{\Proj 2\times\Id{A}})}{g\comp\Proj 2}\\
\quad +\  D(h') \comp\Pair{\Pair{0_{C}}{D(g)}}{\Pair{\Proj 2}{g\comp\Proj 2}} = &\textrm{by (D-curry)}\\
\eval\comp\Pair{D(\curry(h'))}{g\comp\Proj 2} + D(\curry^-(\curry(h'))) \comp\Pair{\Pair{0_{C}}{D(g)}}{\Pair{\Proj 2}{g\comp\Proj 2}} = &\textrm{by def.\ of }h'\\
\eval\comp\Pair{D(h)}{g\comp\Proj 2} + D(\curry^-(h)) \comp\Pair{\Pair{0_{C}}{D(g)}}{\Pair{\Proj 2}{g\comp\Proj 2}}\\
\end{array}
$$
\end{proof}

The axiom \emph{(D-eval)} can be seen as a chain rule for denotations of differential $\lambda$-terms 
(cf. Lemma~\ref{lemma:main2}(i), below).

In Cartesian closed differential categories we are able to define a binary operator $\star$ on morphisms, 
that can be seen as the semantic counterpart of differential substitution.

\begin{definition}\label{def:star} The operator 
$$
\infer[(\star)]{f\star g:C\times A\to B}{f:C\times A\to B & g:C\to A}
$$
is defined by $f\star g = D(f)\comp\Pair{\Pair{0^{C\times A}_C}{g\comp\Proj1}}{\Id{C\times A}}$.
\end{definition}
The morphism $f\star g$ is obtained by differentiating $f$ in its second component,
and applying $g$ in that component.

\begin{remark} Actually the operators $D(-)$ and $\star$ are mutually definable. 
To define $D(-)$ in terms of $\star$ just set $D(f) = (f\comp\Proj2)\star\Id{}$.
To check that this definition is meaningful we show that it holds in every Cartesian differential category:
indeed, by Definition~\ref{def:star}, $(f\comp\Proj2)\star\Id{} = 
D(f\comp\Proj2)\comp\Pair{\Pair{0}{\Proj1}}{\Id{}} = 
D(f)\comp\Pair{\Proj2\comp\Proj1}{\Proj2\comp\Proj2}\comp\Pair{\Pair{0}{\Proj1}}{\Id{}} =
D(f)$.
Thus it would be possible to formulate the whole theory of Cartesian closed differential categories
by axiomatizing the behaviour of $\star$ instead of that of $D(-)$.
In this work we prefer to use $D(-)$ because it is a more basic operation, already studied in the literature,
and the complexities of the two approaches are comparable.
\end{remark}

It is possible to characterize linear morphisms in terms of the operator $\star$ as follows.

\begin{lemma}\label{lemma:eveylinadd} A morphism $f:A\to B$ is linear iff for all $g:C\to A$:
$$
	(f\comp\Proj 2)\star g = (f\comp g)\comp \Proj 1 : C\times A \to B
$$
\end{lemma}

\begin{proof} $(\imp)$ Suppose that $f$ is linear. 
By definition of $\star$ we have that $(f\comp\Proj 2)\star g = D(f\comp\Proj 2)\comp\Pair{\Pair{0_C}{g\comp\Proj 1}}{\Id{C\times A}}$.
By applying \emph{(D5)} and \emph{(D3)}, this is equal to $D(f)\comp\Pair{\Proj 2\comp\Proj 1}{\Proj 2\comp\Proj 2}\comp\Pair{\Pair{0_C}{g\comp\Proj 1}}{\Id{C\times A}} = D(f)\comp\Pair{g\comp\Proj 1}{\Proj 2}$. 
Since $f$ is linear we have $D(f) = f\comp\Proj 1$, thus $D(f)\comp\Pair{g\comp\Proj 1}{\Proj 2} = f\comp g\comp \Proj 1$.

$(\Leftarrow)$ Suppose $(f\comp\Proj 2)\star g = (f\comp g)\comp \Proj 1$ for all $g:C\to A$. 
In particular, this is true for $C = A$ and $g = \Id{A}$.
Thus we have $(f\comp\Proj 2)\star \Id{A} = f\comp \Proj 1$. 
We conclude since:
$$
\begin{array}{rll}
(f\comp\Proj 2)\star \Id{A} = &D(f\comp\Proj 2)\comp\Pair{\Pair{0_A}{\Proj 1}}{\Id{A\times A}}&\textrm{by def.\ of }\star\\
=&D(f)\comp\Pair{\Proj 2\comp\Proj 1}{\Proj 2\comp\Proj 2}\comp\Pair{\Pair{0_A}{\Proj 1}}{\Id{A\times A}}&\textrm{by (D5)+(D3)}\\
=&D(f)\comp\Pair{\Proj 1}{\Proj 2}=D(f)\\
\end{array}
$$
\end{proof}

The operator $\star$ enjoys the following commutation property.

\begin{lemma}\label{lemma:new} Let $f:C\times A\to B$ and $g,h:C\to A$. Then $(f\star g)\star h = (f\star h)\star g$.
\end{lemma}

\begin{proof} We set $\phi_g = \Pair{\Pair{0_C}{g\comp\Proj 1}}{\Id{C\times A}}$ and $\phi_h = \Pair{\Pair{0_C}{h\comp\Proj 1}}{\Id{C\times A}}$. We have:\\
$
\begin{array}{ll}
(f\star g)\star h = D(D(f)\comp\Pair{\Pair{0_C}{g\comp\Proj 1}}{\Id{C\times A}})\comp\phi_h = &\textrm{by (D5)}\\
D(D(f))\comp\Pair{D(\Pair{\Pair{0_C}{g\comp\Proj 1}}{\Id{}})}{\Pair{\Pair{0_C}{g\comp\Proj 1}}{\Id{}}\comp\Proj 2}\comp\phi_h = &\textrm{by (D4)}\\
D(D(f))\comp\Pair{\Pair{\Pair{0_C}{D(g\comp\Proj 1)}}{\Proj 1}}{\Pair{\Pair{0_C}{g\comp\Proj 1}}{\Id{}}\comp\Proj 2}\comp\phi_h = &\textrm{by (D5)}\\
D(D(f))\comp\Pair{\Pair{\Pair{0_C}{D(g)\comp\Pair{\Proj 1\comp\Proj 1}{\Proj 1\comp\Proj 2}}}{\Proj 1}}{\Pair{\Pair{0_C}{g\comp\Proj 1}}{\Id{}}\comp\Proj 2}\comp\phi_h = \\
D(D(f))\comp\Pair{\Pair{\Pair{0_C}{D(g)\comp\Pair{0_C}{\Proj 1}}}{\Pair{0_C}{h\comp\Proj 1}}}{\Pair{\Pair{0_C}{g\comp\Proj 1}}{\Id{}}} = &\textrm{by (D2)}\\
D(D(f))\comp\Pair{\Pair{0_{C\times A}}{\Pair{0_C}{h\comp\Proj 1}}}{\Pair{\Pair{0_C}{g\comp\Proj 1}}{\Id{}}} = &\textrm{by (D7)}\\
D(D(f))\comp\Pair{\Pair{0_{C\times A}}{\Pair{0_C}{g\comp\Proj 1}}}{\Pair{\Pair{0_C}{h\comp\Proj 1}}{\Id{}}} = &\textrm{by (D2)}\\
D(D(f))\comp\Pair{\Pair{\Pair{0_C}{D(h)\comp\Pair{0_C}{\Proj 1}}}{\Pair{0_C}{g\comp\Proj 1}}}{\Pair{\Pair{0_C}{h\comp\Proj 1}}{\Id{}}} = \\
D(D(f))\comp\Pair{\Pair{\Pair{0_C}{D(h)\comp\Pair{\Proj 1\comp\Proj 1}{\Proj 1\comp\Proj 2}}}{\Proj 1}}{\Pair{\Pair{0_C}{h\comp\Proj 1}}{\Id{}}\comp\Proj 2}\comp\phi_g = &\textrm{by (D5)}\\
D(D(f))\comp\Pair{\Pair{\Pair{0_C}{D(h\comp\Proj 1)}}{\Proj 1}}{\Pair{\Pair{0_C}{h\comp\Proj 1}}{\Id{}}\comp\Proj 2}\comp\phi_g =&\textrm{by (D4)}\\
D(D(f))\comp\Pair{D(\Pair{\Pair{0_C}{h\comp\Proj 1}}{\Id{CA}})}{\Pair{\Pair{0_C}{h\comp\Proj 1}}{\Id{}}\comp\Proj 2}\comp\phi_g = &\textrm{by (D5)}\\
D(D(f)\comp\Pair{\Pair{0_C}{h\comp\Proj 1}}{\Id{}})\comp\phi_g =(f\star h)\star g\\
\end{array}
$\\
\end{proof}

\begin{definition}\label{def:swap} 
Let $\sw_{ABC} = \Pair{\Pair{\Proj 1\comp\Proj 1}{\Proj 2}}{\Proj 2\comp\Proj 1} : (A\times B)\times C \to (A\times C)\times B$.
\end{definition}

\begin{remark}\label{rem:swap} 
$\sw\comp\sw = \Id{(A\times B)\times C}$, $\sw\comp \Pair{\Pair{f}{g}}{h} = \Pair{\Pair{f}{h}}{g}$ and $D(\sw) = \sw\comp\Proj 1$.
\end{remark}

The following two technical lemmas will be used in Subsection~\ref{subsec:soundness} to show the 
soundness of the categorical models of the differential $\lambda$-calculus. 
The interested reader can find the whole proofs in the technical Appendix~\ref{app:proofs}.

\begin{lemma}\label{lemma:main1} Let $f:(C\times A)\times D\to B$ and $g:C\to A$,  $h:C\to B'$. Then:
\begin{enumerate}[(i)]
\item\label{lemma:main11} $\Proj{2}\star g = g\comp \Proj{1}$,
\item\label{lemma:main12} $(h\comp\Proj{1})\star g = 0$,
\item\label{lemma:main13} $\curry(f)\star g = \curry(((f\comp\sw)\star(g\comp\Proj{1}))\comp\sw)$.
\end{enumerate}
\end{lemma}

\begin{proof} (Outline) $(i)$ follows by applying \emph{(D3)}.
$(ii)$ follows by applying \emph{(D2)}, \emph{(D3)} and \emph{(D5)}.
$(iii)$ follows by \emph{(Curry)}, \emph{(D-curry)} and \emph{(D2), (D3), (D5)}.
\end{proof}

\begin{lemma}\label{lemma:main2} Let $f: C\times A\to [\Funint{D}{B}]$ and $g:C\to A$, $h:C\times A\to D$. Then:
\begin{itemize}
\item[(i)] $(\eval\comp\Pair{f}{h})\star g = \eval\comp\Pair{f\star g + \curry(\curry^-(f)\star(h\star g))}{h}$,
\item[(ii)]  $\curry(\curry^-(f)\star h) \star g = \curry(\curry^-(f\star g)\star h) + \curry(\curry^-(f)\star(h\star g))$,
\item [(iii)] $\curry(\curry^-(f)\star h)\comp\Pair{\Id{C}}{g}  = \curry(\curry^-(f\comp\Pair{\Id{C}}{g})\star (h\comp \Pair{\Id{C}}{g}))$.
\end{itemize}
\end{lemma}

\begin{proof} (Outline) 
\emph{(i)} follows by applying \emph{(D-eval)} and \emph{(beta-cat)}.

\emph{(ii)} This equation can be simplified by using the axioms of Cartesian closed left-additive categories. 
Indeed, the right side can be written as $\curry((\curry^-(f\star g)\star h) + \curry^-(f)\star(h\star g))$. 
By taking a morphism $f'$ such that $f = \curry(f')$ and by applying Lemma~\ref{lemma:main1}$(iii)$ the 
item $(ii)$ becomes equivalent to
$((f'\star h)\comp\sw)\star(g\comp\Proj 1)\comp\sw = (((f'\comp\sw)\star(g\comp\Proj{1}))\comp\sw)\star h + f'\star (h\star g).$
This follows by {\em (Curry)} and \emph{(D2-7)}.

\emph{(iii)} follows by {\em (Curry)} and \emph{(D2-5)}.
\end{proof}

\section{Categorical Models of the Differential Lambda Calculus}\label{sec:cat-models}

In \cite{BucciarelliEM10} we have proved that Cartesian closed differential categories constitute \emph{sound}
models of the simply typed differential $\lambda$-calculus.
In this section we will show that all reflexive objects living in these
categories and satisfying a linearity condition are sound models of the \emph{untyped} version of this calculus.

\subsection{Linear Reflexive Objects in Cartesian Closed Differential Categories}\label{subsec:RODLC}

In a category $\cat{C}$, an object $A$ is a \emph{retract} of an object $B$, written $A\retract B$, 
if there are morphisms $f: A\to B$ and $g:B\to A$ such that $g\comp f = \Id{A}$. 
When also $f\comp g = \Id{B}$ holds we say that $A$ and $B$ are \emph{isomorphic}, written $A\cong B$, 
and that $f,g$ are \emph{isomorphisms}.

In a Cartesian closed category $\cat{C}$ a \emph{reflexive object} $\cU$ 
ought to mean a triple $(U,\App,\Abs)$ where $U$ is an object of $\catC$ and 
$\App:U\to [\Funint{U}{U}]$ and $\Abs:[\Funint{U}{U}]\to U$ are two morphisms 
performing the retraction $[\Funint{U}{U}]\retract U$.
When $[\Funint UU] \cong U$ we say that $\cU$ is \emph{extensional}.

\begin{definition} A reflexive object $\cU = (U,\App,\Abs)$ in a Cartesian closed differential category is 
{\em linear} if both $\App$ and $\Abs$ are linear morphisms.
\end{definition}

We are now able to provide our definition of model of the untyped differential $\lambda$-calculus.

\begin{definition} A {\em categorical model $\cU$ of the differential $\lambda$-calculus} is a 
linear reflexive object in a Cartesian closed differential category. 
The model $\cU$ is called {\em extensional} if the reflexive object $\cU$ is extensional 
(i.e., $[\Funint UU]\cong U$).
\end{definition}

The following lemma is useful for proving that a reflexive object in a Cartesian closed differential category
 is linear.

\begin{lemma} Let $\cU$ be a reflexive object.
\begin{itemize}
\item[(i)] If $\App$ and $\Abs\comp\App$ are linear then $\cU$ is linear.
\item[(ii)] If $\cU$ is extensional and either $\App$ or $\Abs$ are linear then $\cU$ is linear.
\end{itemize}
\end{lemma}

\begin{proof} (i) Suppose $\App$ and $\Abs\comp\App$ are linear morphisms. 
We now show that also $\Abs$ is linear. 
Indeed we have:
$$
\begin{array}{rll}
D(\Abs) =& D(\Abs)\comp(\App\times\App)\comp(\Abs\times\Abs) = D(\Abs)\comp\Pair{\App\comp\Proj 1}{\App\comp\Proj 2}\comp(\Abs\times\Abs)=&\textrm{by $\App$ linear}\\
=&D(\Abs)\comp\Pair{D(\App)}{\App\comp\Proj 2}\comp(\Abs\times\Abs) = D(\Abs\comp\App)\comp(\Abs\times\Abs)=&\textrm{by $\Abs\comp\App$ linear}\\
=&\Abs\comp\App\comp\Proj 1\comp\Pair{\Abs\comp\Proj 1}{\Abs\comp\Proj 2} = \Abs\comp\App\comp\Abs\comp\Proj 1 = \Abs\comp\Proj 1.
\end{array}
$$

(ii) If $\App$ is linear then it follows directly from (i) since $\Abs\comp\App = \Id{U}$ and the identity is linear.
If $\Abs$ is linear, calculations analogous to those made in (i) show that also $\App$ is.
\end{proof}

Notice that, in general, there may be extensional reflexive objects that are not linear.
However, in the concrete example of Cartesian closed differential category we will 
provide in Section~\ref{sec:examples} every extensional reflexive object will be linear (see Corollary~\ref{cor:isolin}).

\begin{lemma}\label{lemma:main3} Let $\cU$ be a linear reflexive object and let 
$f:U^{n+1}\to [\Funint UU]$, $h:U^{n+1}\to U$ $g:U^n\to U$. Then:
\begin{enumerate}[(i)]
\item $\Abs\comp(f\star g)=(\Abs\comp f)\star g$,
\item $\App\comp(h\star g) = (\App\comp h)\star g$.
\end{enumerate}
\end{lemma}

\begin{proof} $(i)$ By definition of $\star$ we have $(\Abs\comp f)\star g = 
D(\Abs\comp f)\comp\Pair{\Pair{0_{U^n}}{g\comp\Proj 1}}{\Id{U^{n+1}}}$.
By \emph{(D5)} we have $D(\Abs\comp f) = D(\Abs)\comp \Pair{D(f)}{f\comp\Proj 2}$. 
Since $\Abs$ is linear we have $D(\Abs) = \Abs\comp\Proj 1$, thus 
$D(\Abs)\comp \Pair{D(f)}{f\comp\Proj 2} = 
\Abs\comp\Proj 1\comp \Pair{D(f)}{f\comp\Proj 2} = 
\Abs\comp D(f)$. 
Hence, $D(\Abs\comp f)\comp\Pair{\Pair{0_{U^n}}{g\comp\Proj 1}}{\Id{U^{n+1}}} =
\Abs\comp D(f)\comp \Pair{\Pair{0_{U^n}}{g\comp\Proj 1}}{\Id{U^{n+1}}} =
\Abs\comp(f\star g)$.

$(ii)$ Analogous to $(i)$.
\end{proof}

\subsection{Defining the Interpretation}

Let $\seq x = x_1,\ldots,x_n$ be an ordered sequence of variables without repetitions.
We say that $\seq x$ is \emph{adequate} for $S_1,\ldots,S_k\in\Lambda^d$ if 
$\FV(S_1,\ldots,S_k)\subseteq\{x_1,\ldots,x_n\}$.
Given an object $U$ we write $U^{\seq x}$ for the $\{x_1,\ldots,x_n\}$-indexed categorical 
product of $n$ copies of $U$ (when $n = 0$ we consider $U^{\seq x} = \Termobj$).
Moreover, we define the $i$-th projection $\Proj i^{\seq{x}} : U^{\seq x}\to U$ by 
$$\Proj i^{\seq x} = \left\{ 
\begin{array}{ll} 
\Proj 2 & \text{if } i = n,\\ 
\Proj i^{x_1,\ldots,x_{n-1}}\comp\Proj 1 & \text{otherwise}.\\ 
\end{array} \right.
$$

\begin{definition}\label{def:interpretation} 
Let $\cU$ be a categorical model, $S$ be a differential $\lambda$-term and $\seq x = x_1,\ldots,x_n$ be adequate for $S$. 
The \emph{interpretation of $S$ in $\cU$ (with respect to $\seq x$)} will be a morphism $\Int{S}_{\seq x} : U^{\seq x}\to U$ defined by induction as follows:
\begin{itemize}
\item $\Int{x_i}_{\seq x} = \Proj i^{\seq x}$,
\item $\Int{sT}_{\seq x} = \eval\comp\Pair{\App\comp\Int{s}_{\seq x}}{\Int{T}_{\seq x}}$,
\item $\Int{\lambda z.s}_{\seq x} = \Abs\comp\curry(\Int{s}_{\seq x,z})$, where by $\alpha$-conversion we suppose that $z$ does not occur in $\seq x$,
\item $\Int{\Dern{1}{s}{t}}_{\seq x} = \Abs\comp\curry(\curry^-(\App\comp\Int{s}_{\seq x})\star \Int{t}_{\seq x}) $,
\item $\Int{\Dern{n+1}{s}{t_1,\ldots,t_n,t_{n+1}}}_{\seq x} = \Abs\comp\curry(\curry^-(\App\comp\Int{\Dern{n}{s}{t_1,\ldots,t_n}}_{\seq x})\star \Int{t_{n+1}}_{\seq x}) $,
\item $\Int{0}_{\seq x} = 0_U^{U^{\seq x}}$,
\item $\Int{s + S}_{\seq x} = \Int{s}_{\seq x} + \Int{S}_{\seq x}$.
\end{itemize}
\end{definition}

\begin{remark} Easy calculations give 
$$
	\Int{\Dern{n}{s}{t_1,\ldots,t_n}}_{\seq x} = \Abs\comp\curry((\cdots(\curry^-(\App\comp\Int{s}_{\seq x})\star \Int{t_1}_{\seq x})\cdots)\star \Int{t_n}_{\seq x}).
$$
Lemma~\ref{lemma:new} entails that this interpretation does not depend on the chosen representative of the 
permutative equivalence class. 
In other words, we have $\Int{\Dern{n}{s}{t_1,\ldots,t_{n}}}_{\seq x} = \Int{\Dern{n}{s}{t_{\sigma(1)},\ldots,t_{\sigma(n)}}}_{\seq x}$ 
for every permutation $\sigma\in\mathfrak{S}_{n}$.
\end{remark}

\subsection{Soundness}\label{subsec:soundness}

Given a categorical model $\cU$ we can define the \emph{equational theory of $\cU$} as follows:
$$
	\Th(\cU) = \{ S = T \st \Int{S}_{\seq x} = \Int{T}_{\seq x}\textrm{ for some $\seq x$ adequate for }S,T\}.
$$
The aim of this section is to prove that the interpretation we have defined is \emph{sound}, i.e., that $\Th(\cU)$ 
is a differential $\lambda$-theory for every model $\cU$.

The following convention allows us to lighten the statements of our theorems.

\begin{convention} Hereafter, and until the end of the section, we consider a fixed (but arbitrary)
linear reflexive object $\cU$ living in a Cartesian closed differential category $\cat C$.
Moreover, whenever we write $\Int{S}_{\seq x}$, we suppose that $\seq x$ is an adequate sequence 
for $S$.
\end{convention}

The proof of the next lemma is easy, and it is left to the reader.
Recall that the morphism $\sw$ has been introduced in Definition~\ref{def:swap}.

\begin{lemma}\label{lemma:swap} Let $S\in\Lambda^d$.
\begin{itemize}
\item[(i)] If $z\notin\FV(S)$ then $\Int{S}_{\seq x;z} = \Int{S}_{\seq x}\comp\Proj 1$, where $z$ does not occur in $\seq x$,
\item[(ii)] $\Int{S}_{\seq x;y;z} = \Int{S}_{\seq x;z;y}\comp\sw$, where $z$ and $y$ do not occur in $\seq x$.
\end{itemize}
\end{lemma}

\begin{theorem} \label{thm:subst1} (Classic Substitution Theorem) Let $S,T\in\Lambda^d$, $\seq x = x_1,\ldots,x_n$ 
and $y$ not occurring in $\seq x$. 
Then:
$$
	\Int{\subst{S}{y}{T}}_{\seq x} = \Int{S}_{\seq x;y}\comp\Pair{\Id{}}{\Int{T}_{\seq x}}.
$$
\end{theorem}

\begin{proof} By induction on $S$. The only interesting case is $S\equiv \Dern{n}{s}{u_1,\ldots,u_n}$: we treat it by cases on $n$.

Case $n=1$.
By definition of substitution we have $\Int{\subst{(\Der{s}{u_1})}{y}{T}}_{\seq x} = \Int{\Der{\subst{s}{y}{T}}{\subst{u_1}{y}{T}}}_{\seq x}$. 
By definition of $\Int{-}$ this is equal to $\Abs\comp\curry(\curry^-(\App\comp\Int{\subst{s}{y}{T}}_{\seq x})\star\Int{\subst{u_1}{y}{T}}_{\seq x})$.
By induction hypothesis we get $\Abs\comp\curry(\curry^-(\App\comp\Int{s}_{\seq x;y}\comp\Pair{\Id{}}{\Int{T}_{\seq x}})\star (\Int{u_1}_{\seq x;y}\comp \Pair{\Id{}}{\Int{T}_{\seq x}}))$. 
By applying Lemma~\ref{lemma:main2}$(iii)$ this is equal to $\Abs\comp\curry(\curry^-(\App\comp\Int{s}_{\seq x;y})\star\Int{u_1}_{\seq x;y})\comp\Pair{\Id{}}{\Int{T}_{\seq x}} = \Int{\Der{s}{u_1}}_{\seq x;y}\comp\Pair{\Id{}}{\Int{T}_{\seq x}}$. 

Case $n>1$. By definition of substitution we have $\Int{\subst{(\Dern{n}{s}{u_1,\ldots,u_n})}{y}{T}}_{\seq x} =
\Int{(\Dern{n}{\subst{s}{y}{T}}{\subst{u_1}{y}{T},\ldots,\subst{u_n}{y}{T}})}_{\seq x}$.
Applying the definition of $\Int{-}$ this is equal to $\Abs\comp\curry(\curry^-(\App\comp\Int{\Dern{n-1}{\subst{s}{y}{T}}{\subst{u_1}{y}{T},\ldots,\subst{u_{n-1}}{y}{T}}}_{\seq x})\star \Int{\subst{u_{n}}{y}{T}}_{\seq x})$.
By definition of substitution this is $\Abs\comp\curry(\curry^-(\App\comp\Int{\subst{(\Dern{n-1}{s}{u_1,\ldots,u_{n-1}})}{y}{T}}_{\seq x})\star \Int{\subst{u_{n}}{y}{T}}_{\seq x})$. By Applying the induction hypothesis twice we get
$\Abs\comp\curry(\curry^-(\App\comp\Int{(\Dern{n-1}{s}{u_1,\ldots,u_{n-1}}}_{\seq x,y}\comp\Pair{\Id{}}{\Int{T}_{\seq x}})\star (\Int{u_{n}}_{\seq x,y}\comp\Pair{\Id{}}{\Int{T}_{\seq x}}))$.
By Lemma~\ref{lemma:main2}$(iii)$ this is equal to 
$\Abs\comp\curry(\curry^-(\App\comp\Int{\Dern{n-1}{s}{u_1,\ldots,u_{n-1}}}_{\seq x,y})\star \Int{u_{n}}_{\seq x,y})\comp\Pair{\Id{}}{\Int{T}_{\seq x}} =
 \Int{(\Dern{n}{s}{u_1,\ldots,u_n})}_{\seq x;y}\comp\Pair{\Id{}}{\Int{T}_{\seq x}}$.
\end{proof}

\begin{theorem}\label{thm:subst2} (Differential Substitution Theorem) Let $S,T\in\Lambda^d$, $\seq x = x_1,\ldots,x_n$ 
and $y$ not occurring in $\seq x$. 
Then:
$$
	\Int{\dsubst{S}{y}{T}}_{\seq x;y} = \Int{S}_{\seq x;y} \star \Int{T}_{\seq x}.
$$
\end{theorem}

\begin{proof} By structural induction on $S$. 
\begin{itemize}
\item case $S\equiv y$. Then $\Int{\dsubst{y}{y}{T}}_{\seq x,y} = \Int{T}_{\seq x,y} = \Int{T}_{\seq x}\comp\Proj 1= \Proj 2\star\Int{T}_{\seq x}= \Int{y}_{\seq x,y}\star\Int{T}_{\seq x}$ by Lemma~\ref{lemma:main1}$(i)$.

\item case $S\equiv x_i\neq y$. Then $\Int{\dsubst{x_i}{y}{T}}_{\seq x,y}= \Int{0}_{\seq x,y}= 0$. 
By Lemma~\ref{lemma:main1}$(ii)$ we have 
$0 = (\Int{x_i}_{\seq x}\comp\Proj 1)\star\Int{T}_{\seq x}= \Int{x_i}_{\seq x,y}\star\Int{T}_{\seq x}$.

\item case $S\equiv \lambda z.v$. 
By definition of differential substitution we have that $\Int{\dsubst{\lambda z.v}{y}{T}}_{\seq x,y} =
\Int{\lambda z.\dsubst{v}{y}{T}}_{\seq x,y} = \Abs\comp\curry(\Int{\dsubst{v}{y}{T}}_{\seq x,y,z})$. 
Applying Lemma~\ref{lemma:swap}$(ii)$, this is equal to $\Abs\comp\curry(\Int{\dsubst{v}{y}{T}}_{\seq x,z,y}\comp\sw)$.
By induction hypothesis we obtain $\Abs\comp\curry((\Int{v}_{\seq x,z,y}\star\Int{T}_{\seq x,z})\comp\sw)$. 
Supposing without loss of generality that $z\notin\FV(T)$ we have, by Lemma~\ref{lemma:swap}$(i)$, 
$\Int{T}_{\seq x,z} = \Int{T}_{\seq x}\comp\Proj 1$.
Thus, applying Lemma~\ref{lemma:main1}$(iii)$, we have that 
$$
	\Abs\comp\curry((\Int{v}_{\seq x,z,y}\star(\Int{T}_{\seq x}\comp\Proj 1))\comp\sw) = 
	\Abs\comp(\curry(\Int{v}_{\seq x,z,y} \comp\sw)\star\Int{T}_{\seq x})
$$
which is equal to $\Abs\comp(\curry(\Int{v}_{\seq x,y,z})\star\Int{T}_{\seq x})$ by Lemma~\ref{lemma:swap}$(ii)$. 
Since $\cU$ is linear, we can apply Lemma~\ref{lemma:main3}$(i)$ and get 
$\Abs\comp(\curry(\Int{v}_{\seq x,y,z})\star\Int{T}_{\seq x})=(\Abs\comp\curry(\Int{v}_{\seq x,y,z}))\star\Int{T}_{\seq x} =
\Int{\lambda z.v}_{\seq x,y}\star\Int{T}_{\seq x}$.

\item case $S\equiv sU$. 
By definition of differential substitution we have that 
$\Int{\dsubst{sU}{y}{T}}_{\seq x,y} = 
\Int{(\dsubst{s}{y}{T})U}_{\seq x,y} + \Int{(\Der{s}{(\dsubst{U}{y}{T})})U}_{\seq x,y}
$. 
Let us consider the two addenda componentwise. 
On the one side we have 
$\Int{(\dsubst{s}{y}{T})U}_{\seq x,y} = 
\eval\comp\Pair{\App\comp\Int{\dsubst{s}{y}{T}}_{\seq x,y}}{\Int{U}_{\seq x,y}}$ which is equal, 
by induction hypothesis, to 
$\eval\comp\Pair{\App\comp(\Int{s}_{\seq x,y}\star\Int{T}_{\seq x})}{\Int{U}_{\seq x,y}}$.
By Lemma~\ref{lemma:main3}$(ii)$ this is equal to $\eval\comp\Pair{(\App\comp\Int{s}_{\seq x,y})\star\Int{T}_{\seq x}}{\Int{U}_{\seq x,y}}$.

On the other side we have (using $\App\comp\Abs = \Id{\Funint{U\ }{\ U}}$): 
$$
	\Int{(\Der{s}{(\dsubst{U}{y}{T})})U}_{\seq x,y} = 
	\eval\comp\Pair{\curry(\curry^-(\App\comp\Int{s}_{\seq x,y})\star
	\Int{\dsubst{U}{y}{T}}_{\seq x,y})}{\Int{U}_{\seq x,y}},
$$
by induction hypothesis this is equal to 
$$
	\eval\comp\Pair{\curry(\curry^-(\App\comp\Int{s}_{\seq x,y})\star(\Int{U}_{\seq x,y}\star\Int{T}_{\seq x}))}{\Int{T}_{\seq x, y}}. 
$$
By applying Lemma~\ref{lemma:evalplus} we can rewrite the sum of this two addenda as follows:
$$
\eval\comp\Pair{(\App\comp\Int{s}_{\seq x,y})\star\Int{T}_{\seq x} + 
\curry(\curry^-(\App\comp\Int{s}_{\seq x,y})\star(\Int{U}_{\seq x,y}\star\Int{T}_{\seq x}))}{\Int{U}_{\seq x,y}}.
$$
By Lemma~\ref{lemma:main2}$(i)$ this is 
$(\eval\comp\Pair{\App\comp\Int{s}_{\seq x,y}}{\Int{U}_{\seq x,y}})\star\Int{T}_{\seq x} = 
\Int{sU}_{\seq x,y}\star\Int{T}_{\seq x}$.

\item case $S\equiv \Dern{n}{v}{u_1,\ldots,u_n}$. By cases on $n$.

Subcase $n=1$. By definition of differential substitution, we have 
$$
	\Int{\Dsubst{(\Der{v}{u_1})}{y}{T}}_{\seq x,y} =
	\Int{\Der{(\dsubst{v}{y}{T})}{u_1}}_{\seq x,y} + 
	\Int{\Der{v}{(\dsubst{u_1}{y}{T})}}_{\seq x,y}.
$$
Consider the two addenda separately. 
On the one side we have
$\Int{\Der{(\dsubst{v}{y}{T})}{u_1}}_{\seq x,y} = 
\Abs\comp\curry(\curry^-(\App\comp\Int{\dsubst{v}{y}{T}}_{\seq x,y})\star\Int{u_1}_{\seq x,y})$.
By the inductive hypothesis this is equal to 
$\Abs\comp\curry(\curry^-(\App\comp(\Int{v}_{\seq x,y}\star \Int{T}_{\seq x}))\star\Int{u_1}_{\seq x,y})$,
which is equal to 
$\Abs\comp\curry(\curry^-((\App\comp\Int{v}_{\seq x,y})\star \Int{T}_{\seq x})\star\Int{u_1}_{\seq x,y})$
by Lemma~\ref{lemma:main3}$(ii)$.

On the other side, we have that $\Int{\Der{v}{(\dsubst{u_1}{y}{T})}}_{\seq x,y} = 
\Abs\comp\curry(\curry^-(\App\comp\Int{v}_{\seq x,y})\star\Int{\dsubst{u_1}{y}{T}}_{\seq x,y})$. 
By induction hypothesis this is 
$\Abs\comp\curry(\curry^-(\App\comp\Int{v}_{\seq x,y})\star(\Int{u_1}_{\seq x,y}\star\Int{T}_{\seq x}))$. 

Since $\Abs$ is linear, we can apply Lemma~\ref{lemma:eveylinadd} and write the sum of the two morphisms as:
$$
\Abs\comp\big(
\curry(\curry^-((\App\comp\Int{v}_{\seq x,y})\star \Int{T}_{\seq x})\star\Int{u_1}_{\seq x,y})
+ \curry(\curry^-(\App\comp\Int{v}_{\seq x,y})\star(\Int{u_1}_{\seq x,y}\star\Int{T}_{\seq x}))
\big).
$$
By applying Lemma~\ref{lemma:main2}$(ii)$, we obtain
$\Abs\comp(\curry(\curry^-(\App\comp\Int{v}_{\seq x,y})\star \Int{u_1}_{\seq x,y}) \star \Int{T}_{\seq x})$ 
which is equal to $\Int{\Der{v}{u_1}}_{\seq x,y} \star \Int{T}_{\seq x}$.

Subcase $n>1$. Performing easy calculations we get $\Int{\Dsubst{(\Dern{n}{v}{u_1,\ldots,u_n})}{y}{T}}_{\seq x;y} =
\Int{\Der{(\Dsubst{(\Dern{n-1}{v}{u_1,\ldots,u_{n-1}})}{y}{T})}{u_n}}_{\seq x;y} +
\Int{\Der{(\Dern{n-1}{v}{u_1,\ldots,u_{n-1}})}{(\dsubst{u_n}{y}{T})}}_{\seq x;y}$.
We consider the two addenda separately:
$$
\begin{array}{ll}
(1)\quad\Abs\comp\curry(\curry^-(\App\comp\Int{\Dsubst{(\Dern{n-1}{v}{u_1,\ldots,u_{n-1}})}{y}{T}}_{\seq x,y})\star\Int{u_n}_{\seq x,y}) =&\textrm{by IH}\\
\Abs\comp\curry(\curry^-(\App\comp(\Int{\Dern{n-1}{v}{u_1,\ldots,u_{n-1}}}_{\seq x,y}\star\Int{T}_{\seq x}))\star\Int{u_n}_{\seq x,y}) =&\textrm{by Lemma~\ref{lemma:main3}(ii)}\\
\Abs\comp\curry(\curry^-((\App\comp\Int{\Dern{n-1}{v}{u_1,\ldots,u_{n-1}}}_{\seq x,y})\star\Int{T}_{\seq x})\star\Int{u_n}_{\seq x,y}).\\
\\
(2)\quad \Abs\comp\curry(\curry^-(\App\comp\Int{\Dern{n-1}{v}{u_1,\ldots,u_{n-1}}}_{\seq x,y})\star\Int{\dsubst{u_n}{y}{T}}_{\seq x,y}) =&\textrm{by IH}\\
\Abs\comp\curry(\curry^-(\App\comp\Int{\Dern{n-1}{v}{u_1,\ldots,u_{n-1}}}_{\seq x,y})\star(\Int{u_n}_{\seq x,y}\star\Int{T}_{\seq x})).\\
\end{array}
$$
Since $\Abs$ is linear, we have that (1) + (2) is equal to
$$
 \begin{array}{l}
	\Abs\ \comp\big(\curry(\curry^-((\App\comp\Int{\Dern{n-1}{v}{u_1,\ldots,u_{n-1}}}_{\seq x,y})\star\Int{T}_{\seq x})\star\Int{u_n}_{\seq x,y})\ +\\
             \qquad\, \curry(\curry^-(\App\comp\Int{\Dern{n-1}{v}{u_1,\ldots,u_{n-1}}}_{\seq x,y})\star(\Int{u_n}_{\seq x,y}\star\Int{T}_{\seq x}))\big)\\
 \end{array}
$$
By Lemma~\ref{lemma:main2}$(ii)$ we get 
$\Abs\comp(\curry(\curry^-(\App\comp\Int{\Dern{n-1}{v}{u_1,\ldots,u_{n-1}}}_{\seq x,y})\star\Int{u_n}_{\seq x}) \star\Int{T}_{\seq x})$.
By Lemma~\ref{lemma:main3}$(i)$ this is equal to $\Abs\comp\curry(\curry^-(\App\comp\Int{\Dern{n-1}{v}{u_1,\ldots,u_{n-1}}}_{\seq x,y})\star\Int{u_n}_{\seq x}) \star\Int{T}_{\seq x}$, {\em i.e.}, to $\Int{\Dern{n}{v}{u_1,\ldots,u_n}}_{\seq x;y} \star \Int{T}_{\seq x}$.
\item all other cases ({\em i.e.,} $S\equiv 0$ and $S\equiv s + U$) are straightforward.
\end{itemize}
\end{proof}

We are now able to provide the main result of this section.

\begin{theorem}\label{thm:sound} (Soundness) Every reflexive object $\cU$ in a Cartesian closed differential category 
$\cat C$ is a sound model of the differential $\lambda$-calculus.
\end{theorem}

\begin{proof} \
It is easy to check that the categorical interpretation is contextual.
We now prove that $\Th(\cU)$ is closed under the rules $(\beta)$ and $(\beta_D)$
\begin{itemize}
\item
$(\beta)$ Let $\Int{(\lambda y.s)T}_{\seq x} = \eval\comp\Pair{\App\comp\Abs\comp\curry(\Int{s}_{\seq x,y})}{\Int{T}_{\seq x}}$. Since $\App\comp\Abs = \Id{}$ this is equal to $\eval\comp\Pair{\curry(\Int{s}_{\seq x,y})}{\Int{T}_{\seq x}}$.
On the other side we have $\Int{\subst{s}{y}{T}}_{\seq x} = \Int{s}_{\seq x,y}\comp\Pair{\Id{}}{\Int{T}_{\seq x}}$ 
by the Theorem~\ref{thm:subst1} and, by (beta-cat), 
$\Int{s}_{\seq x,y}\comp\Pair{\Id{}}{\Int{T}_{\seq x}} = 
\eval\comp\Pair{\curry(\Int{s}_{\seq x,y})}{\Int{T}_{\seq x}}$.
\item 
$(\beta_D)$ Let $\Int{\Der{(\lambda y.s)}{t}}_{\seq x} = 
\Abs\comp\curry(\curry^-(\App\comp\Abs\comp\curry(\Int{s}_{\seq x,y}))\star\Int{t}_{\seq x})$.
Since $\App\comp\Abs = \Id{}$ this is equal to 
$\Abs\comp\curry(\curry^-(\curry(\Int{s}_{\seq x,y}))\star\Int{t}_{\seq x}) = 
\Abs\comp\curry(\Int{s}_{\seq x,y}\star\Int{t}_{\seq x})$. 
By applying Theorem~\ref{thm:subst2}, this is equal to 
$\Abs\comp\curry(\Int{\dsubst{s}{y}{t}}_{\seq x,y}) = 
\Int{\lambda y.\dsubst{s}{y}{t}}_{\seq x}
$.
\end{itemize}
We conclude that $\Th(\cU)$ is a differential $\lambda$-theory.
\end{proof}

The above theorem shows that linear reflexive objects in Cartesian closed differential categories are 
sound models of the untyped differential $\lambda$-calculus; 
it is not known at the moment whether this notion of model is also \emph{complete}
(i.e., whether for every differential $\lambda$-theory $\cT$ there is a linear reflexive object $\cU_\cT$ living in a 
suitable Cartesian closed differential category $\cat{C}_\cT$ such that $\Th(\cU) = \cT$). 
The problem of completeness will be discussed in Subsection~\ref{sec:FurtherWorks:completeness}.

\begin{proposition} If $\cU$ is extensional, then $\Th(\cU)$ is extensional.
\end{proposition}

\begin{proof} Like in the case of usual $\lambda$-calculus, easy calculations show that 
$\Int{\lambda x.sx}_{\seq x} = \Abs\comp \curry(\eval)\comp\App\comp\Int{s}_{\seq x}$ 
which is equal to $\Int{s}_{\seq x}$ since $\curry(\eval) = \Id{}$
and $\Abs\comp\App = \Id{}$. 
\end{proof}

\subsubsection{Comparison with the Categorical Models of the Untyped Lambda Calculus}

The definition of categorical model of the differential $\lambda$-calculus proposed in this paper
seems to be a generalization without surprises of the classical definition of model of the 
$\lambda$-calculus, i.e., the notion of reflexive object in a Cartesian closed category.
However, while this notion is -- by far -- the most famous categorical definition of model of $\lambda$-calculus,
it is not the most general one.
Indeed, as pointed out by Martini in \cite{Martini92}, in the proof of soundness \cite[Prop.~5.5.5]{Bare} 
for categorical models there is one axiom of Cartesian closed categories that is never used, namely 
the axiom (Id-Curry) which is equivalent to ask for the unicity of the operator $\curry(-)$ in the category
(and this entails $\curry(\curry^-(f)) = f$).

For this reason Martini proposed reflexive objects living in \emph{weak} Cartesian closed categories 
as a more general notion of model of $\lambda$-calculus.
In these categories we have just a retraction (not an isomorphism) between the homsets 
$\cat C(C\times A,B)\retract\catC(C,\Funint AB)$. 
Thus $\Funint AB$ is no longer an object representing \emph{exactly} $\catC(A,B)$ --- 
there are different objects that can equally well accomplish the job.
Recently, De Carvalho \cite{CarvalhoTh} successfully used this notion to build concrete models living in very 
natural weak Cartesian closed categories inspired from the semantics of linear logic.

In our differential framework this generalization cannot be applied since the proof of soundness 
relies on the fact that $\curry(\curry^-(f)) = f$.
This is actually needed to give a meaningful interpretation of the linear application $\Der{s}{t}$.
Hence the definition of categorical model of the differential $\lambda$-calculus
differs from the corresponding one for the usual $\lambda$-calculus more than one could imagine at a first look.

\subsection{Modeling the Taylor Expansion}

In this subsection we provide sufficient conditions for models living in Cartesian closed 
differential categories to equate all terms having the same Taylor expansion. 
As an interesting fact, this happens to be a property of the category rather than of the 
reflexive objects.
Therefore, all models living in a category ``modeling the Taylor expansion'' have an equational
theory including $\TE$.

Since the definition of the Taylor expansion asks for infinite sums, we need to consider 
Cartesian closed differential categories $\cat C$ where it is possible to sum infinitely 
many morphisms.
Formally, we require that for every countable set $I$ and every family $\{f_i\}_{i\in I}$ 
of morphisms $f_i : A\to B$ we have $\sum_{i\in I} f_i \in \cat C(A,B)$.
In this case we say that $\cat C$ \emph{has countable sums}.
To avoid the tedious problem of handling coefficients we suppose that the sum on the 
morphisms is idempotent.

\begin{definition} A Cartesian closed differential category  
\emph{models the Taylor Expansion} if it has countable sums and
the following axiom holds (for every $f:C\times A\to B$ and $g:C\to A$):
$$
	\textrm{(Taylor)}\qquad \eval\comp\Pair{f}{g} = \sum_{k \in\nat} ((\cdots (\curry^-(f)\underbrace{\star g)\cdots )\star g}_{k\textrm{ times}})\comp\Pair{\Id{}}{0}.
$$
\end{definition}

Recall that the Taylor expansion $S^*$ of a differential $\lambda$-term $S$ has been 
defined in Subsection~\ref{subsec:TE}.
Given a model $\cU$ of the differential $\lambda$-calculus living in a Cartesian closed 
differential category having countable sums we can extend the interpretation given in
Definition~\ref{def:interpretation} to terms in $\diffinf$ by setting 
$
	\Int{\Sigma_{i\in I} s_i}_{\seq x} = \sum_{i\in I}\Int{ s_i}_{\seq x},
$ for every countable set $I$. 

\begin{theorem}\label{thm:TaylorInt} Let $S$ be a differential $\lambda$-term and $\cU$ be a model living in a 
Cartesian closed differential category having countable sums and modeling the Taylor Expansion.
Then:
$$
	\Int{S}_{\seq x} = \Int{S^*}_{\seq x}.
$$
\end{theorem}

\begin{proof} By structural induction on $S$.
The only interesting case is $S \equiv sT$.
$$
\begin{array}{rll}
\Int{sT}_{\seq x} = & \eval\comp\Pair{\App\comp\Int{s}_{\seq x}}{\Int{T}_{\seq x}}&\textrm{by def.\ of }\Int{-}_{\seq x}\\
	       =& \sum_{k \in\nat} ((\cdots (\curry^-(\App\comp\Int{s}_{\seq x})\underbrace{\star \Int{T}_{\seq x})\cdots )\star \Int{T}_{\seq x}}_{k\textrm{ times}})\comp\Pair{\Id{}}{0}&\textrm{by (Taylor)}\\
		=&\sum_{k\in\nat}\eval\comp\Pair{\curry((\cdots(\curry^-(\Int{s}_{\seq x}\underbrace{\star\Int{T}_{\seq x})\cdots)\star\Int{T}_{\seq x}}_{k\textrm{ times}})}{0}&\textrm{ by (beta-cat)}\\		
		=&\sum_{k\in\nat}\eval\comp\Pair{\App\comp\Abs\comp\curry((\cdots(\curry^-(\Int{s}_{\seq x}\underbrace{\star\Int{T}_{\seq x})\cdots)\star\Int{T}_{\seq x}}_{k\textrm{ times}})}{0}&\textrm{ by }\App\comp\Abs = \Id{}\\		
		=&\sum_{k\in\nat}\eval\comp\Pair{\App\comp\Int{\Dern{k}{s}{T,\ldots,T}}_{\seq x}}{0}&\textrm{by def.\ of }\Int{-}_{\seq x}\\
		=&\Int{\Sigma_{k\in\nat} (\Dern{k}{s}{T,\ldots,T})0}_{\seq x}&\textrm{by def.\ of }\Int{-}_{\seq x}\\
		=&\Int{(sT)^*}_{\seq x}&\textrm{by def.\ of }(\cdot)^*\\
\end{array}
$$
\end{proof}

By adapting the proof of Theorem~\ref{thm:sound} one can prove that 
$\Int{S^*}_{\seq x} = \Int{\NF(S^*)}_{\seq x}$ for every differential $\lambda$-term $S$.
From this fact and Theorem~\ref{thm:TaylorInt} we get the following result.

\begin{corollary} Every model $\cU$ living in a Cartesian closed differential category that models the Taylor expansion 
satisfies $\TE \subseteq \Th(\cU)$.
\end{corollary}

\section{A Relational Model of the Differential Lambda Calculus}\label{sec:examples}

In this section we provide the main example of Cartesian closed differential category known in the literature.
What we have in mind is the category $\MRel$ \cite{Girard88,BucciarelliEM07}, which is the co-Kleisli category of the functor $\Mfin{-}$ over the 
$\star$-autonomous category $\Rel$ of sets and relations.
We will also show that the reflexive object $\cD$ living in $\MRel$ built in \cite{BucciarelliEM07} to model
the usual $\lambda$-calculus is linear, and then it constitutes a model of the untyped differential 
$\lambda$-calculus.
We will then provide a partial characterization of its equational theory showing that it contains $\lambda\beta\eta^d$
and $\TE$ (this follows from the fact that $\MRel$ models the Taylor expansion).

\begin{remark}\label{rem:Mfin} In \cite{BucciarelliEM10} we have provided another example of Cartesian closed differential 
category: the category $\MFin$, which is the co-Kleisli of the functor $\Mfin{-}$ over the 
$\star$-autonomous category of finiteness spaces and finitary relations \cite{Ehrhard05}.
In this paper we do not present the category $\MFin$ since it does not contain any reflexive object 
(see \cite{Ehrhard05,Vaux09}) and hence it cannot be used as a semantics of the untyped
differential $\lambda$-calculus.
Other examples of semantics useful for modeling the untyped differential $\lambda$-calculus (including semantics that
\emph{do not} model the Taylor expansion) will be discussed in Subsection~\ref{sec:FurtherWorks:examples}.
\end{remark}

\subsection{Relational Semantics}

We recall that the definitions and notations concerning multisets have been introduced in Subsection~\ref{subs:msets}.
We now provide a direct definition of the category $\MRel$:
\begin{itemize}
\item 
    The objects of $\MRel$ are all the sets.
\item 
    A  morphism from $A$ to $B$ is a relation from $\Mfin A$ to $B$; in other words, $\MRel(A,B)=\Pow{\Mfin A\times B}$.
\item 
    The identity of $A$ is the relation $\Id{A}=\{(\Mset \alpha,\alpha)\st \alpha\in A\}\in\MRel(A,A)$.
\item 
    The composition of $s\in\MRel(A,B)$ and $t\in\MRel(B,C)$ is defined by:
    $$
    \begin{array}{ll}
    t\comp s=\{(m,\gamma)\quad \st&\exists k\in\nat\ \exists (m_1,\beta_1),\dots,(m_k,\beta_k)\in s\textrm{ such that } \\
    &m = m_1\mcup\dots\mcup m_k\ \text{and}\ (\Mset{\beta_1,\dots,\beta_k},\gamma)\in t\}.\\
    \end{array}
    $$
\end{itemize}

Given two sets $A_1,A_2$, we denote by $\With{A_1}{A_2}$ their disjoint union $(\{1\}\times A_1) \cup (\{2\}\times A_2)$. 
Hereafter we adopt the following convention.

\begin{convention} 
We consider the canonical bijection between $\Mfin{A_1}\times\Mfin{A_2}$ and $\Mfin{\With{A_1}{A_2}}$ 
as an equality.
Therefore, we will still denote by $\Mpair{m_1}{m_2}$ the corresponding element of $\Mfin{\With{A_1}{A_2}}$.
\end{convention}

\begin{theorem}\label{thm:MRel-ccc} The category $\MRel$ is a Cartesian closed category.
\end{theorem}

\begin{proof} The terminal object $\Termobj$ is the empty set $\emptyset$, and the unique element of
$\MRel(A,\emptyset)$ is the empty relation.

Given two sets $A_1$ and $A_2$, their categorical product in $\MRel$ is
their disjoint union $\With{A_1}{A_2}$ and the projections $\Proj{1},\Proj{2}$ are given by:
$$
  \Proj{i}=\{(\Mset{(i,a)},a)\st a\in A_i\}\in\MRel(\With{A_1}{A_2},A_i)\textrm{, for } i =1,2.
$$
It is easy to check that this is actually the categorical product of $A_1$ and $A_2$ in $\MRel$; 
given $s\in\MRel(B,A_1)$ and $t\in\MRel(B,A_2)$, the corresponding morphism $\Pair st\in\MRel(B,\With{A_1}{A_2})$ is given by:
$$
  \Pair st=\{(m,(1,a))\st(m,a)\in s\}\cup\{(m,(2,b))\st(m,b)\in t\}\,.
$$
Given two objects $A$ and $B$, the exponential object $\Funint AB$ is $\Mfin A\times B$ and the evaluation morphism is given by:
$$
\eval_{AB} =\{((\Mset{(m,b)},m),b)\st m\in\Mfin A\ \text{and}\ b\in B\}\in\MRel(\With{[\Funint AB]}A,B)\,.
$$
Again, it is easy to check that in this way we defined an exponentiation. 
Indeed, given any set $C$ and any morphism $s\in\MRel(\With CA,B)$, there is exactly one morphism $\curry(s)\in\MRel(C,\Funint AB)$ such that:
$$
  \eval_{AB}\comp (\curry(s)\times \Id{S}) = s.
$$
which is $\curry(s)=\{(p,(m,b))\st(\Mpair{p}{m},b)\in s\}$. 
\end{proof}

\begin{theorem} The category $\MRel$ is a Cartesian closed differential category.
\end{theorem}

\begin{proof} By Theorem~\ref{thm:MRel-ccc} $\MRel$ is Cartesian closed.
It is Cartesian closed left-additive since every homset $\MRel(A,B)$
can be endowed with the following additive structure $(\MRel(A,B),\cup,\emptyset)$.

Finally, given $f\in\MRel(A,B)$ we can define its derivative as follows:
$$
	D(f) = \{ (([\alpha],m),\beta)\st (m\mcup [\alpha],\beta)\in f\}\in\MRel(\With AA,B).
$$

It is not difficult to check that $D(-)$ satisfies \emph{(D1-7)}. 
We now show that also \emph{(D-curry)} holds. 
Let $f\subseteq(\Mfin{C}\times\Mfin{A})\times B$. On the one side we have:
$$
D(\curry(f)) = \{ (([\gamma],m_1),(m_2,\beta)) \st ((m_1\mcup[\gamma],m_2),\beta)\in f\}.
$$
On the other side we have $D(f) =  f_1\cup f_2$, where:
\begin{gather*}
f_1 = \{ ((([\gamma],[]),(m_1,m_2)), \beta)\st ((m_1\mcup[\gamma],m_2),\beta)\in f \}, \\
f_2 = \{ ((([],[\alpha]),(m_1,m_2)), \beta)\st ((m_1,m_2\mcup[\alpha]),\beta)\in f \}.
\end{gather*}
Since \MRel{} is left-additive we have that 
$$
	(f_1\cup f_2)\comp\Pair{\Proj 1\times 0}{\Proj 2\times\Id{}} = 
	(f_1\comp\Pair{\Proj 1\times 0}{\Proj 2\times\Id{}})\cup(f_2\comp\Pair{\Proj 1\times 0}{\Proj 2\times\Id{}})
$$
Easy calculations give:
\begin{gather*}
f_1\comp\Pair{\Proj 1\times 0}{\Proj 2\times\Id{}}= \{ ((([\gamma],m_1),m_2),\beta)) \st ((m_1\mcup[\gamma],m_2),\beta)\in f\}\\
f_2\comp\Pair{\Proj 1\times 0}{\Proj 2\times\Id{}} = \emptyset.
\end{gather*}
We then get $\curry(D(f)\comp\Pair{\Proj 1\times 0}{\Proj 2\times\Id{}}) = \curry(f_1\comp\Pair{\Proj 1\times 0}{\Proj 2\times\Id{}}) = D(\curry(f))$.
\end{proof}

The operator $\star$ can be directly defined in $\MRel$ as follows:
$$
	f\star g = \{((m_1 \mcup m_2, m), \beta)\st (m_1,\alpha)\in g,\ ((m_2, m \mcup [\alpha]), \beta)\in f \} \in\MRel(\With CA, B).
$$
We now provide a characterization of the linear morphisms of $\MRel$.

\begin{lemma}\label{lemma:linsingleton} A morphism $f\in\MRel(A,B)$ is linear iff for all $(m,\beta)\in f$ we have that $m$ is a singleton.
\end{lemma}

\begin{proof} Easy calculations give $f\comp\Proj{1}= \{ ((m,[]),\beta) \st  (m,\beta)\in f \}$. 
This is equal to $D(f)$ if and only if $m$ is a singleton.
\end{proof}

\begin{corollary}\label{cor:isolin} In $\MRel$ every isomorphism is linear.
\end{corollary}

\begin{proof} 
Let $f\in\MRel(B,A)$ and $g\in\MRel(A,B)$ such that $f\comp g = \Id{A}$ and $g\comp f = \Id{B}$. 
Notice that $f$ does not contain any pair $([],\alpha)$ because otherwise such a pair would also appear in 
$f\comp g$, and this is impossible since $f\comp g= \Id{}$.
Similarly, $g$ cannot contain any pair $([],\beta)$. 
Thus:
$$
    f\comp g=\{([\alpha],\alpha)\st \exists \beta\in B\ ([\alpha],\beta)\in g\textrm{ and } (\Mset{\beta},\alpha)\in f\}.
$$
Since by hypothesis $f\comp g = \{([\alpha],\alpha) \st \alpha\in A\}$ we have that for all $\alpha\in A$ there is a 
$\beta\in B$ such that $([\beta],\alpha)\in f$. 
Suppose now, by the way of contradiction, that there is a $([\alpha_1,\ldots, \alpha_k],\beta)\in g$ such that
$k>1$. 
From the property above there are $\beta_1,\ldots, \beta_k\in B$ such that $([\beta_i],\alpha_i)\in f$ for $1\le i\le k$, 
thus we would have $([\beta_1,\ldots,\beta_k],\beta)\in f\comp g = \Id{B}$, which is impossible. 
By Lemma~\ref{lemma:linsingleton} we conclude that $g$ is linear.
Analogous considerations show that also $f$ is linear.
\end{proof}

\subsubsection{An Extensional Relational Model}

In this section we build a reflexive object $\cD$ in $\MRel$ which is extensional by construction, and hence linear 
by Corollary~\ref{cor:isolin}.
We first give some preliminary definitions.

Recall that $\nat$ denotes the set of natural numbers.
An $\nat$-indexed sequence $\sigma = (m_1, m_2,\ldots)$ of multisets is 
\emph{quasi-finite} if $m_i = []$ holds for all but a finite number of indices $i$.
If $X$ is a set, we denote by $\Omegatuple{X}$ the set of all quasi-finite 
$\nat$-indexed sequences of finite multisets over $X$. 
Notice that the only inhabitant of $\Omegatuple{X}$ is the sequence 
$([],[],[],\ldots)$.

We now define a family of sets $\{ D_n \}_{n\in\nat}$ as follows: 
\begin{itemize}
\item $D_0=\emptyset$,
\item $D_{n+1}=\Omegatuple{D_n}$.
\end{itemize}

Since the operation $S\mapsto\Omegatuple S$ is monotonic on sets, and since $D_0\subseteq D_1$, we have 
$D_n\subseteq D_{n+1}$ for all $n\in\nat$. 
Finally, we set $D=\cup_{n\in\nat}D_n$.

So we have $D_0=\emptyset$ and $D_1 = \{(\Mset{},\Mset{},\dots)\}$. 
The elements of $D_2$ are quasi-finite sequences of multisets over a singleton, i.e., quasi-finite sequences of natural numbers.
More generally, an element of $D$ can be represented as a finite tree which
alternates two kinds of layers: 
\begin{itemize}
\item
    ordered nodes (the quasi-finite sequences), where immediate subtrees are indexed by 
    distinct natural numbers,
\item
    unordered nodes where subtrees are organized in a {\em non-empty} multiset.
\end{itemize}
In order to define an isomorphism in $\MRel$ between $D$ and $[\Funint{D}{D}]=\Mfin{D}\times D$ it is enough to remark that 
every element $\sigma\in D$ is canonically associated with the pair $(\sigma_0, (\sigma_1, \sigma_2,\ldots))$ and {\em vice versa}. 
Given $\sigma\in D$ and $m\in\Mfin{D}$, we write $\at{m}{\sigma}$ for the element $\tau = (\tau_1,\tau_2,\ldots)\in D$ 
such that $\tau_1 = m$ and $\tau_{i+1} = \sigma_i$.
This defines a bijection between $\Mfin{D}\times D$ and $D$, and hence an isomorphism in $\MRel$ as follows: 

\begin{proposition} The triple $\cD = (D,\App,\Abs)$ where:
\begin{itemize}
\item
    $\Abs = \{(\Mset{(m,\sigma)},\at{m}{\sigma})\st m\in\Mfin{D},\sigma\in D\}\in\MRel(\Funint{D}{D}, D)$,
\item
    $\App = \{(\Mset{\at{m}{\sigma}},(m, \sigma)) \st m\in\Mfin{D},\sigma\in D\}\in\MRel(D,\Funint{D}{D})$,
\end{itemize}
is an extensional categorical model of differential $\lambda$-calculus.
\end{proposition}

\begin{proof} It is trivial that $\Abs\comp\App = \Id{D}$ and $\App\comp\Abs = \Id{[D\To D]}$. 
We conclude by Corollary~\ref{cor:isolin}.
\end{proof}

\subsection{Interpreting the Differential Lambda Calculus in $\cD$}\label{subsec:interpreting-DLC}

In Section~\ref{sec:cat-models}, we have defined the interpretation of a differential $\lambda$-term in 
any linear reflexive object of a Cartesian closed differential category. 
We provide the result of the corresponding computation, when it is performed in $\cD$.

Given a differential $\lambda$-term $S$ and a sequence $\seq{x} = x_1,\ldots,x_n$ adequate for $S$, the 
interpretation $\Int{S}_{\seq x}$ is an element of $\MRel(D^{\seq x},D)$, i.e., $\Int{S}_{\seq{x}}\subseteq\Mfin{D}^n\times D$. 
The interpretation $\Int{S}_{\seq x}$ is defined by structural induction on $S$ as follows:
\begin{itemize}
\item $\Int{x_i}_{\seq x}=
  \{((\Mset{},\dots,\Mset{},\Mset\sigma,\Mset{},\dots,\Mset{}),\sigma)
  \st\sigma\in D\}$, where the only non-empty multiset occurs in the $i$-th  position.
\item
    $\Int{sT}_{\seq{x}} = \{((m_1,\ldots,m_n),\sigma)\st\exists k\in\nat\\
   \begin{array}{lll}
   \hspace{4.36cm}&\exists (m^j_1,\ldots,m^j_n)\in\Mfin{D}^n &\textrm{for }j=0,\ldots,k\\
    &\exists \sigma_1,\ldots,\sigma_k\in D&\textrm {such that}\\
    &m_i = m_i^0\mcup\ldots\mcup m_i^k &\textrm{for }i=1,\ldots,n\\
    &((m_1^0,\ldots,m_n^0),\at{[\sigma_1,\ldots,\sigma_k]}{\sigma})\in\Int{s}_{\seq{x}}\\
    &((m_1^j,\ldots,m_n^j),\sigma_j)\in\Int{T}_{\seq{x}}&\textrm{ for }j=1,\ldots,k\},\\
    \end{array}$
\item
    $\Int{\lambda z.s}_{\seq{x}} = \{((m_1,\ldots,m_n),\at{m}{\sigma})\st ((m_1,\ldots,m_n,m),\sigma)\in \Int{s}_{\seq{x},z}\}$, where we assume that $z$ 
    does not occur in $\seq{x}$,
\item $\Int{\Dern{1}{s}{t}}_{\seq x} = \{ ((m_1\mcup m'_1,\ldots, m_n\mcup m'_n), m ::\beta) \st \exists \alpha\in D\ ((m_1,\ldots,m_n), \alpha)\in\Int{t}_{\seq x}$
and $((m'_1,\ldots,m'_n),m\mcup[\alpha] :: \beta)\in\Int{s}_{\seq x}\},$ 
\item $\Int{\Dern{n+1}{s}{t_1,\ldots, t_{n+1}}}_{\seq x} =  \{ ((m_1\mcup m'_1,\ldots, m_n\mcup m'_n), m ::\beta) \st \exists \alpha\in D\ ((m_1,\ldots,m_n), \alpha)\in\Int{t_{n+1}}_{\seq x}$
and $((m'_1,\ldots,m'_n),m\mcup[\alpha] :: \beta)\in\Int{\Dern{n}{s}{t_1,\ldots, t_{n}}}_{\seq x}\},$ 
\item $\Int{0}_{\seq x} = \emptyset$,
\item $\Int{s + S}_{\seq x} = \Int{s}_{\seq x}\cup \Int{S}_{\seq x}$.
\end{itemize}
Note that if $M$ is a \emph{closed} differential $\lambda$-term then $\Int{M}\subseteq D$.
Moreover, it is easy to check that $\Int{\Omega} = \emptyset$ 
(actually from \cite{Manzonetto09} we know that the interpretation of all unsolvable ordinary $\lambda$-terms is empty).
In the next subsection we will prove some general properties of $\Th(\cD)$.

\subsection{An Extensional Model of Taylor Expansion}

In \cite{Manzonetto09} we characterized the equational theory of $\cD$,
seen as a model of the pure untyped $\lambda$-calculus.
More precisely we proved that $\Th(\cD) = \mathcal{H}^\star$, the theory equating 
two $\lambda$-terms $M,N$ whenever they behave in the same way in every context.
This is not surprising since Ehrhard proved in \cite{Ehrhard09} 
that the continuous semantics \cite{Scott72} can be seen as the extensional collapse of the category $\MRel$ and 
that $\cD$ corresponds to Scott's $\cD_\infty$ under this collapse.

In this subsection we give a partial characterization of the theory of $\cD$ 
seen as a model of the differential $\lambda$-calculus.

\begin{remark} Given an arbitrary set $I$ and an $I$-indexed family of relations $\{f_i\}_{i\in I}$ from
$\Mfin{A}$ to $B$ we have that $\cup_{i\in I} f_i \subseteq \Mfin{A}\times B$.
In particular, $\MRel$ has countable sums.
\end{remark}

\begin{proposition} $\MRel$ models the Taylor expansion.
\end{proposition}

\begin{proof}
Let $f\subseteq\Mfin{C}\times(\Mfin{A}\times B)$ and $g\subseteq\Mfin{C}\times A$. Easy calculations give:
$$
    \begin{array}{rll}
    \eval\comp\Pair{f}{g} = &\{(m,\gamma)\st\exists k\in\nat\\
    &\qquad\qquad\, \exists m_j\in\Mfin{C} &\textrm{for }j=0,\ldots,k\\
    &\qquad\qquad\, \exists \alpha_1,\ldots,\alpha_k\in A&\textrm {such that}\\
    &\qquad\qquad\, m = m_0\mcup\ldots\mcup m_k &\textrm{for }i=1,\ldots,n\\
    &\qquad\qquad\, (m_0,([\alpha_1,\ldots,\alpha_k],\gamma))\in f\\
    &\qquad\qquad\, (m_j,\alpha_j)\in g&\textrm{ for }j=1,\ldots,k\}\\
    = &\bigcup_{k\in\nat}\{(m,\gamma)\st \exists m_j\in\Mfin{C} &\textrm{for }j=0,\ldots,k\\
    &\qquad\qquad\, \exists \alpha_1,\ldots,\alpha_k\in A&\textrm {such that}\\
    &\qquad\qquad\, m = m_0\mcup\ldots\mcup m_k &\textrm{for }i=1,\ldots,n\\
    &\qquad\qquad\, (m_0,([\alpha_1,\ldots,\alpha_k],\gamma))\in f\\
    &\qquad\qquad\, (m_j,\alpha_j)\in g&\textrm{ for }j=1,\ldots,k\}\\
    =&\sum_{k\in\nat}((\cdots (\curry^-(f)\underbrace{\star g)\cdots) \star g}_{k\textrm{ times}})\comp\Pair{\Id{A}}{\emptyset}&\\
    \end{array}
$$
\end{proof}

\begin{corollary} Every categorical model $\cU$ of the differential $\lambda$-calculus living in $\MRel$
satisfies $\TE\subseteq \Th(\cU)$.
\end{corollary}

\begin{corollary} The theory of $\cD$ includes both $\lambda\beta\eta^d$ and $\TE$.
\end{corollary}

\begin{conjecture} We conjecture that 
$$
	\Th(\cD) = \{ (S,T)\in\Lambda^d\times\Lambda^d \st\textrm{ for all contexts } C(\cdot),\ C(S)\textrm{ is solvable iff }C(T)\textrm{ is solvable }\},
$$
where a context is a differential $\lambda$-term with a hole denoted by $(\cdot)$, and $C(S)$ denotes the result 
of substituting $S$ (possiblly with capture of variables) for the hole in $C$.
`Solvable' here has to be intended as \emph{may-solvable}\footnote{
May and must solvability have been studied in \cite{PaganiR10} in the context of the resource calculus.
} (i.e., a sum of terms converges if at least one of its components converges).
\end{conjecture}
A complete syntactical characterization of the theory of $\cD$ is difficult to provide, and it is kept for future works.

\section{The Resource Calculus}\label{sec:ResCal}

In this section we present the resource calculus \cite{Boudol93,BoudolCL99} (using the formalization \emph{\`a la Tranquilli} given in \cite{PaganiT09}) 
and we show that every model of the differential $\lambda$-calculus is also a model of the resource calculus.
We then discuss the (tight) relationship existing between the differential $\lambda$-calculus and the resource calculus.

\subsection{Its Syntax}

The resource calculus has three syntactical categories: 
\emph{resource $\lambda$-terms} ($\Lambda^r$) that are in functional position; 
\emph{bags} ($\Lambda^b$) that are in argument position and represent multisets of resources, 
and \emph{sums} that represent the possible results of a computation. 
A \emph{resource} ($\Lambda^{(\bang)}$) can be linear or intuitionistic, in the latter case it is written with a $\bang$ apex. 
An \emph{expression} ($\Lambda^{(b)}$) is either a term or a bag.

Formally, we have the following grammar:
$$
\begin{array}{lllr}
\Lambda^r: & M,N,L &::=\quad x\ |\ \lambda x . M\ |\ MP \hspace{5cm}&\textrm{resource $\lambda$-terms}\\
\Lambda^{(\bang)}: & M^{(\bang)},N^{(\bang)} &::=\quad M\ |\ M^{\bang} &\textrm{resources}\\
\Lambda^{b}: & P,Q,R &::=\quad [M^{(\bang)}_1,\ldots,M^{(\bang)}_n] &\textrm{bags}\\
\Lambda^{(b)}: & A,B &::=\quad M\ |\ P &\textrm{expressions}\\
\end{array}
$$

Hereafter, resource $\lambda$-terms are considered up to $\alpha$-conversion and permutation of resources in the bags.
Intuitively, linear resources are available exactly once, while banged resources zero or many times.
\begin{definition}
Given an expression $A$ the set $\FV(A)$ of \emph{free variables of $A$} is defined by
induction on $A$ as follows: 
\begin{itemize}
\item $\FV(x) = \{x\}$, 
\item $\FV(\lambda x.M) = \FV(M) - \{x\}$, 
\item $\FV(MP) = \FV(M)\cup\FV(P)$, 
\item $\FV([]) = \emptyset$, 
\item $\FV([M^{(\bang)}]\mcup P) = \FV(M)\cup\FV(P)$. 
\end{itemize}
Given expressions $A_1,\ldots,A_k$ we set $\FV(A_1,\ldots,A_k) = \FV(A_1)\cup\cdots\cup\FV(A_k)$.
\end{definition}
Concerning sums, $\sums{r}$ (resp.\ $\sums{b}$) denotes the set of finite formal sums of terms (resp. bags).  
As usual, we suppose that the sum is commutative and associative, and that 0 is its neutral element.
$$
\sM,\sN\in \sums{r}\qquad \sP,\sQ\in \sums{b}\qquad \sA,\sB,\sC\in \sums{(b)} = \sums{b}\cup\sums{(b)}\hspace{1,2cm}\textrm{sums}
$$
Note that in writing $\sums{(b)}$ we are abusing the notation, as it does not denote the $\cN$-module generated over $\Lambda^{(b)} = \Lambda^{r} \cup \Lambda^{b}$ but rather the union of the two $\cN$-modules. In other words, sums must be taken only in the same sort.

The definition of $\FV(\cdot)$ is extended to elements of $\sums{(b)}$ in the obvious way.

In the grammar for resource $\lambda$-terms and bags sums do not appear, indeed in this calculus they may arise only on the ``surface''
(while in the differential $\lambda$-calculus sums may appear in the right argument of an application). 
Nevertheless, as a syntactic sugar and not as actual syntax, we extend all the constructors to sums as follows. 

\begin{notation} We set the following abbreviations on  $\sums{(b)}$.
\begin{itemize}
\item $\lambda x.\sum_{i=1}^k M_i = \sum_{i=1}^k \lambda x.M_i$,
\item $(\sum_{i=1}^k M_i)(\sum_{j=1}^n P_j) = (\sum_{i,j} M_i P_j)$,
\item $[(\sum_{i=1}^k M_i)]\mcup P = \sum_{i=1}^k[M_i]\mcup P$,
\item $[(\sum_{i=1}^k M_i)^\bang]\mcup P = [M_1^\bang,\ldots,M_k^\bang]\mcup P$.
\end{itemize}
\end{notation}

These equalities make sense since all constructors, but the $(\cdot)^\bang$, are linear. 
Notice the difference between these rules and the analogous ones for the
differential $\lambda$-calculus introduced in Notation~\ref{not:sumsonDLT}.
In the differential $\lambda$-calculus the application operator 
is only linear in its left component while here it is bilinear.

\begin{definition}\label{def:linearsubstres} Let $A$ be an expression and $N$ be a resource $\lambda$-term. 
\begin{itemize}
\item $\subst{A}{x}{N}$ is the usual substitution of $N$ for $x$ in $A$. 
It is extended to sums as in $\subst{\sA}{x}{\sN}$ by linearity\footnote{
A unary operator $F(\cdot)$ is extended by linearity by setting $F(\Sigma_i A_i) = \Sigma_{i,j} F( A_i)$.
} in $\sA$. 
\item $\lsubst{A}{x}{N}$ is the {\em linear substitution} defined inductively as follows:
$$
\begin{array}{ll}
\lsubst{y}{x}{N} = \left\{ 
\begin{array}{ll} 
N & \text{if } x = y\\ 
0 & \text{otherwise}\\ 
\end{array} \right.\qquad
&\begin{array}{l}
\lsubst{(\lambda y.M)}{x}{N} = \lambda y.\lsubst{M}{x}{N}\\
\lsubst{(M P)}{x}{N} = \lsubst{M}{x}{N}P+ M(\lsubst{P}{x}{N})\\
\end{array}\\
\lsubst{[M]}{x}{N} = [\lsubst{M}{x}{N}]&\lsubst{[]}{x}{N} = 0\\
\lsubst{[M^!]}{x}{N} = [\lsubst{M}{x}{N},\ M^!]&\lsubst{(P\mcup R)}{x}{N} = \lsubst{P}{x}{N}\mcup R+ P\mcup \lsubst{R}{x}{N}\\
\end{array}
$$
It is extended to $\lsubst{\sA}{x}{\sN}$ by bilinearity\footnote{
A binary operator $F(\cdot,\cdot)$ is extended by bilinearity by setting $F(\Sigma_i A_i,\Sigma_j B_j) = \Sigma_{i,j} F( A_i, B_j)$.
} in both $\sA$ and $\sN$.
\end{itemize}
\end{definition}

The operation $\lsubst{M}{x}{N}$ on resource $\lambda$-terms is roughly equivalent to the operation $\dsubst{S}{x}{T}$ on 
differential $\lambda$-terms (cf.~Lemma~\ref{lemma:subst} below). 
Notice that in defining $\lsubst{[M^!]}{x}{N}$ we morally extract a linear copy of $M$ from the infinitely many represented by $M^\bang$,
that receives the substitution, and we keep the other ones unchanged. 

\begin{example} \
\begin{enumerate}[1.] 
\item
	$\lsubst{x}{x}{M} = M$ and $\lsubst{y}{x}{M} = 0$,
\item
	$\lsubst{(x[x])}{x}{M + N} = (M + N)[x] + x[M + N] = M[x] + N[x] + x[M] + x[N]$,
\item
	$\lsubst{(x[x^\bang])}{x}{M + N} = (M+N)[x^\bang]+ x[(M+N),x^\bang] = M[x^\bang]+N[x^\bang]+ x[M,x^\bang] + x[N,x^\bang]$,	
\item 
	$\subst{(x[x^\bang])}{x}{M + N} = (M + N)[(M + N)^\bang] = M[M^\bang, N^\bang] + N[M^\bang, N^\bang]$.
\end{enumerate}
\end{example}

As a matter of notation, we will write $\seq L$ for $L_1,\ldots,L_k$ and $\seq N^!$ for $N_1^!,\ldots,N_n^!$. 
We will also abbreviate $\lsubst{\lsubst{M}{x}{L_1}\cdots}{x}{L_k}$ in $\lsubst{M}{x}{\seq L }$.
Moreover, given a sequence $\seq L$ and an index $1\le i\le k$ we will write $\seq L_{-i} $ for $ L_1,\ldots,L_{i-1},L_{i+1},\ldots,L_k$.

\begin{remark}
Every applicative resource $\lambda$-term $MP$ can be written in a unique way as $M[\seq L,\seq N^!]$. 
\end{remark}

\subsection{Resource Lambda Theories}

We now define the equational theories of the resource calculus, 
namely the \emph{resource $\lambda$-theories}.
To begin with, we present the main axiom associated with this calculus:
$$ 
	(\beta^r)\quad (\lambda x.M)[\seq L, \seq N^!] = \subst{\lsubst{M}{x}{\seq L}}{x}{\Sigma_{i=1}^n N_i}
$$
Notice that, when $n=0$, this rule becomes $(\lambda x.M)[\seq L] = \subst{\lsubst{M}{x}{\seq L}}{x}{0}$.
Once oriented from left to right, the $(\beta^r)$-conversion expresses the way of calculating a function 
$\lambda x.M$ applied to a bag containing depletable resources $\seq L$ and perpetual resources $\seq N$.

\begin{remark} The left-to-right oriented version of $(\beta^r)$ corresponds to the equational version of the 
\emph{giant-step} reduction, in the terminology of \cite{PaganiT09}.
In the same paper the authors also consider a \emph{baby-step} reduction rule. 
They prove that both reductions are confluent and that every giant-step can be emulated by several baby-steps. 
For our purposes we can consider the rule $(\beta^r)$ without loss of generality, because both reductions generate the same
equational theory.
\end{remark}

In the resource calculus the axiom equating all resource $\lambda$-terms having 
the same extensional behaviour has the shape:
$$
	(\eta^r)\quad \lambda x.M[x^\bang] = M,\textrm{ where $x\notin\FV(M)$}.
$$
The resource calculus can be seen as a proper extension of the classic $\lambda$-calculus.

\begin{remark} The classic $\lambda$-calculus can be easily injected within the resource calculus.
Indeed, given an ordinary $\lambda$-term $M$, it is sufficient to translate every subterm of $M$ of shape 
$PQ$ into $P[Q^\bang]$. In this restricted system, the rules $(\beta^r)$ and $(\eta^r)$ are completely
equivalent to the classic $(\beta)$ and $(\eta)$-conversions, respectively.
\end{remark}

We now define the equational theories associated with this calculus, namely the \emph{resource $\lambda$-theories}.

A \emph{$\lambda^r$-relation} $\cR$ is any set of equations between sums of resource $\lambda$-terms (resp.\ bags).
Thus $\cR$ can be thought as a binary relation on $\sums{(b)}$. 

A $\lambda^r$-relation $\cR$ is called:
\begin{itemize}
\item an \emph{equivalence} if it is closed under the following rules (for all $\sA,\sB,\sC\in\sums{(b)}$):
$$
\infer[\textrm{reflexivity}]{\sA = \sA}{}\qquad\quad
\infer[\textrm{symmetry}]{\sA = \sB}{\sB = \sA}\qquad\quad
\infer[\textrm{transitivity}]{\sA = \sC}{\sA = \sB& \sB = \sC}
$$
\item \emph{compatible} if it is closed under the following structural rules 
(for all $\sM,\sM_i\in\sums{r},\ \sP\in\sums{b},\ M,M_i\in\Lambda^r$ and  $P\in\Lambda^{b}$):
$$
\begin{array}{ccc}
\infer[\textrm{lambda}]{\lambda x.M = \lambda x.\sM}{M = \sM}&\quad&
\infer[\textrm{app}]{MP = \sM\sP}{M = \sM& P = \sP}\\
~\\
\infer[\textrm{bag}]{[M^{(\bang)}]\mcup P = [\sM^{(\bang)}]\mcup \sP}{M = \sM& P=\sP}&&
\infer[\textrm{sum}]{\sum_{i = 1}^{n} M_i = \sum_{i = 1}^{n} \sM_i}{M_i = \sM_i& \textrm{ for all }1\le i\le n}\\
\end{array}
$$
\end{itemize}
As a matter of notation, we will write $\cR\vdash \sM = \sN$ or $\sM =_{\cR} \sN$ for $\sM= \sN\in\cR$.

\begin{definition} A \emph{resource $\lambda$-theory} is any compatible $\lambda^r$-relation $\cR$ 
which is an equivalence relation and includes $(\beta^r)$. 
$\cR$ is called \emph{extensional} if it also contains $(\eta^r)$.
\end{definition}

We denote by $\lambda\beta^r$ (resp.\ $\lambda\beta\eta^r$) the minimum resource $\lambda$-theory 
(resp.\ the minimum extensional resource $\lambda$-theory).

\begin{example} \
\begin{enumerate}[1.]
\item
	$\lambda\beta^r\vdash (\lambda x.x[x])[\bold{I}] =0$, $\lambda\beta^r\vdash (\lambda x.x[x])[\bold{I},\bold{I}] = \bold{I}$ and 
	$\lambda\beta^r\vdash(\lambda x.x[x])[\bold{I},\bold{I},\bold{I}] = 0$,
\item
	$\lambda\beta^r\vdash  (\lambda x.x[x])[M,N] = M[N] + N[M]$,
\item
	$\lambda\beta^r\vdash (\lambda x.x[x,x])[(\lambda y.y[y^\bang])^\bang] = 
	(\lambda x.x[x^\bang])[\lambda y.y[y^\bang], \lambda z.z[z^\bang]] = 
	2(\lambda y.y[y^\bang])[(\lambda z.z[z^\bang])^\bang]$,
\item
	$\lambda\beta\eta^r\vdash (\lambda xz.y[y][z^\bang])[] = \lambda z.y[y][z^\bang] = y[y]$.
\end{enumerate}
\end{example}

\subsection{From the Resource to the Differential Lambda Calculus\ldots}

In this subsection we show that every linear reflexive object living in a Cartesian closed differential category 
is also a sound model of the untyped resource calculus.
This result is achieved by first translating the resource calculus in the differential $\lambda$-calculus, 
and then applying the machinery of Section~\ref{sec:cat-models}.

\begin{definition}
The resource calculus can be easily translated into the differential $\lambda$-calculus as follows:
\begin{itemize}
\item $x^d = x$,
\item $(\lambda x.M)^d = \lambda x.M^d$,
\item $(M[L_1,\ldots,L_k,N_1^\bang,\ldots,N_n^\bang])^d = (\Dern{k}{M^d}{L_1^d,\ldots,L_k^d})(\Sigma_{i = 1}^{n}N_i^d)$.
\end{itemize}
The translation is then extended to elements in $\sums{r}$ by setting $(\Sigma_{i=1}^n M_i)^d = \Sigma_{i=1}^n M_i^d$.
\end{definition}

The next lemma shows that this translation behaves well with respect to the differential and the usual substitution.

\begin{lemma}\label{lemma:subst} Let $M,N\in\Lambda^r$ and $x$ be a variable. Then:
\begin{itemize}
\item[(i)] $(\lsubst{M}{x}{N})^d = \dsubst{M^d}{x}{N^d}$,
\item[(ii)] $(\subst{M}{x}{N})^d = \subst{M^d}{x}{N^d}$.
\end{itemize}
\end{lemma}

\begin{proof} $(i)$ By structural induction on $M$.
The only difficult case is $M\equiv M'[\seq L, \seq N^!]$. 
By definition of $(-)^d$ and of linear substitution we have:
$$
\begin{array}{l}
(\lsubst{(M'[\seq L, \seq N^!])}{x}{N})^d = (\lsubst{M'}{x}{N}[\seq L, \seq N^!])^d + (M'(\lsubst{[\seq L, \seq N^!]}{x}{N}))^d = \\
\underbrace{(\lsubst{M'}{x}{N}[\seq L, \seq N^!])^d}_{(1)} + 
\underbrace{(\Sigma_{j=1}^k M'[\lsubst{L_j}{x}{N},\seq L_{-j}, \seq N^!])^d}_{(2)} +
\underbrace{(\Sigma_{i=1}^n M'[\lsubst{N_i}{x}{N},\seq L, \seq N^!])^d}_{(3)}.\\
\end{array}
$$
Let us consider the three addenda separately. 

(1) By definition of $(-)^d$ we have that
$(\lsubst{M'}{x}{N}[\seq L, \seq N^!])^d = 
(\Dern{k}{(\lsubst{M'}{x}{N})^d}{\seq L^d})(\Sigma_{i = 1}^n N_i^d)$. 
By applying the induction hypothesis, this is equal to 
$(\Dern{k}{(\dsubst{M'^d}{x}{N^d})}{\seq L^d})(\Sigma_{i = 1}^n N_i^d)$.

(2) By definition of the translation map $(-)^d$ we have that
$(\Sigma_{j=1}^k M'[\lsubst{L_j}{x}{N},\seq L_{-j}, \seq N^!])^d = 
\Sigma_{j=1}^k (\Dern{k-1}{(\Der{M'^d}{(\lsubst{L_j}{x}{N}})^d)}{\seq L_{-j}^d})(\Sigma_{i=1}^n N_i^d)$. 
By applying the induction hypothesis, this is equal to
$\Sigma_{j=1}^k (\Dern{k-1}{(\Der{M'^d}{(\dsubst{L_j^d}{x}{N^d})})}{\seq L_{-j}^d})(\Sigma_{i=1}^n N_i^d)$.

(3) By definition of 
$(-)^d$ we have $(\Sigma_{j=1}^n M'[\lsubst{N_j}{x}{N},\seq L, \seq N^!])^d = 
\Sigma_{j=1}^n (M'[\lsubst{N_j}{x}{N},\seq L, \seq N^!])^d = 
\Sigma_{j=1}^n (\Dern{k}{(\Der{M'^d}{(\lsubst{N_j}{x}{N})^d})}{\seq L^d})(\Sigma_{i=1}^n N_i^d)$. 
By applying the induction hypothesis, this is equal to 
$\Sigma_{j=1}^n (\Dern{k}{(\Der{M'^d}{(\dsubst{N_j^d}{x}{N^d})})}{\seq L^d})(\Sigma_{i=1}^n N_i^d)$.
By permutative equality this is equal to 
$\Sigma_{j = 1}^{n}(\Der{(\Dern{k}{M'^d}{\seq L^d})}{(\dsubst{N_j^d}{x}{N^d})})(\Sigma_{i = 1}^{n} N_i^d)$.

To conclude the proof it is sufficient to verify that $\Dsubst{((\Dern{k}{M'^d}{\seq L^d})(\Sigma_{i = 1}^{n} N_i^d))}{x}{N^d}$ is equal to the sum of 
(1), (2) and (3).

$(ii)$ By straightforward induction on $M$. 
\end{proof}

The translation $(\cdot)^d$ is `faithful' in the sense expressed by the next proposition.

\begin{proposition}\label{prop:res} For all $M\in \Lambda^r$ we have that $\lambda\beta^r\vdash M = N$ implies 
$\lambda\beta^d\vdash M^d = N^d$.
\end{proposition}

\begin{proof} It is easy to check that the proposition holds for the contextual rules.

Suppose then that $\lambda\beta^r\vdash M = N$ because $M\equiv (\lambda x.M')[\seq L,\seq N^!]$ and 
$N\equiv \subst{\lsubst{M'}{x}{\seq L}}{x}{\Sigma_{i=1}^n N_i}$. 
By definition of the map $(-)^d$ we have 
$((\lambda x.M')[\seq L,\seq N^!])^d = 
(\Dern{k}{(\lambda x.M'^d)}{\seq L^d})(\Sigma_{i=1}^n N_i^d) =_{\lambda\beta^d}
(\lambda x.\dsubstn{k}{M'^d}{x,\ldots,x}{\seq L^d})(\Sigma_{i=1}^n N_i^d) =_{\lambda\beta^d}
\subst{(\dsubstn{k}{M'^d}{x,\ldots,x}{\seq L^d})}{x}{\Sigma_{i=1}^n N_i^d}$ 
which is equal to $N^d$ by Lemma~\ref{lemma:subst}.
\end{proof}

\begin{remark}\label{rem:gensums} The two results above generalize easily to sums of 
resource $\lambda$-terms (\emph{i.e.}, to elements $\sM\in\sums{r}$).
\end{remark}

\subsubsection{Interpreting the Resource Calculus by Translation}

Given a linear reflexive object $\cU$ living in a Cartesian closed differential category $\cat C$
it is possible to interpret resource $\lambda$-terms trough their translation $(-)^d$.
Indeed, it is sufficient to set 
$$
	\Int{M}_{\seq x} = \Int{M^d}_{\seq x} : U^n\to U.
$$
From this fact, Proposition~\ref{prop:res} and Remark~\ref{rem:gensums} it follows that 
$\cU$ is a sound model of the untyped resource calculus.

\begin{remark}
If $\cU$ is an extensional model of the differential $\lambda$-calculus,
then it is also an extensional model of the resource calculus.
Indeed $\Int{(\lambda x.M[x^\bang])^d}_{\seq{x}} = \Int{\lambda x.M^dx}_{\seq{x}} =  \Int{M^d}_{\seq{x}}$.
\end{remark}

\subsection{And back\ldots}

In this subsection we define a translation from the differential to the resource calculus.
This translation is more tricky because in the differential $\lambda$-calculus the result of 
the linear application $\Der{(\lambda x.s)}{t}$ mantains the lambda abstraction 
(since it waits for other arguments that may substitute the remaining occurrences of $x$ in $s$),
while the na\"ively corresponding resource $\lambda$-term $(\lambda x.M)[N]$ does erase it 
(since all other free occurrences of $x$ in $M$ are substituted by 0).
\begin{definition} The differential $\lambda$-calculus can be translated into the resource calculus as follows:
$$
	\begin{array}{l}
	x^r = x,\\
	(\lambda x.s)^r = \lambda x.s^r,\\
	(sT)^r = s^r[(T^r)^!],\\
	(\Dern{k}{s}{t_1,\ldots,t_k})^r = \lambda y.s^r[t_1^r,\ldots,t_k^r,y^!], \textrm{ where $y$ is a fresh variable and $k\ge 1$},\\
	(s + S)^r = s^r + S^r.\\
	\end{array}
$$
\end{definition}

Notice that while the shape of the term $\lambda y.s^r[t_1^r,\ldots,t_k^r,y^!]$ looks similar to an $(\eta^r)$-expansion of
$s^r[t_1^r,\ldots,t_k^r]$, it is not
! 
Indeed, in the $(\eta^r)$-rule, $y^\bang$ is supposed to be in a singleton bag.

\begin{lemma}\label{lemma:substrR} Let $S,T\in\Lambda^d$ and $x$ be a variable. Then:
\begin{enumerate}[(i)]
\item $(\dsubst{S}{x}{T})^r = \lsubst{S^r}{x}{T^r}$,
\item $(\subst{S}{x}{T})^r = \subst{S^r}{x}{T^r}$.
\end{enumerate}
\end{lemma}

\begin{proof}
$(i)$ By structural induction on $S$. If $S$ is a variable, a lambda abstraction or a sum,
the lemma follows straight from the induction hypothesis.

\begin{itemize}
\item case $S\equiv \Dern{k}{s}{t_1,\ldots,t_k}$. We have:
$$
\begin{array}{ll}
(\Dsubst{(\Dern{k}{s}{t_1,\ldots,t_k})}{x}{T})^r = \\
\qquad\qquad = \Sigma_{i = 1}^{k} ((\Dern{k}{s}{t_1,\ldots,\dsubst{t_i}{x}{T},\ldots,t_k}))^r &\\
\qquad\qquad\quad +\ ((\Dern{k}{(\dsubst{s}{x}{T})}{t_1,\ldots,t_k}))^r &\textrm{by def.\ of }\dsubst{(\cdot)}{x}{T}\\
\qquad\qquad =\Sigma_{i = 1}^{k} \lambda y.s^r[t_1^r,\ldots,(\dsubst{t_i}{x}{T})^r,\ldots,t_k^r,y^\bang] &\textrm{}\\
\qquad\qquad\quad +\  \lambda y.(\dsubst{s}{x}{T})^r[t_1^r,\ldots,t_k^r,y^\bang]&\textrm{by def.\ of }(\cdot)^r\\
\qquad\qquad =  \Sigma_{i = 1}^{k} \lambda y.s^r[t_1^r,\ldots,\lsubst{t_i^r}{x}{T^r},\ldots,t_k^r,y^\bang] &\textrm{}\\
\qquad\qquad\quad +\ \lambda y.(\lsubst{s^r}{x}{T^r})[t_1^r,\ldots,t_k^r,y^\bang]&\textrm{by induction hypothesis}\\
\qquad\qquad = \lsubst{(\lambda y.s^r[t_1^r,\ldots,t_k^r,y^\bang])}{x}{T^r}&\textrm{by def.\ of }\lsubst{}{x}{T^r}\\
\qquad\qquad =\lsubst{(\Dern{k}{s}{t_1,\ldots,t_k})^r}{x}{T^r}&\textrm{by def.\ of }(\cdot)^r\\
\end{array}
$$

\item case $S\equiv sU$. By definition, we have 
$(\dsubst{(sU)}{x}{T})^r = ((\dsubst{s}{x}{T})U + (\Der{s}{(\dsubst{U}{x}{T})})U)^r = 
((\dsubst{s}{x}{T})U)^r + ((\Der{s}{(\dsubst{U}{x}{T})})U)^r = 
(\dsubst{s}{x}{T})^r[(U^r)^\bang] + 
(\lambda y.s^r[(\dsubst{U}{x}{T})^r,y^\bang])[(U^r)^\bang]$.
By induction hypothesis this is equal to 
$(\lsubst{s^r}{x}{T^r})[(U^r)^\bang] + 
(\lambda y.s^r[\lsubst{U^r}{x}{T^r},y^\bang])[(U^r)^\bang]$.
By $\beta$-conversion this is equal to $(\lsubst{s^r}{x}{T^r})[(U^r)^\bang] + s^r[\lsubst{U^r}{x}{T^r}, (U^r)^\bang]$.
By definition of linear substitution this is 
$\lsubst{(s^r[(U^r)^\bang])}{x}{T^r} = 
\lsubst{(sU)^r}{x}{T^r}$.
\end{itemize}

$(ii)$ By straightforward induction on $S$.
\end{proof}

The next proposition shows that also the translation $(\cdot)^r$ is faithful.

\begin{proposition}\label{prop:diff} 
For all $S,T\in\Lambda^d$ we have that $\lambda\beta^d\vdash S = T$ implies $\lambda\beta^r \vdash S^r = T^r$.
\end{proposition}

\begin{proof}  It is easy to check that the proposition holds for the contextual rules.

Suppose that $\lambda\beta^d\vdash S = T$ holds because $S\equiv \Dern{k}{(\lambda x.s)}{u_1,\ldots,u_k}$ and 
$T\equiv \lambda x.\dsubstn{k}{s}{x,\ldots,x}{u_1,\ldots,u_k}$.
Then we have 
$$
	\begin{array}{lcll}
	S^r &=& \lambda y.(\lambda x.s^r)[u_1^r,\ldots,u_k^r,y^\bang]&\textrm{by def.\ of }(\cdot)^r\\
	&= _{\lambda\beta^r}& \lambda y.\subst{\lsubst{\lsubst{s^r}{x}{u_1^r}\cdots}{x}{u_k^r}}{x}{y}&\textrm{by $\beta^r$-conversion}\\
	&\equiv& \lambda x.\lsubst{\lsubst{s^r}{x}{u_1^r}\cdots}{x}{u_k^r}&\textrm{by $\alpha$-conversion}\\	
	&=& \lambda x.\big(\dsubstn{k}{s}{x,\ldots,x}{u_1,\ldots,u_k}\big)^{r}&\textrm{by Lemma~\ref{lemma:substrR}(i)}\\
	&=& T^r&\textrm{by def.\ of }(\cdot)^r\\
	\end{array}
$$
\end{proof}

The two translations $(\cdot)^d$ and $(\cdot)^r$ are not exactly one the inverse of the other one.
The next proposition summarizes the properties that they do satisfy.

\begin{proposition}\label{prop:niceprop} The translations $(\cdot)^d$ and $(\cdot)^r$ enjoy the following properties:
\begin{enumerate}[(i)]
\item $(s^r)^d \equiv s$, for all usual $\lambda$-terms $s$,
\item $(S^r)^d \not\equiv S$ and $(\sM^d)^r \not\equiv \sM$, for some $S\in\Lambda^d $ and $\sM\in\sums{r}$,
\item $\lambda\beta\eta^d\vdash (S^r)^d = S$, for all $S\in\Lambda^d$,
\item $\lambda\beta^r\vdash (\sM^d)^r = \sM$, for all $ \sM\in\sums{r}$.
\end{enumerate}
\end{proposition}

\begin{proof} $(i)$ By straightforward induction on the structure of $s$.

$(ii)$ For instance $((\Der{x}{x})^r)^d = (\lambda y.x[x,y^\bang])^d = \lambda y.(\Der{x}{x})y\not\equiv\Der{x}{x}$.
On the other hand we have $((x[L])^d)^r = ((\Der{x}{y})0)^r = (\lambda z.x[y,z^\bang])0\not\equiv x[L]$.

$(iii)$ By induction on the structure of $S$.
\begin{itemize}
\item
	case $S \equiv \Dern{k}{s}{t_1,\ldots,t_k}$. By definition of $(\cdot)^r$ we have that $((\Dern{k}{s}{t_1,\ldots,t_k})^r)^d$
	is equal to $(\lambda y.s^r[t^r_1,\cdots,t^r_k,y^\bang])^d = \lambda y.(\Dern{k}{(s^r)^d}{(t^r_1)^d,\cdots,(t^r_k)^d})y$.
	By induction hypothesis we have $(s^r)^d =_{\lambda\beta\eta^d}s$ and $(t^r_i)^d =_{\lambda\beta\eta^d}t_i$ for all $1\le i\le k$,
	thus $\lambda y.(\Dern{k}{(s^r)^d}{(t^r_1)^d,\cdots,(t^r_k)^d})y =_{\lambda\beta\eta^d} 
	\lambda y.(\Dern{k}{s}{t_1,\ldots,t_k})y
	=_{\lambda\beta\eta^d} 
	\Dern{k}{s}{t_1,\ldots,t_k}$.
\item
	case $S\equiv sT$. We have $((sT)^r)^d = (s^r[(T^r)^!])^d = (s^r)^d(T^r)^d$.
	By induction hypothesis, we know that $(s^r)^d =_{\lambda\beta\eta^d} s$ and $(T^r)^d =_{\lambda\beta\eta^d} T$, 
	thus we conclude $(s^r)^d(T^r)^d =_{\lambda\beta\eta^d} sT$.
\item All other cases are trivial.
\end{itemize}

$(iv)$ By induction on the structure of $\sM$.
The only interesting case is $\sM \equiv M[\seq L,\seq N^\bang]$. 
We have $((M[\seq L,\seq N^\bang])^d)^r = ((\Dern{k}{M^d}{\seq L^d})(\Sigma_{i = 1}^{n} N_i^d))^r =
(\lambda y. (M^d)^r[(\seq L^d)^r,y^\bang])[((\seq N^d)^r)^\bang]$.
By induction hypothesis we know that $(M^d)^r =_{\lambda\beta^r} M$, $(L_j^d)^r =_{\lambda\beta^r} L_j$ and $(N_i^d)^r =_{\lambda\beta^r} N_i$, thus 
$(\lambda y. (M^d)^r[(\seq L^d)^r,y^\bang])[((\seq N^d)^r)^\bang] =_{\lambda\beta^r}
(\lambda y. M[\seq L,y^\bang])[(\seq N)^\bang]$.
Since $y\notin\FV(M,\seq L)$ we have that $(\lambda y. M[\seq L,y^\bang])[(\seq N)^\bang] =_{\lambda\beta^r}M[\seq L,\seq N^\bang]$.
\end{proof}

\section{Discussion, Further Works and Related Works}\label{sec:FurtherWorks}

In this paper we proposed a general categorical definition of model of the untyped differential 
$\lambda$-calculus, namely the notion of linear reflexive object living in a Cartesian closed differential
category. 
We have proved that this notion of model is \emph{sound} (i.e., the equational theory induced by a 
model is actually a differential $\lambda$-theory), and inhabited (indeed we gave a concrete 
example of such a definition).

Finally, we have shown that the equational theories of the differential $\lambda$-calculus and 
of the resource calculus are tightly connected.
Formally, we have provided faithful translations between the two calculi, thus showing that they share the same 
notion of model. 
In particular, this shows that linear reflexive objects in Cartesian closed differential categories are also sound models 
of the untyped resource calculus.

\subsection{Other Examples of Cartesian Closed Differential Categories}\label{sec:FurtherWorks:examples}

In Section~\ref{sec:examples} we have presented $\MRel$ (and cited \bold{MFin} in Remark~\ref{rem:Mfin}) 
as an instance of the definition of Cartesian closed differential category.
We briefly discuss here other examples of such categories that have been recently defined in the literature.

In the forthcoming paper \cite{ManzonettoM10} we have described, in collaboration with McCusker, 
a Cartesian closed differential category $\cat{G}$ based on games, i.e.,
having arenas as objects and strategies as morphisms.
In this category strategies are defined as arbitrary sets of complete plays that are fully justified, 
well-bracketed and satisfy suitable visibility conditions. 
As expected, since the differential $\lambda$-calculus is intrinsically non-deterministic,
also the strategies we consider are non-deterministic.
{\em Complete} plays are needed to check easily whether a strategy plays on a certain component {\em exactly once};
intuitively this captures the fact that such a strategy is {\em linear} in that component. 
This category of games, just like $\MRel$, models the Taylor expansion.
Actually, these two categories share many properties as $\cat G$ can be `collapsed' into \MRel{}
in the sense that it is possible to define a time-forgetting lax-functor from $\cat G$ to $\MRel$ 
in the spirit of \cite{BaillotDE98}.

Natural examples of differential Cartesian closed categories that \emph{do not} model the Taylor expansion 
have been recently defined in \cite{CarraroES10a} by introducing new exponential operations on $\cat{Rel}$. 
The intuition behind this construction is rather simple: the authors replace the set of natural numbers 
(that are used for counting multiplicities of elements in multisets) with more general semi-rings containing
elements $\omega$ such that $\omega + 1 = \omega$ (i.e., elements that are morally infinite). 
In these models with infinite multiplicities all differential constructions are available, but the Taylor formula 
does not hold.
Indeed, in these categories it is possible to find a morphism $f\neq 0$ such that, 
for all $n\in\nat$, the $n$-th derivative of $f$ evaluated on 0 is equal to 0: 
the Taylor expansion of such an $f$ is the 0 map, and hence the morphism is different from its Taylor expansion.
In particular, the authors exhibit models where the interpretation of $\Omega$ is different from $0$.

\subsection{Completeness and Incompleteness}\label{sec:FurtherWorks:completeness}
The categorical notion of model of the classic $\lambda$-calculus enjoys a completeness theorem \cite{Scott80}
stating that every $\lambda$-theory $\cT$ can be represented as the theory of a reflexive object 
in a particular Cartesian closed category.
The proof of this theorem is achieved in two steps: 
$(i)$ given a $\lambda$-theory $\cT$ one proves that the set of $\lambda$-terms modulo $\cT$ together with 
the application operator defined between equivalence classes constitutes an applicative structure 
that can be endowed with a structure of {\em $\lambda$-model}\footnote{
A `$\lambda$-model' is a combinatory algebra satisfying the five axioms of Curry and the Meyer-Scott axiom.
We refer to \cite[Ch.~5]{Bare} for more details.
} $\cM_\cT$ (usually called ``the term model of $\cT$'');
$(ii)$ by applying to $\cM_\cT$ a construction called \emph{Karubi envelope} \cite{Koymans82} one builds a 
(very syntactical) Cartesian closed category $\cat{C}_{\cT}$ in which the identity $\bold{I}$ 
is a reflexive object such that $\Th(\bold{I}) = \cT$.

We conjecture that the categorical notion of model of the differential $\lambda$-calculus proposed in this paper
enjoys a similar theorem. 
However, to adapt the original proof to this framework we would need first to understand what is 
a suitable algebraic notion of model of the differential $\lambda$-calculus, in order to built the term model.
Preliminary investigations on this subject have been recently made by Carraro, Ehrhard and Salibra in \cite{CarraroES10b}, 
where the authors provide a notion of ``resource $\lambda$-models'' and show that they can be used to model
the \emph{strictly linear} fragment of the resource calculus (i.e., the fragment without $(\cdot)^\bang$).
At the moment, a generalization allowing to model the full fragment of resource calculus (or, equivalently, 
the differential $\lambda$-calculus) does not seem easy, and is kept for future work. 

We would like to conclude this subsection by noticing that -- although the completeness theorem is 
interesting from a theoretical point of view -- it is not really helpful for the working computer scientist.
Indeed, as noticed above, the term models $\cM_\cT$ and the corresponding categorical models $\bold{I}$ living in $\cat{C}_\cT$
are rather syntactical. 
Thus, proving properties of $\lambda$-terms via these models does not make it any easier than working directly with the syntax.
On the other hand, the non-syntactical semantics of $\lambda$-calculus known in the literature 
(e.g., the continuous semantics \cite{Scott72}, the stable semantics \cite{Berry78}, the strongly stable semantics \cite{BucciarelliE91} 
and the relational semantics \cite{BucciarelliEM07})
are all \emph{hugely incomplete} --- there are $2^{\aleph_0}$ $\lambda$-theories that cannot be represented 
as theories of models living in these semantics.
This follows from a general theorem proved by Salibra in \cite{Salibra01}.
The problem of finding a non-syntactical complete semantics is still open, and very difficult.

\subsection{Working at the Monoidal Level}
Another interesting line of research is to characterize categorical models of the differential $\lambda$-calculus at the level of SMCC's (symmetric monoidal closed categories). 
In \cite{BluteCS09}, Blute \emph{et al.}\ show that (monoidal) differential categories \cite{BluteCS06} give rise to Cartesian differential categories via the co-Kleisli construction.
In the same spirit, we would like to provide sufficient and necessary conditions on SMCC's for giving rise to Cartesian closed differential categories (indeed, all the examples given in Section~\ref{sec:examples} and Subsection~\ref{sec:FurtherWorks:examples} may be generated in this way).

Notice that, in monoidal frameworks, categorical proofs become often awkward due to the symmetric properties of the tensor product $\otimes$ . 
It would be then interesting to define a graphical formalism allowing to represent in a pleasant and intuitive way the morphisms of these categories. 
This formalism could be inspired by differential proofnets or interaction nets \cite{EhrhardR06b}, but should satisfy (at least) the following properties: 
there should be a 1-to-1 correspondence between a morphism and its graphical representation (maybe up to some well chosen equivalence on morphisms); the formalism should not ask for extra properties of the category, like the presence of the operator $\parr$ or the dualizing object $\bot$.

\smallskip {\bf Acknowledgements.} 
We are grateful to Antonio Bucciarelli, Thomas Ehrhard and Guy McCusker.
Many thanks to Michele Pagani and Paolo Tranquilli for helpful comments and suggestions.

\bibliographystyle{plain}
\bibliography{giulio}

\newpage
\appendix
\newpage
\addtolength{\oddsidemargin}{-10pt}
\addtolength{\evensidemargin}{-10pt}

\begin{adjustwidth}{-0.7cm}{-1.2cm}

\section{Technical Appendix}\label{app:proofs}

\small
This technical appendix is devoted to provide the full proofs of the two main lemmas in Subsection~\ref{subsec:DLC}.
These proofs are not particularly difficult, but quite long and require some preliminary notations.
\medskip

\begin{notation} We will adopt the following notations:
\begin{itemize}
\item Given a sequence of indices $\seq i = i_1,\ldots,i_k$ with $i_j \in\{1,2\}$ we write $\Proj{\seq i}$ for $\Proj{i_1}\comp\cdots\comp \Proj{i_k}$.
	Thus $\Proj{1,2} = \Proj{1}\comp\Proj{2}$.
\item For brevity, when writing a Cartesian product of objects as subscript of $0$ or $\Id{}$, we will replace the operator $\times$ by
         simple juxtaposition. 
         For instance, the morphism $\Id{(A\times B)\times(C\times D)}$ will be written $\Id{(AB)(CD)}$.
\end{itemize}
\end{notation}
\medskip

Hereafter ``(proj)'' will refer to the rules $\Proj 1\comp \Pair{f}{g} = f$ and $\Proj 2\comp \Pair{f}{g} = g$ that hold in every Cartesian category.
We recall that $\sw_{ABC} = \Pair{\Pair{\Proj{1,1}}{\Proj 2}}{\Proj{2,1}} : (A\times B)\times C \to (A\times C)\times B$.

\begin{lemma} (Lemma~\ref{lemma:main1}) Let $f:(C\times A)\times D\to B$, $g:C\to A$,  $h:C\to B'$.
\begin{itemize}
\item[](i) $\Proj{2}\star g = g\comp \Proj{1}$,
\item[](ii) $(h\comp\Proj{1})\star g = 0$,
\item[](iii) $\curry(f)\star g = \curry(((f\comp\sw)\star(g\comp\Proj{1}))\comp\sw)$.
\end{itemize}
\end{lemma}

\begin{proof} (i)
$$
\begin{array}{rll}
\Proj{2}\star g = & D(\Proj{2})\comp\Pair{\Pair{0_C}{g\comp\Proj 1}}{\Id{CA}}&\textrm{by def.\ of }\star \\
=&\Proj{2}\comp\Proj{1}\comp\Pair{\Pair{0_C}{g\comp\Proj 1}}{\Id{CA}}&\textrm{by D3}  \\
=&\Proj{2}\comp\Pair{0_C}{g\comp\Proj 1}&\textrm{by (proj)}  \\
=&g\comp \Proj{1}&\textrm{by (proj)} \\
\end{array}
$$
(ii)
$$
\begin{array}{rll}
(h\comp\Proj{1})\star g =&D(h\comp\Proj{1})\comp\Pair{\Pair{0_C}{g\comp\Proj 1}}{\Id{CA}} &\textrm{by def.\ of }\star\\
=&D(h)\comp\Pair{D(\Proj{1})}{\Proj{1,2}}\Pair{\Pair{0_C}{g\comp\Proj 1}}{\Id{CA}} &\textrm{by D5}\\
=&D(h)\comp\Pair{\Proj{1}\comp\Proj{1}}{\Proj{1,2}}\comp \Pair{\Pair{0_C}{g\comp\Proj 1}}{\Id{C\times A}} &\textrm{by D3}\\
=&D(h)\comp\Pair{0_C}{\Proj 1} &\textrm{by (proj)}\\
=&0&\textrm{by D2}\\
\end{array}
$$
(iii) We first prove the following claim.\\



\begin{claim}\label{claim:big} Let $g:C\to A$, then the following diagram commutes:
$$
\xymatrix@R=30pt@C=80pt{
	(C\!\times\!A)\!\times\!D\ar[d]^{\Pair{\Proj 1\times\Id{D}}{\sw}}\ar[r]^{\Pair{\Proj 1}{\Id{C\!\times\!A}}\!\times\!\Id{D}}&(C\!\times\!(C\!\times\!A))\!\times\!D\ar[r]^{(\Pair{0_C}{g}\times\Id{C\!\times\!A})\times\Id{D}}&((C\!\times\!A)\!\times\!(C\!\times\!A))\!\times\!D\ar[d]_{\Pair{\Proj 1\times 0_D}{\Proj 2\times\Id{D}}}\\
	(C\!\times\!D)\!\times\!((C\!\times\!D)\!\times\!A)\ar[r]^{\hspace{-5mm}\Pair{0_{C\!\times\!D}}{g\circ\Proj{1}}\!\times\!\Id{(C\!\times\!D)\!\times\!A}}&((C\!\times\!D)\!\times\!A)\!\times\!((C\!\times\!D)\!\times\!A)\ar[r]^{\Pair{D(\sw)}{\sw\circ\Proj 2}}&((C\!\times\!A)\!\times\!D)\!\times\!((C\!\times\!A))\!\times\! D)
}
$$
\end{claim}

\begin{subproof}
$$
    \begin{array}{ll}
    \Pair{\Proj 1\times 0_D}{\Proj 2\times\Id{D}}\comp((\Pair{0_C}{g}\times\Id{CA})\times\Id{D})\comp(\Pair{\Proj 1}{\Id{CA}}\times\Id{D})=&\textrm{}\\
    \Pair{\Pair{\Pair{0_C}{g\comp\Proj{1,1}}}{0_D}}{\Pair{\Proj{2,1}}{\Proj 2}}\comp\Pair{\Pair{\Proj 1}{\Pair{\Proj 1}{\Proj 2}}\comp\Proj 1}{\Proj 2}=&\textrm{}\\
    \Pair{\Pair{\Pair{0_C}{g\comp\Proj{1,1}}}{0_D}}{\Pair{\Proj{2,1}}{\Proj 2}}\comp\Pair{\Pair{\Proj{1,1}}{\Pair{\Proj{1,1}}{\Proj{2,1}}}}{\Proj 2}=&\textrm{}\\
    \Pair{\Pair{\Pair{0_C}{g\comp\Proj{1,1}}}{0_D}}{\Pair{\Pair{\Proj{1,1}}{\Proj{2,1}}}{\Proj 2}}=&\textrm{}\\
    \Pair{\Pair{\Pair{0_C}{g\comp\Proj{1,1}}}{0_D}}{\Pair{\Pair{\Proj{1,1,2}}{\Proj{2,2}}}{\Proj{2,1,2}}}\comp\Pair{\Proj 1\!\times\!\Id{D}}{\sw} =&\textrm{}\\
    \Pair{\Pair{\Pair{\Proj{1,1,1}}{\Proj{2,1}}}{\Proj{2,1,1}}}{\Pair{\Pair{\Proj{1,1,2}}{\Proj{2,2}}}{\Proj{2,1,2}}}\comp\Pair{\Pair{0_{CD}}{g\comp\Proj{1,1}}}{\Proj 2}\comp\Pair{\Proj 1\times\Id{D}}{\sw} =&\textrm{}\\
    \Pair{D(\sw)}{\sw\comp\Proj 2}\comp(\Pair{0_{CD}}{g\comp\Proj{1}}\times\Id{(CD)A})\comp\Pair{\Proj 1\times\Id{D}}{\sw}\\
    \end{array}
$$
\end{subproof}

We can now conclude the proof as follows:
$$
\begin{array}{rll}
\curry(f)\star g = & D(\curry(f))\comp\Pair{\Pair{0_C}{g\comp\Proj 1}}{\Id{CA}} &\textrm{by def.\ of }\star\\
=&\curry(D(f)\comp\Pair{\Proj 1\times 0_D}{\Proj 2\times\Id{D}})\comp\Pair{\Pair{0_C}{g\comp\Proj 1}}{\Id{CA}}  &\textrm{by (D-curry)}\\
=&\curry(D(f)\comp\Pair{\Proj 1\times 0_D}{\Proj 2\times\Id{D}}\comp((\Pair{\Pair{0_C}{g\comp\Proj 1}}{\Id{CA}})\times\Id{D}))  &\textrm{by (Curry)}\\
=&\curry(D(f)\comp\Pair{D(\sw)}{\sw\comp\Proj 2}\comp(\Pair{0_{CD}}{g\comp\Proj{1}}\times\Id{(CD)A})\comp\Pair{\Proj 1\times\Id{D}}{\sw})\quad  &\textrm{by Claim~\ref{claim:big}}\\
=&\curry(D(f\comp\sw)\comp(\Pair{0_{CD}}{g\comp\Proj{1}}\times\Id{(CD)A})\comp\Pair{\Proj 1}{\Id{}}\comp\sw)  &\textrm{by D5}\\
=&\curry(((f\comp\sw)\star(g\comp\Proj{1}))\comp\sw)&\textrm{by def.\ of }\star\\
\end{array}
$$
\end{proof}

\begin{lemma} (Lemma~\ref{lemma:main2}) Let $f: C\times A\to [\Funint{D}{B}]$, $g:C\to A$, $h:C\times A\to D$ 
\begin{itemize}
\item[](i) $(\eval\comp\Pair{f}{h})\star g = \eval\comp\Pair{f\star g + \curry(\curry^-(f)\star(h\star g))}{h}$
\item[](ii)  $\curry(\curry^-(f)\star h) \star g = \curry(\curry^-(f\star g)\star h) + \curry(\curry^-(f)\star(h\star g))$
\item[](iii) $\curry(\curry^-(f)\star h)\comp\Pair{\Id{C}}{g}  = \curry(\curry^-(f\comp\Pair{\Id{C}}{g})\star (h\comp \Pair{\Id{C}}{g}))$
\end{itemize}
\end{lemma}

\begin{proof}
\begin{itemize}
\item [$(i)$] Let us set $\phi \equiv \Pair{\Pair{0_C}{g\comp\Proj 1}}{\Id{CA}}$. Then we have:
$$
\begin{array}{ll}
(\eval\comp\Pair{f}{h})\star g =&\textrm{by def.\ of }\star\\
D(\eval\comp\Pair{f}{h})\comp\phi =&\textrm{by (D-eval)}\\
(\eval\comp\Pair{D(f)}{h\comp\Proj 2} + D(\curry^-(f)) \comp\Pair{\Pair{0_{CA}}{D(h)}}{\Pair{\Proj 2}{h\comp\Proj 2}})\comp\phi =& \textrm{by Def.~\ref{def:cccLA}}\\
\eval\comp\Pair{D(f)}{h\comp\Proj 2}\comp\phi
 + D(\curry^-(f))\comp\Pair{\Pair{0_{CA}}{D(h)\comp\phi}}{\Pair{\Id{CA}}{h}} =&\textrm{by def.\ of }\star\\
\eval\comp\Pair{D(f)\comp\phi}{h} + D(\curry^-(f))\comp\Pair{\Pair{0_{CA}}{(h\star g)\comp\Proj 1}}{\Id{(CA)D}}\comp \Pair{\Id{CA}}{h} =\quad&\textrm{by def.\ of }\star\\
\eval\comp\Pair{f\star g}{h}+ (\curry^-(f)\star(h\star g))\comp \Pair{\Id{}}{h} =&\textrm{by (beta-cat)}\\
\eval\comp\Pair{f\star g}{h}+ \eval\comp\Pair{\curry(\curry^-(f)\star(h\star g))}{h} =&\textrm{by Lemma~\ref{lemma:evalplus}}\\
\eval\comp\Pair{f\star g + \curry(\curry^-(f)\star(h\star g))}{h}&\\
\end{array}
$$
\item [$(ii)$] We first simplify the equation $\curry(\curry^-(f)\star h) \star g = \curry(\curry^-(f\star g)\star h) + \curry(\curry^-(f)\star(h\star g))$ to get rid
of the Cartesian closed structure. The right side can be rewritten as $\curry((\curry^-(f\star g)\star h) + \curry^-(f)\star(h\star g))$.
By taking a morphism $f':(C\times A)\times D\to B$ such that $f = \curry(f')$ and by applying Lemma~\ref{lemma:main1}$(iii)$ we discover that it is equivalent to show that:
$$
	((f'\star h)\comp\sw)\star(g\comp\Proj 1)\comp\sw = (((f'\comp\sw)\star(g\comp\Proj{1}))\comp\sw)\star h + f'\star (h\star g).
$$

By definition of $\star$ we have:
$$
((f'\star h)\comp\sw)\star(g\comp\Proj 1)\comp\sw =D(D(f')\comp\Pair{\Pair{0_{CA}}{h\comp\Pair{\Proj{1,1}}{\Proj 2}}}{\sw})\comp\Pair{\Pair{0_{CD}}{g\comp\Proj{1,1}}}{\sw}
$$
Let us call now $\phi \equiv \Pair{\Pair{0_{CD}}{g\comp\Proj{1,1}}}{\sw}$ and write $D^2(f)$ for $D(D(f))$.
Then we have:\\
$$
\begin{array}{ll}
D^2(f')\comp\Pair{\Pair{0_{CA}}{h\comp\Pair{\Proj{1,1}}{\Proj 2}}}{\sw})\comp\phi =&\textrm{by D5}\\
D^2(f')\comp\Pair{D(\Pair{\Pair{0_{CA}}{h\comp\Pair{\Proj{1,1}}{\Proj 2}}}{\sw})}{\Pair{\Pair{0_{CA}}{h\comp\Pair{\Proj{1,1}}{\Proj 2}}}{\sw}\comp\Proj 2}\comp\phi =&\textrm{by (pair)}\\
D^2(f')\comp\Pair{D(\Pair{\Pair{0_{CA}}{h\comp\Pair{\Proj{1,1}}{\Proj 2}}}{\sw})\comp\phi}{\Pair{\Pair{0_{CA}}{h\comp\Pair{\Proj{1,1}}{\Proj 2}}}{\sw}\comp\Proj 2\comp\phi} =&\textrm{by D4}\\
D^2(f')\comp\Pair{\Pair{D(\Pair{0_{CA}}{h\comp\Pair{\Proj{1,1}}{\Proj 2}})\comp\phi}{D(\sw)\comp\phi}}{\Pair{\Pair{0_{CA}}{h\comp\Pair{\Proj{1,1}}{\Proj 2}}}{\sw}\comp\sw} =&\textrm{by Rem.~\ref{rem:swap}}\\
D^2(f')\comp\Pair{\Pair{D(\Pair{0_{CA}}{h\comp\Pair{\Proj{1,1}}{\Proj 2}})\comp\phi}{D(\sw)\comp\phi}}{\Pair{\Pair{0_{CA}}{h\comp\Proj 1}}{\Id{(CA)D}}}\\
\end{array}
$$

Since $\Pair{D(\Pair{0_{CA}}{h\comp\Pair{\Proj{1,1}}{\Proj 2}})\comp\phi}{D(\sw)\comp\phi} = \Pair{0}{D(\sw)\comp\phi}+\Pair{D(\Pair{0_{CA}}{h\comp\Pair{\Proj{1,1}}{\Proj 2}})\comp\phi}{0}$ we can apply D2 and rewrite the expression above as a sum of two morphisms:\\
$$
\begin{array}{lc}
(1)&D^2(f')\comp\Pair{\Pair{0_{(CA)D}}{D(\sw)\comp\phi}}{\Pair{\Pair{0_{CA}}{h\comp\Proj 1}}{\Id{(CA)D}}} \ +\\
(2)\quad&D^2(f')\comp\Pair{\Pair{D(\Pair{0_{CA}}{h\comp\Pair{\Proj{1,1}}{\Proj 2}})\comp\phi}{0_{(CA)D}}}{\Pair{\Pair{0_{CA}}{h\comp\Proj 1}}{\Id{(CA)D}}}\\
\end{array}
$$
We now show that  $(1) = (((f'\comp\sw)\star(g\comp\Proj{1}))\comp\sw)\star h$. Indeed, we have:\\
$
\begin{array}{ll}
D^2(f')\comp\Pair{\Pair{0_{(CA)D}}{D(\sw)\comp\phi}}{\Pair{\Pair{0_{CA}}{h\comp\Proj 1}}{\Id{(CA)D}}}=&\textrm{by Rem.~\ref{rem:swap}}\\
D^2(f')\comp\Pair{\Pair{0_{(CA)D}}{\sw\comp\Proj 1\comp\phi}}{\Pair{\Pair{0_{CA}}{h\comp\Proj 1}}{\Id{(CA)D}}}=&\textrm{by (proj)}\\
D^2(f')\comp\Pair{\Pair{0_{(CA)D}}{\sw\comp\Pair{0_{CD}}{g\comp\Proj{1,1}}}}{\Pair{\Pair{0_{CA}}{h\comp\Proj 1}}{\Id{(CA)D}}}=&\textrm{by Rem.~\ref{rem:swap}}\\
D^2(f')\comp\Pair{\Pair{0_{(CA)D}}{\Pair{\Pair{0_{C}}{g\comp\Proj{1,1}}}{0_{D}}}}{\Pair{\Pair{0_{CA}}{h\comp\Proj 1}}{\Id{(CA)D}}}=&\textrm{by D7}\\
D^2(f')\comp\Pair{\Pair{\Pair{\Pair{0_C}{0_A}}{0_D}}{\Pair{0_{CA}}{h\comp\Proj 1}}}{\Pair{\Pair{\Pair{0_{C}}{g\comp\Proj{1,1}}}{0_{D}}}{\Id{(CA)D}}} =&\textrm{by D2}\\
D^2(f')\comp\Pair{\Pair{\Pair{\Pair{0_{C}}{D(g)\comp\Pair{0_C}{\Proj{1,1}}}}{0_{D}}}{\Pair{0_{CA}}{h\comp\Proj 1}}}{\Pair{\Pair{\Pair{0_{C}}{g\comp\Proj{1,1}}}{0_{D}}}{\Id{(CA)D}}}.\\
\end{array}
$

Let us set  $\psi \equiv \Pair{\Pair{0_{CA}}{h\comp\Proj 1}}{\Id{(CA)D}}$. Then we have:\\
$
\begin{array}{ll}
D(D(f))\comp\Pair{\Pair{\Pair{\Pair{0_{C}}{D(g)\comp\Pair{0_C}{\Proj{1,1}}}}{0_{D}}}{\Pair{0_{CA}}{h\comp\Proj 1}}}{\Pair{\Pair{\Pair{0_{C}}{g\comp\Proj{1,1}}}{0_{D}}}{\Id{(CA)D}}}=&\textrm{by (proj)}\\
D(D(f))\comp\Pair{\Pair{\Pair{\Pair{0_{C}}{D(g)\comp\Pair{\Proj{1,1,1}}{\Proj{1,1,2}}}}{0_{D}}}{\Proj 1}}{\Pair{\Pair{\Pair{0_{C}}{g\comp\Proj{1,1,2}}}{0_{D}}}{\Proj 2}}\comp \psi =&\textrm{by D3}\\
D(D(f))\comp\Pair{\Pair{\Pair{\Pair{0_{C}}{D(g)\comp\Pair{D(\Proj{1,1})}{\Proj{1,1,2}}}}{0_{D}}}{\Proj 1}}{\Pair{\Pair{\Pair{0_{C}}{g\comp\Proj{1,1,2}}}{0_{D}}}{\Proj 2}}\comp\psi =&\textrm{by D5}\\
D(D(f))\comp\Pair{\Pair{\Pair{\Pair{0_{C}}{D(g\comp\Proj{1,1})}}{0_{D}}}{\Proj 1}}{\Pair{\Pair{\Pair{0_{C}}{g\comp\Proj{1,1,2}}}{0_{D}}}{\Proj 2}}\comp\psi =&\textrm{by D1}\\
D(D(f))\comp\Pair{\Pair{\Pair{\Pair{D(0_{C})}{D(g\comp\Proj{1,1})}}{D(0_{D})}}{D(\Id{(CA)D})}}{\Pair{\Pair{\Pair{0_{C}}{g\comp\Proj{1,1,2}}}{0_{D}}}{\Proj 2}}\comp\psi =&\textrm{by D4 }\\
D(D(f))\comp\Pair{D(\Pair{\Pair{\Pair{0_{C}}{g\comp\Proj{1,1}}}{0_{D}}}{\Id{(CA)D}})}{\Pair{\Pair{\Pair{0_{C}}{g\comp\Proj{1,1}}}{0_{D}}}{\Id{(CA)D}}\comp\Proj 2}\comp \psi =&\textrm{by D5 }\\
D(D(f)\comp\Pair{\Pair{\Pair{0_{C}}{g\comp\Proj{1,1}}}{0_{D}}}{\Id{(CA)D}})\comp\psi =&\textrm{by Rem.~\ref{rem:swap}}\\
D(D(f)\comp\Pair{\sw\comp\Pair{0_{CD}}{g\comp\Proj{1,1}}}{\sw\comp\sw})\comp\psi =&\textrm{by (proj)}\\
D(D(f)\comp\Pair{\sw\comp\Proj 1}{\sw\comp\Proj 2}\comp\Pair{\Pair{0_{CD}}{g\comp\Proj{1,1}}}{\sw})\comp\psi =&\textrm{by Rem.~\ref{rem:swap}}\\
D(D(f)\comp\Pair{D(\sw)}{\sw\comp\Proj 2}\comp\Pair{\Pair{0_{CD}}{g\comp\Proj{1,1}}}{\sw})\comp\psi =&\textrm{by D5}\\
D(D(f\comp\sw)\comp\Pair{\Pair{0_{CD}}{g\comp\Proj{1,1}}}{\Id{(CD)A}}\comp\sw)\comp \Pair{\Pair{0_{CA}}{h\comp\Proj 1}}{\Id{(CA)D}} =&\textrm{by def.\ of }\star\\
(((f\comp\sw)\star(g\comp\Proj{1}))\comp\sw)\star h \\
\end{array}
$

\medskip
We will now show that $(2) = f\star (h\star g)$, and this will conclude the proof.\\
$$
\begin{array}{ll}
D^2(f)\comp\Pair{\Pair{D(\Pair{0_{CA}}{h\comp\Pair{\Proj{1,1}}{\Proj 2}})\comp\phi}{0_{(CA)D}}}{\Pair{\Pair{0_{CA}}{h\comp\Proj 1}}{\Id{(CA)D}}}=&\textrm{by D1+4}\\
D^2(f)\comp\Pair{\Pair{\Pair{0_{CA}}{D(h\comp\Pair{\Proj{1,1}}{\Proj 2})}\comp\phi}{0_{(CA)D}}}{\Pair{\Pair{0_{CA}}{h\comp\Proj 1}}{\Id{(CA)D}}}=&\textrm{by D5}\\
D^2(f)\comp\Pair{\Pair{\Pair{0_{CA}}{D(h)\comp\Pair{D(\Pair{\Proj{1,1}}{\Proj 2})}{\Pair{\Proj{1,1,2}}{\Proj{2,2}}}}\comp\phi}{0_{(CA)D}}}{\Pair{\Pair{0_{CA}}{h\comp\Proj 1}}{\Id{(CA)D}}}=&\textrm{by D4+D3}\\
D^2(f)\comp\Pair{\Pair{\Pair{0_{CA}}{D(h)\comp\Pair{\Pair{D(\Proj{1,1})}{D(\Proj 2)}}{\Pair{\Proj{1,1,2}}{\Proj{2,2}}}}\comp\phi}{0_{(CA)D}}}{\Pair{\Pair{0_{CA}}{h\comp\Proj 1}}{\Id{(CA)D}}}=&\textrm{by D5+D3}\\
D^2(f)\comp\Pair{\Pair{\Pair{0_{CA}}{D(h)\comp\Pair{\Pair{\Proj{1,1,1}}{\Proj{2,1}}}{\Pair{\Proj{1,1,2}}{\Proj{2,2}}}}\comp\phi}{0_{(CA)D}}}{\Pair{\Pair{0_{CA}}{h\comp\Proj 1}}{\Id{(CA)D}}}=&\textrm{by (proj)}\\
D^2(f)\comp\Pair{\Pair{\Pair{0_{CA}}{D(h)\comp\Pair{\Pair{0_C}{g\comp\Proj{1,1}}}{\Proj 1}}}{0_{(CA)D}}}{\Pair{\Pair{0_{CA}}{h\comp\Proj 1}}{\Id{(CA)D}}}=&\textrm{by D6}\\
D(f)\comp\Pair{\Pair{0_{CA}}{D(h)\comp\Pair{\Pair{0_{C}}{g\comp\Proj{1,1}}}{\Proj 1}}}{\Id{(CA)D}} =&\textrm{by (proj)}\\
D(f)\comp\Pair{\Pair{0_{CA}}{D(h)\comp\Pair{\Pair{0_{C}}{g\comp\Proj 1}}{\Id{CA}}\comp\Proj 1}}{\Id{(CA)D}} =&\textrm{by def.\ of }\star\\
f\star (h\star g)&\\
\end{array}
$$
\item [$(iii)$] By (Curry) we have $\curry(\curry^-(f)\star h)\comp\Pair{\Id{C}}{g}  = \curry((\curry^-(f)\star h)\comp(\Pair{\Id{C}}{g}\times\Id{D}))$,
thus if we show that $(\curry^-(f)\star h)\comp(\Pair{\Id{C}}{g}\times\Id{D}) = \curry^-(f\comp\Pair{\Id{C}}{g})\star (h\comp \Pair{\Id{C}}{g})$
we have finished.\\

We proceed then as follows:
$$
\begin{array}{ll}
(\curry^-(f)\star h)\comp(\Pair{\Id{C}}{g}\times\Id{D})  =&\textrm{by def.\ of }\star\\
D(\curry^-(f))\comp\Pair{\Pair{0_{CA}}{h\comp\Proj 1}}{\Id{(CA)D}}\comp(\Pair{\Id{C}}{g}\times\Id{D}) =&\textrm{by def.\ of }\curry^-\\
D(\eval\comp \Pair{f\comp\Proj 1}{\Proj 2})\comp\Pair{\Pair{0_{CA}}{h\comp\Proj 1}}{\Id{(CA)D}}\comp(\Pair{\Id{C}}{g}\times\Id{D}) =&\textrm{by D5+D4}\\
D(\eval)\comp\Pair{\Pair{D(f\comp\Proj 1)}{D(\Proj 2)}}{\Pair{f\comp\Proj{1,2}}{\Proj{2,2}}}\comp \Pair{\Pair{0_{CA}}{h}\comp\Pair{\Proj 1}{g\comp\Proj 1}}{\Pair{\Id{C}}{g}\times\Id{D}} =&\textrm{by D5+D3}\\
D(\eval)\comp\Pair{\Pair{D(f)\comp\Pair{\Proj{1,1}}{\Proj{1,2}}}{\Proj{2,1}}}{\Pair{f\comp\Proj{1,2}}{\Proj{2,2}}}\comp \Pair{\Pair{0_{CA}}{h\comp\Pair{\Proj 1}{g\comp\Proj 1}}}{\Pair{\Id{C}}{g}\times\Id{D}} =&\textrm{by (proj)}\\
D(\eval)\comp\Pair{\Pair{D(f)\comp\Pair{0_{CA}}{\Pair{\Proj 1}{g\comp\Proj 1}}}{h\comp\Pair{\Proj 1}{g\comp\Proj 1}}}{\Pair{f\comp\Pair{\Proj 1}{g\comp\Proj 1}}{\Proj 2}} =&\textrm{by D2}\\
D(\eval)\comp\Pair{\Pair{D(f)\comp\Pair{\Pair{0_C}{D(g)\comp\Pair{0_C}{\Id{C}}}}{\Pair{\Id{C}}{g}}}{h\comp\Pair{\Proj 1}{g\comp\Proj 1}}}{\Pair{f\comp\Pair{\Id{C}}{g}}{\Id{D}}} = \\
\textrm{by setting } \phi = \Pair{\Pair{0_C}{h\comp\Pair{\Proj 1}{g\comp\Proj 1}}}{\Id{CD}}\\
D(\eval)\comp\Pair{\Pair{D(f)\comp\Pair{\Pair{\Proj{1,1}}{D(g)\comp\Pair{\Proj{1,1}}{\Proj{1,2}}}}{\Pair{\Proj{1,2}}{g\comp\Proj{1,2}}}}{\Proj{2,1}}}{\Pair{f\comp\Pair{\Proj{1,2}}{g\comp\Proj{1,2}}}{\Proj{2,2}}}\comp\phi =& \textrm{by D5} \\
D(\eval)\comp\Pair{\Pair{D(f\comp\Pair{\Proj1}{g\comp\Proj 1})}{D(\Proj2)}}{\Pair{f\comp\Pair{\Proj{1,2}}{g\comp\Proj{1,2}}}{\Proj{2,2}}}\comp\phi =& \textrm{by D4} \\
D(\eval)\comp\Pair{D(\Pair{f\comp\Pair{\Proj1}{g\comp\Proj 1}}{\Proj2})}{\Pair{f\comp\Pair{\Proj{1,2}}{g\comp\Proj{1,2}}}{\Proj{2,2}}}\comp\phi =& \textrm{by D5}\\
D(\eval\comp\Pair{f\comp\Pair{\Proj 1}{g\comp\Proj1}}{\Proj 2})\comp\phi =&\textrm{by def.\ of }\curry^- \\
D(\curry^-(f\comp\Pair{\Id{C}}{g}))\comp\phi =&\textrm{by def.\ of }\star \\
\curry^-(f\comp\Pair{\Id{C}}{g})\star (h\comp \Pair{\Id{C}}{g})&\\
\end{array}
$$
\end{itemize}
\end{proof}

\end{adjustwidth}

\end{document}